\documentclass[11pt]{article}

\usepackage{setspace}
\usepackage{enumerate}
\usepackage{enumitem}
\usepackage{algorithm2e}
\usepackage{amsmath}
\usepackage{float}
\usepackage{bbm}
\usepackage{comment}
\usepackage{mathtools}
\newcommand{\ifims}[2]{#1} 
\newcommand{\ifAMS}[3]{#2}   
\newcommand{\ifau}[3]{#3}  
\newcommand{\ifbook}[2]{#1}   

\usepackage{subcaption}
\usepackage{multirow,dcolumn,booktabs,caption,threeparttable,longtable,tabularx}

\newcommand{\emojirocket}{\raisebox{-2pt}{\includegraphics{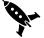}}}
\newcommand{\emojipoop}{\raisebox{-2pt}{\includegraphics{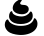}}}

\ifbook{

  }
  {

  }

\def\thetitle{SONIC: SOcial Network analysis with Influencers and Communities}
\def\theruntitle{SONIC}

\newcommand{\rev}[1]{#1}

\def\theabstract{
The integration of social media characteristics into an econometric framework requires modeling a high dimensional dynamic network with dimensions of parameter typically much larger than the number of observations. To cope with this problem, we introduce {SONIC, a new \rev{high-dimensional {network}} model that} assumes that (1) {only} few influencers drive the network dynamics; (2) the community structure of the network is characterized {by} homogeneity of response to {specific} influencers, implying their underlying similarity. An estimation procedure is proposed based on a greedy algorithm and LASSO regularization. Through theoretical study and simulations, we show that the matrix parameter can be estimated even when \rev{sample size} is smaller than the size of the network. Using a novel dataset retrieved from one of leading social media platforms --- StockTwits and quantifying their opinions via natural language processing, we model the opinions network dynamics among a select group of users and further detect the latent communities. With a sparsity regularization, we can identify important nodes in the network.}
\def\kwdp{C1, C22, C51, G41}
\def\kwds{65F05}

\def\thekeywords{social media, network, community, influencers, sentiment}

\def\thanksc{The work was done while this author was a postgraduate student at Humboldt-Universit\"at zu Berlin. Financial support from the German Research Foundation (DFG) via the International Research Training Group 1792 ``High Dimensional Nonstationary Time Series'' in Humboldt-Universit\"at zu Berlin and  is gratefully acknowledged.}
\def\thanksa{Adam Smith Business School, University of Glasgow, UK and Humboldt-Universit\"at zu Berlin in Germany,  corresponding author}
\def\thanksb{BRC Blockchain Research Center, Humboldt-Universit\"{a}t zu Berlin, Germany; Sim Kee Boon Institute, Singapore Management University, Singapore; WISE Wang Yanan Institute for Studies in Economics, Xiamen Uni- versity, Xiamen, China; Dept. Information Science and Finance, National Chiao Tung University, Hsinchu, Taiwan, ROC; Dept. Mathematics and Physics, Charles University, Prague, Czech Republic, Grants–DFG IRTG 1792, CAS: XDA 23020303, and COST Action CA19130 gratefully acknowledged.}
\def\authorc{Yegor Klochkov}
\def\authora{Cathy Yi-Hsuan Chen}
\def\authorb{Wolfgang Karl H\"ardle}
\def\runauthorc{Y. Klochkov}
\def\runauthora{C. Y.-H. Chen}
\def\runauthorb{W.K. H\"ardle}
\def\addressc{
  Cambridge-INET, Faculty of Economics, University of Cambridge
}
\def\emailc{yk376@cam.ac.uk}
\def\addressa{
	University of Glasgow
}
\def\emaila{CathyYi-Hsuan.Chen@glasgow.ac.uk}
\def\addressb{
	Humboldt-Universit\"at zu Berlin
}
\def\emailb{haerdle@hu-berlin.de}

\usepackage{amsmath,amssymb,amsthm}
\usepackage{natbib}
\usepackage{epsfig}
\usepackage{comment}
\usepackage{color}
\usepackage{srcltx}
\usepackage[mathscr]{eucal}
\usepackage{etoolbox}
\usepackage{hyperref}

\definecolor{myblue}{rgb}{0.2, 0.4, 0.8}
\definecolor{myred}{rgb}{1.0, 0.0, 0.2}
\definecolor{darkspringgreen}{rgb}{0.09, 0.45, 0.27}
\definecolor{myorange}{rgb}{1.0, 0.4, 0.2}
\definecolor{mypurple}{rgb}{0.6, 0.0, 0.6}

\usepackage{nameref}

\ifims{
	\textheight=23cm
	\textwidth=14.8cm
	\topmargin=0pt
	\oddsidemargin=0.5cm
	\evensidemargin=0.5cm
	\linespread{1.0}
	\renewenvironment{abstract}
	{\centerline{\textbf{Abstract}}\bigskip
		\begin{center}
			\begin{minipage}{11cm}
				\begin{small}
				
				}
				{   \end{small}
			\end{minipage}
		\end{center}
		\bigskip
	}

}{ 
}

\numberwithin{equation}{section}
\numberwithin{figure}{section}
\newcounter{example}[section]
\numberwithin{example}{section}
\newcounter{remark}[section]
\numberwithin{remark}{section}
\newtheorem{theorem}{Theorem}[section]

\newtheorem{proposition}[theorem]{Proposition}
\newtheorem{lemma}[theorem]{Lemma}
\newtheorem{corollary}[theorem]{Corollary}
\newtheorem{definition}[theorem]{Definition}
\newtheorem{exmp}[example]{Example}
\newtheorem{rmrk}[remark]{Remark}
\newenvironment{example}{\begin{exmp}\rm}{\end{exmp}}
\newenvironment{remark}{\begin{rmrk}\rm}{\end{rmrk}}
\newtheorem{assumption}{Assumption}

\def\rr{\tilde{\mathbf{r}}}

\begin{document}
\thispagestyle{empty}
\ifims{
	\title{\thetitle}
	\ifau{ 
		\author{
			\authora
			\ifdef{\thanksa}{\thanks{\thanksa}}{}
			\\[5.pt]
			\addressa \\
			\texttt{ \emaila}
		}
	}
	{  
		\author{
			\authora
			\ifdef{\thanksa}{\thanks{\thanksa}}{}
			\\[5.pt]
			\addressa \\
			\texttt{ \emaila}
			\and
			\authorb
			\ifdef{\thanksb}{\thanks{\thanksb}}{}
			\\[5.pt]
			\addressb \\
			\texttt{ \emailb}
		}
	}
	{   
		\author{
			\authora
			\ifdef{\thanksa}{\thanks{\thanksa}}{}
			\\[5.pt]
			\addressa \\
			\texttt{ \emaila}
			\and
			\authorb
			\ifdef{\thanksb}{\thanks{\thanksb}}{}
			\\[5.pt]
			\addressb \\
			\texttt{ \emailb}
			\and
			\authorc
			\ifdef{\thanksc}{\thanks{\thanksc}}{}
			\\[5.pt]
			\addressc \\
			\texttt{ \emailc}
		}
	}
	
	\maketitle
	\pagestyle{myheadings}
	\markboth
	{\hfill \textsc{ \small \theruntitle} \hfill}
	{\hfill
		\textsc{ \small
			\ifau{\runauthora}
			{\runauthora and \runauthorb}
			{\runauthora, \runauthorb, and \runauthorc}
		}
		\hfill}
	\begin{abstract}
		\theabstract
	\end{abstract}
	
	\ifAMS
	{\par\noindent\emph{AMS 2010 Subject Classification:} Primary \kwdp. Secondary \kwds}
	{\par\noindent\emph{JEL codes}: \kwdp}
	{}
	
} 
{ 
	\begin{frontmatter}
		\title{\thetitle}

		
		\runtitle{\theruntitle}
		
		\ifau{ 
			\begin{aug}
				\author{\authora\ead[label=e1]{\emaila}}
				\address{\addressa \\
					\printead{e1}}
			\end{aug}
			
			\runauthor{\runauthora}
			\affiliation{\affiliationa} }
		{ 
			\begin{aug}
				\author{\authora\ead[label=e1]{\emaila}\thanksref{t21}}
				\and
				\author{\authorb\ead[label=e2]{\emailb}\thanksref{t22}}
				
				\address{\addressa \\
					\printead{e1}}
				\address{\addressb \\
					\printead{e2}}
				\thankstext{t21}{\thanksa}
				\thankstext{t22}{\thanksb}
				\affiliation{\affiliationa, \affiliationb} 
				\runauthor{\runauthora and \runauthorb}
			\end{aug}
		} 
		{ 
			\begin{aug}
				\author{\authora\ead[label=e1]{\emaila}\thanksref{t21}}
				\and
				\author{\authorb\ead[label=e2]{\emailb}\thanksref{t22}}
				\and
				\author{\authorc\ead[label=e3]{\emailc}\thanksref{t23}}
				
				\address{\addressa \\
					\printead{e1}}
				\address{\addressb \\
					\printead{e2}}
				\address{\addressc \\
					\printead{e3}}
				\thankstext{t21}{\thanksa}
				\thankstext{t22}{\thanksb}
				\thankstext{t23}{\thanksc}
				\affiliation{\affiliationa, \affiliationb, \affiliationc} 
				\runauthor{\runauthora, \runauthorb, and \runauthorc}
		\end{aug}}

		\begin{abstract}
			\theabstract
		\end{abstract}

		
		
	\end{frontmatter}
} 

\newenvironment{myexample}[2][]{\refstepcounter{example}\par\medskip
   \noindent \textbf{Example~\theexample} (#2)\textbf{.}\rmfamily}{\medskip}

    \newcounter{exercise}[section]
    \numberwithin{exercise}{section}
    \newtheorem{exrc}[exercise]{Exercise}
    \newenvironment{exercise}{\begin{exrc}\rm}{\end{exrc}}



\def\betav{\bb{\beta}}
\def\gammav{{\boldsymbol{\gamma}}}
\def\xv{\mathbf{x}}
\def\av{\mathbf{a}}
\def\ev{\mathbf{e}}
\def\uv{\mathbf{u}}
\def\gv{\mathbf{g}}
\def\tv{\mathbf{t}}
\def\yv{\mathbf{y}}

\def\entrlq{\entrl_{1}}
\def\entrlg{\entrl_{2}}

\def\R{\mathbb{R}}
\def\E{\mathsf{E}}
\def\P{\mathsf{P}}
\def\Rplus{\R_{+}}
\def\S{\mathbb{S}}
\def\H{\mathbb{H}}
\def\Splus{\S_{+}}
\def\N{\mbox{N}}
\def\Z{\mathbb{Z}}
\def\BigO{\mathcal{O}}
\def\T{\top}

\def\BigO{\mathcal{O}}

\def\cond{\vert\,}
\def\trinorm{|\mkern-2mu|\mkern-2mu|}
\def\Ltrinorm{\left|\mkern-2mu\left|\mkern-2mu\left|}
\def\Rtrinorm{\right|\mkern-2mu\right|\mkern-2mu\right|}

%
%
\def\Normal{\mathcal{N}}
\def\FF{\mathcal{F}}
\def\SS{\mathcal{S}}
\def\WW{\mathcal{W}}
\def\XX{\mathcal{X}}
\def\AA{\mathcal{A}}
\def\BB{\mathcal{B}}
\def\II{\mathcal{I}}
\def\HH{\mathcal{H}}

\def\dist{\mathsf{d}}
\def\Dist{\mathsf{D}}
\def\sign{\text{sign}}
\def\Arg{\text{Arg}}
\def\cone{\text{Cone}}
\def\Cone{\cone}
\def\conv{\text{Conv}}
\def\tr{\mathrm{Tr}}
\def\Rank{\mbox{Rank}}
\def\Tr{\tr}
\def\Diag{\mathrm{Diag}}
\def\diag{\mathrm{diag}}
\def\Off{\mathrm{Off}}
\def\floor#1{\lfloor #1 \rfloor}
\def\ceil#1{\lceil #1 \rceil}
\def\Ind{\boldsymbol{1}}
\def\Ent{\mathrm{Ent}}
\def\Var{\mbox{Var}}

\def\dmin{d_{\min}}
\def\dmax{d_{\max}}
\def\do{d^{\circ}}

\def\eps{\varepsilon}
\def\epsv{{\boldsymbol \eps}}
\def\Thetah{\hat{\Theta}}
\def\Sigmah{\hat{\Sigma}}

\def\gh{{\hat{g}}}
\def\Qh{\hat{Q}}
\def\Qs{Q^{*}}

\def\yv{\mathbf{y}}
\def\bv{\mathbf{b}}
\def\cv{\mathbf{c}}
\def\phiv{\boldsymbol\phi}
\def\supp{\mathrm{supp}}
\def\bvt{\tilde{\bv}}
\def\Dh{\hat{D}}
\def\cvh{\hat{\cv}}

\def\Db{\bar{D}}
\def\cvb{\bar{\cv}}

\def\bvo{\overline{\bv}}

\def\DD{\mathcal{D}}
\def\vv{\mathbf{v}}
\def\Id{I}
\def\dmin{\underline{\delta}}
\def\sv{\mathbf{s}}
\def\pen{\mathrm{pen}}
\def\wv{\mathbf{w}}
\def\uvt{\tilde{\uv}}
\def\fo{\overline{f}}

\def\CC{\mathcal{C}}
\def\1{\boldsymbol{1}}
\def\zv{\mathbf{z}}
\def\Sigmat{\tilde{\Sigma}}
\def\Phib{\bar{\Phi}}
\def\Phio{\overline{\Phi}}

\def\op{\mathsf{op}}
\def\Frob{\mathsf{F}}
\def\influ{\mathfrak{m}}

\newcommand{\qlet}{\hspace{\fill} {\includegraphics[scale=0.012]{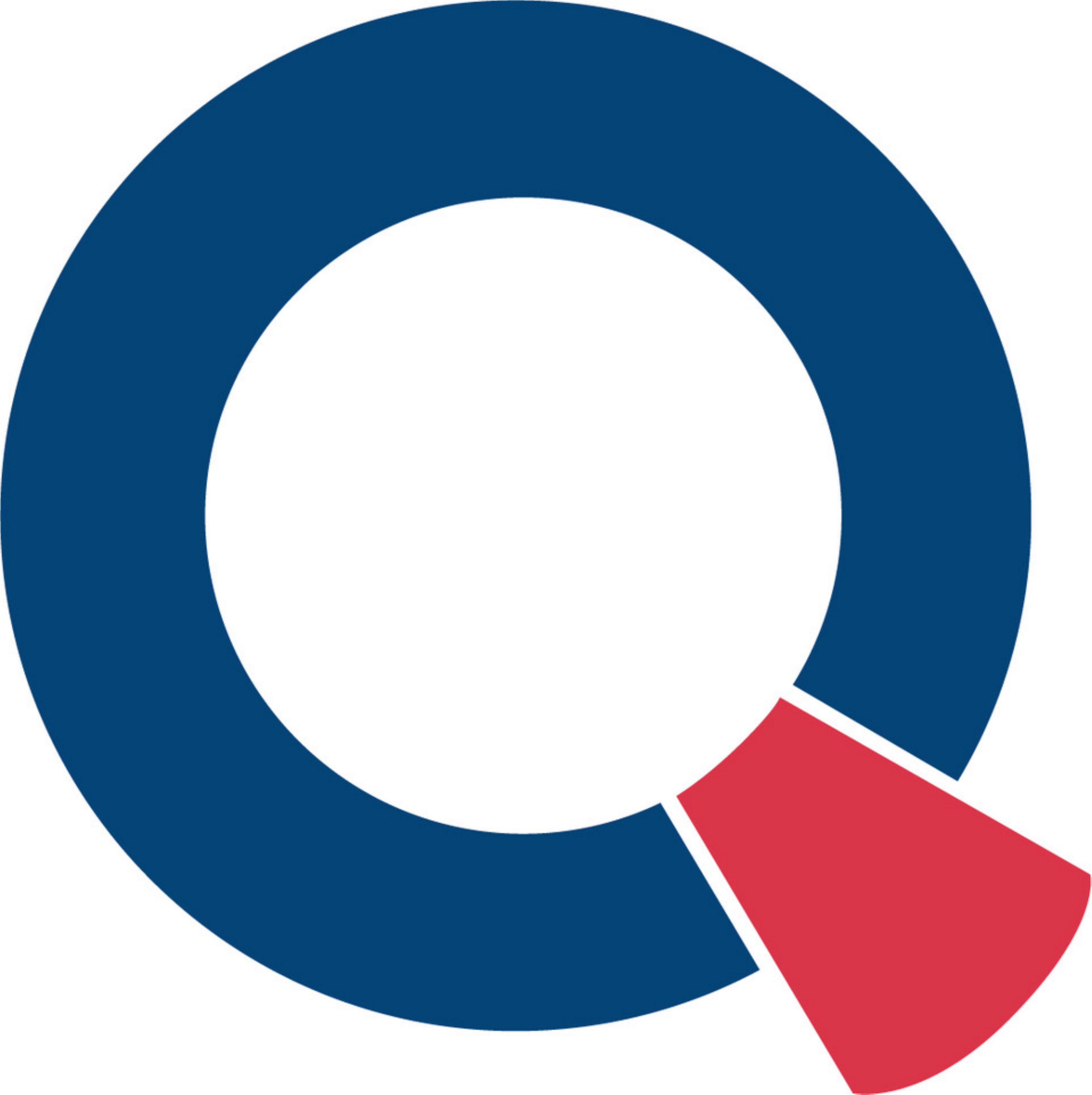}}\,}

%
%
\definecolor{light-green}{rgb}{0.9, 1, 0.8}
\definecolor{dark-green}{rgb}{0.1, 0.5, 0.1}
\definecolor{light-red}{rgb}{1, 0.9, 0.8}

\def\lambdao{\overline{\lambda}}
\def\tilspace{~}

\newcommand{\vertiii}[1]{{\left\vert\kern-0.25ex\left\vert\kern-0.25ex\left\vert #1 
		\right\vert\kern-0.25ex\right\vert\kern-0.25ex\right\vert}}

\section{Introduction}

A network is defined through a set of nodes and edges with a given adjacency structure. In a social, financial, or econometric context, such networks are often dynamic, and nodes, such as individuals or firms, are changing their activities over time. An analysis of such network dynamics is often based on \emph{vector autoregression}.
Consider a network that produces a time series \( Y_{t} \in \R^{N} \), \(t = 1, \dots, T \) and dependencies between its elements are modeled through the equation
\begin{equation}\label{VAR:def}
	Y_{t} = \Theta Y_{t - 1} + W_{t},
\end{equation}
where \( W_{t} \) are innovations that satisfy \( \E [ W_{t} \cond \FF_{t-1} ] = 0 \), \( \FF_{t} = \sigma\{ Y_{t-1}, Y_{t-2}, \dots  \} \), so that the interactions between the nodes are described by an autoregression operator \( \Theta \in \R^{N \times N} \). In terms of the network connections we say that a node \( i \) is connected to the node \( j\) if
\[
	\Theta_{ij}\neq 0 ,
\]
so that the nonzero coefficients represent the adjacency matrix of such network, and the sparsity of $\Theta$ represents the number of edges.
For large-scale time series, one encounters the curse of {dimensionality}, as estimating the matrix-parameter \(\Theta\) with $N^2$ elements requires a significantly large number of observations $T$. 

Several attempts to reduce the dimensionality have been made in the past literature. Assuming that the elements of a time series form a connected network, \cite{zhu2017network} introduce a Network Autoregression (NAR) with \( \Theta_{ij} = \beta A_{ij} / \sum_{k = 1}^{N} A_{ik} \), provided that the adjacency matrix \( A \in \R^{N \times N} \) is known. Here, the regression operator, defined up to a single parameter \( \beta \), which is called the \emph{network effect}, can be estimated through simple least squares. \cite{zhu2016quantile} also extend this model for conditional quantiles. Furthermore, \cite{zhu2017grouped} argue that a single network parameter may not be satisfactory as it treats all nodes of the network homogeneously. In particular, the NAR implies that each node is affected by its neighbors in the same extent, while in reality, we may have, e.g., financial institutions that are affected less or more than the others (see \cite{haerdle2019frm}). Hence they propose to detect communities in a network based on the given adjacency matrix and suggest that the nodes in each community share a separate network effect parameter. \cite{gudmundsson2018community} take a somewhat opposite direction: their BlockBuster algorithm determines the communities through the estimated autoregressive model, which, however, does not solve the dimensionality problem. Apart from this line of work, sparse regularisations have been extensively used, see
\cite{fan2009network, han2015direct, melnyk2016estimating}.

To sum up, we point out the following problems that one may encounter while dealing with vector autoregression in this social media context:
\begin{itemize}
\item The VAR parameter dimension is significant; one requires even larger time intervals for consistent estimation. Even if one can afford such a dataset, in the long run, {autoregressive models may have time-varying parameters}, see e.g., \cite{vcivzek2009adaptive}. We, therefore, impose some  assumptions on the \rev{structure of the} operator \( \Theta \), so that estimation through moderate sample sizes is possible.
\item The NAR model assumes that the adjacency matrix is known. In particular, this is justified for social networks with a stable and natural friendship/follower-followee relationship. For a realistic network of financial institutions, there is no explicitly defined adjacency matrix, and one has to heuristically evaluate it using additional information (identical shareholders, trading volumes) or through analyzing correlations and lagged cross-correlations between returns or risk profiles, see \cite{diebold2014network} and \cite{chen2019tail}. However, there is no rigorous reason to believe that the operator in \eqref{VAR:def} depends explicitly on such an adjacency matrix, see also \cite{cha2010measuring}.
\end{itemize}

Our main contribution is to propose a new method for modeling social network dynamics, which is a challenging task in the presence of the curse of dimensionality and the absence of knowledge of adjacency \rev{matrix}. The proposed SONIC --- \emph{SOcial Network analysis with Influencers and Communities } has the following advantages. First, it allows us to identify the hidden figures who mainly drive the opinion generating process on social media. Second, it discovers the hidden community structure. The proposed estimation algorithm uncovers the hidden figures and communities simultaneously until the minimal empirical risk is attained. Third, we discuss the theoretical properties and underpinnings to ensure estimation efficiency. Apart from dimensionality,  the social media data are featured with missing observations, bringing another challenge to researchers. The proposed SONIC is therefore equipped with a correction mechanism for missing observations.  We demonstrate the applicability of SONIC on a novel social media dataset. 

In more detail, the heuristics about \rev{the assumptions} on SONIC are motivated by social media users' activities and characteristics. Based on well-known user experience on platforms like facebook, twitter, etc., one can assume that some users have significantly more followers than others. Take, for example, celebrities, athletes, analysts, politicians, or Instagram divas. In a network view, these users are the nodes that have much more influence than the rest of nodes: these nodes are thereby defined as \emph{influencers}. In the framework of autoregression, a node \( j \) is an influencer if there is a \rev{substantial} amount of other nodes \(i\) such that \( \Theta_{ij} \neq 0 \). Assuming that the number of influencers is limited, we fix only a few columns of matrix \( \Theta \) to be \rev{non-zero}. This allows us to concentrate on the connections to the influencers, significantly reducing the number of parameters to be estimated. A similar idea is used in \cite{chen2018discover}, with a group-LASSO regularisation imposed, {yielding a solution with few active columns.}
Notice, however, that relying on the sparsity alone still requires \( T > N \), see e.g., \cite{fan2009network, chernozhukov2018lasso}. 

It is also well-known that social networks have small communities, with the nodes exhibiting higher connection density or similar behavior inside communities. \cite{zhu2017grouped} analyze a more realistic set-up by allowing separate parameters for each community instead of a single network effect parameter. In our notation, the conditional mean of the response of the node \( i\) satisfies
\[
	\E [ Y_{it} \vert\; \mathcal{F}_{t-1} ] = \Theta_{i1} Y_{1t-1} + \dots + \Theta_{iN} Y_{Nt-1}.
\] 
Therefore, the behavior of the node \(i\) is characterized by the coefficients \( \Theta_{i1}, \dots, \Theta_{iN} \) {i.e., the nodes it depends upon}. We assume that the nodes are separated into a few clusters such that the nodes from the same cluster share the same dependency structure, which brings a bigger picture into the view: instead of saying that two nodes from the same cluster are more likely to be connected, we say that they connect to the same influencers. 

Our primary focus is an application to the opinion dynamics extracted from a microblogging platform dedicated to stock trading, StockTwits (available at \url{https://stocktwits.com}.) For each user, one can quantify the average sentiment score, via a textual analysis, over the messages he posts during the day. Analyzing these high-dimensional time series, on the one hand, we can identify influencers --- the users whose opinions are overwhelmingly important, and on the other hand, we determine the community structure. One challenge emerges here: the presence of missing observations since sometimes users do not leave any message. We treat this as follows: assume there is an underlying opinion process that follows network dynamics \eqref{VAR:def}. However, such an opinion process might be partially observed, given the random arrival of messages from each user, which renders a commonly used  model for missing observations that involve masked Bernoulli random variables. The proposed SONIC accommodates this situation. We return to it in detail in Section~\ref{section:missing}.

The rest of the paper is organized as follows. Section~\ref{section:stocktwits} introduces readers to the StockTwits platform, describes in detail the available dataset and the process of users' sentiment scores extraction. In Section~\ref{section:main}, we first introduce our SONIC model, then describe the estimation procedure and provide a consistency result. In Section~\ref{section:simu}, we provide simulation results that support the theoretical properties of our estimator. Next, in Section~\ref{section:application}, we present and discuss the results of the application of our model to the dataset retrieved from StockTwits. Section ~\ref{sec:conclusion} concludes. We dedicate Section~\ref{section:proof} to the proofs, as well as Sections~\ref{section:covariance}, \ref{tropp_exact_recovery:sec} in the appendix. Readers can find all numerical examples and the codes developed for the SONIC model on \url{www.quantlet.de}.


\section{StockTwits}\label{section:stocktwits}

Social media are an ideal platform where users can easily communicate with each other, exchange information, and share opinions. The increasing popularity of social media is evidence of growing demand for exchanging opinions and information among granular users. Among social media platforms, we are particularly interested in StockTwits for several reasons.  Firstly, it is a social media platform designed for sharing ideas between investors, traders, and entrepreneurs. It is similar to Twitter but dedicated to the discussion on financial issues. One of the innovations that led to its popularity is a well-designed reference between the message content and the mentioned stock symbols. Conversations are organized around `cashtags' (e.g., `\$AAPL' for APPLE; `\$BTC.X' for BITCOIN) that allow to narrow down streams on specific assets. Secondly, users can express their sentiments/opinions by labeling their messages as `Bearish' (negative) or `Bullish' (positive) {via} a toggle button.  These are so-called \textit{self-report sentiments}, and these labeled data \rev{permits the use of supervised} textual analysis that requires the training dataset.  

We use the StockTwits Application Programming Interface (API) to retrieve all messages containing the preferred cashtags. StockTwits API also provides for each message its unique user identifier, the time it was posted within one-second precision and the sentiments declared by users (`Bullish,' `Bearish,' or unclassified). Among over thousand tickers/symbols, we particularly pick up two symbols,  \$AAPL for APPLE; \$BTC.X for BITCOIN, which represents the most popular security and cryptocurrency, respectively.  Concerning the fact that two symbols may attract the investors/users with different degrees of interaction, we may uncover disparate network dynamics.  In Table~\ref{tab:sum_stat}, we summarize the messages' statistics and document the generated sentiment series. Firstly, the BTC investors tend to disclose their sentiment, evident by 44\% of labeled messages, whereas in AAPL only 28\% of messages are labeled. Secondly, an imbalance between the numbers of positive and negative messages shows that online investors are in general optimistic, also found by \citet{KIMKIM2014} and \citet{AVERYAL2016}.  As to the average message volume per day, we observe that AAPL certainly attracts more attention than BTC does.

\begin{table}[htp!]
	\begin{center}
		\begin{tabularx}{\linewidth}{@{\extracolsep{\fill}}lrr}
			\hline\hline
			\textit{Symbols}                                         & AAPL & BTC \\
			\cline{1-3}
			\quad message volume                         & 449,761     &    644,597\\
			\quad number of distinct users (N)           & 26,521       &     25,492          \\
			\quad number of bullish messages     &  133,316           & 196,555      \\
			\quad number of bearish messages      &   48,186           &  90,677        \\
			\quad percentage of bullish messages     &  20.6\%           & 30.4\%       \\
			\quad percentage of bearish messages     &   7.4\%          & 14.0\%        \\
			\quad percentage of labeled messages     &   28.0\%          & 44.4\%        \\
			\quad mean of sentiment						   	 &   0.285         & 0.292			\\
			\quad standard deviation of sentiment   	 &   0.478         & 0.397			\\	
			\quad size of positive training dataset     &    99,985         & 147,759        \\
			\quad size of negative training dataset     &    36,100         & 67,752        \\
			\quad message volume per day    &    730         & 305        \\
			\quad number of positive terms in lexicon     &   4,000          & 3,775       \\
			\quad number of negative terms in lexicon     &  4,000           & 3,759        \\
			\quad number of daily observations (T)			& 423 				&  2108			\\
			\quad sample period 	 							&  2017-05-22          & 2013-03-21        \\
			\quad 													 &  2019-01-27          & 2018-12-27        \\
			\hline \hline
		\end{tabularx}
	\end{center}
	\caption{Summary statistics of social media messages}\label{tab:sum_stat}
	\small
\end{table}

\subsection{Quantifying message content}
\label{SecTextual}

Two main methods are used for textual sentiment analysis: the dictionary-based approaches and the machine learning techniques.  We opt for the dictionary-based approach in consideration of transparency, \rev{comprehension}, less computational burden and short texts. StockTwits, like Twitter, limits message length to 140 characters, \rev{which further limits the power of a machine learning-based approach concerning little contextual information on the short texts}. A dictionary, or lexicon, is a list of words labeled as positive, negative, or neutral. Given such a list, the \emph{bag-of-words} approach consists of counting the number of positive and negative words in a document in order to assign it a sentiment value or a tone. For example, a simple dictionary containing only the words `good' and `bad' with positive and negative labels, respectively, would classify the sentence `Bitcoin is a good investment' as positive with a tone +1.  

The simplicity of the dictionary-based approach guarantees transparency and replicability, on the con side, it comes with the limitations on natural language analysis.  First,  referring to \citet{DENG2017} to the `context of discourse,'  one needs to be aware of the content domain, to which language interpretation is sensitive. For example, \citet{LM2011JOF} point out that words like `tax' or `cost' are classified as negative by Harvard General Inquirer lexicon, whereas they should be considered neutral in the financial context. Another example is about quantifying sentiment on cryptocurrency.  \cite{chen2019what} point out that many domain-specific terms, such as `blockchain,' `ICO,' `hackers,' `wallet,' and `binance,' `hodl,' are not covered in the existing financial and psychological dictionaries. They construct a new cryptocurrency lexicon in response to the need of adopting a specific approach to measure sentiment about cryptocurrencies.  The second limitation is about the language domain, which \citet{DENG2017} defines as the `lexical and syntactical choices of language.' One example would be the difference between newspapers where one mostly finds a formal and standardized tone, and social media, where slang and emojis prevail. As observed, online investors often use new `emojis' such as \emojirocket (positive) and \emojipoop (negative) when talking about cryptocurrencies. These are missing in the traditional dictionary. 

Bearing the aforementioned considerations in mind, in the sentiment quantification for the messages of AAPL we employ the social media lexicon developed by \cite{renault2017intraday}, while in the case of BTC we advocate the lexicon tailored for cryptocurrency asset, by \cite{chen2019what}. \cite{renault2017intraday} demonstrates that the constructed lexicon significantly outperforms the benchmark dictionaries while remaining competitive with high-level machine learning algorithms. Based on 125,000 bullish and another 125,000 bearish messages published on StockTwits, using the lexicon for social media achieves 90\% of classified messages and 75.24\% of correct classifications.\footnote{The percentage of correct classification is defined as the proportion of correct classifications among all classified messages, while the percentage of classified messages is denoted as the proportion of classified messages among all messages.} With a collection of 1,533,975 messages from 38,812 distinct users, posted between March 2013 and December 2018, and related to 465 cryptocurrencies listed in StockTwits\footnote{This list can be found at \href{https://api.stocktwits.com/symbol-sync/symbols.csv}{https://api.stocktwits.com/symbol-sync/symbols.csv}}, \cite{chen2019what} documents that implementing the crypto lexicon can classify 83\% of messages, with 86\% of them correctly classified.

To convert unstructured text into a machine-readable text,  we proceed by the natural language processing (NLP) using \hyperlink{https://www.nltk.org/}{NLTK toolkit}.  First, all messages are lowercased. Tickers (`\$BTC.X,' `\$LTC.X,' ...), dollar or euro values, hyperlinks, numbers, and mentions of users are respectively replaced by the words `cashtag,' `moneytag,' `linktag,' `numbertag,' and `usertag'. The prefix ``negtag\_" is added to any word consecutive to `not,' `no,' `none,' `neither,' `never,' or `nobody'. Finally, the three stopwords `the,' `a,' `an' and all punctuation except the characters `?' and `!' are removed. For each collected message we filter the terms appearing in the designated lexicon, and equally weight the filtered terms to generate the sentiment score of message, which also means that the sentiment score of a message is estimated as the average over the weights of the lexicon terms it contains. Since the weights of the terms lexicon are in the range of \( -1\) and \( +1\), the sentiment scores fall in this range.

To visualize the resulting sentiment scores from individuals over time, we select the top 100 active users and display their daily sentiment scores over time. The heatmap shown in Figure \ref{fig:heatmap} is a 2-dimensional matrix with \(y\)-axis for user's ID, and \(x\)-axis for message posting date, the cell of the heatmap is the quantified sentiment score. The level of sentiment is color-coded, so that the evolution and dynamics of sentiment among users can be read in such a heatmap presentation. \rev{It appears that users express diverging opinions over time.}
From Figure~\ref{fig:heat_aapl} (AAPL) or Figure~\ref{fig:heat_btc} (BTC), one observes the similar color codes among a group of users at particular date or period, indicating a contemporaneous and potentially intertemporal dependency among users' sentiment time series. The correlation matrices of users' sentiment time series in Figure ~\ref{fig:corr_aapl} and ~\ref{fig:corr_btc} exhibit an  interdependence on sentiment series. In most cases, we observe positive dependencies and the dependencies seem to be centered on a group of users. For those who exhibit negative dependence with others, we may classify them as contrarians, a type of investors whose purchasing and selling decisions are in contrast to the prevailing sentiment. 

By aggregating the individual sentiment scores from 26K users in APPL and 25K users in BTC respectively, for each symbol we construct daily aggregate sentiment indicator by averaging out, at a 24-hour interval, the sentiment score of individual messages published per calendar day. Figure \ref{fig:sentiment} displays the dynamics of aggregate sentiment on APPL and BTC. Such sentiment dynamics may be featured with the hidden community structure and perhaps are driven by a small subset of users.  For the sake of brevity, in Table \ref{tab:sum_stat}, we only report the summary descriptive statistics for the aggregate sentiment indicator. In the case of AAPL (BTC), the mean and the standard deviation of sentiment indicator are 0.285 (0.292) and 0.478 (0.397), respectively. Again, it shows that the sentiment on social media is quite positive.  Compared to the sentiment on AAPL, the sentiment on BTC is more exuberant and relatively volatile.

\begin{figure}
	\centering
	\begin{subfigure}[b]{0.8\textwidth}
		\centering
		\includegraphics[width=\textwidth]{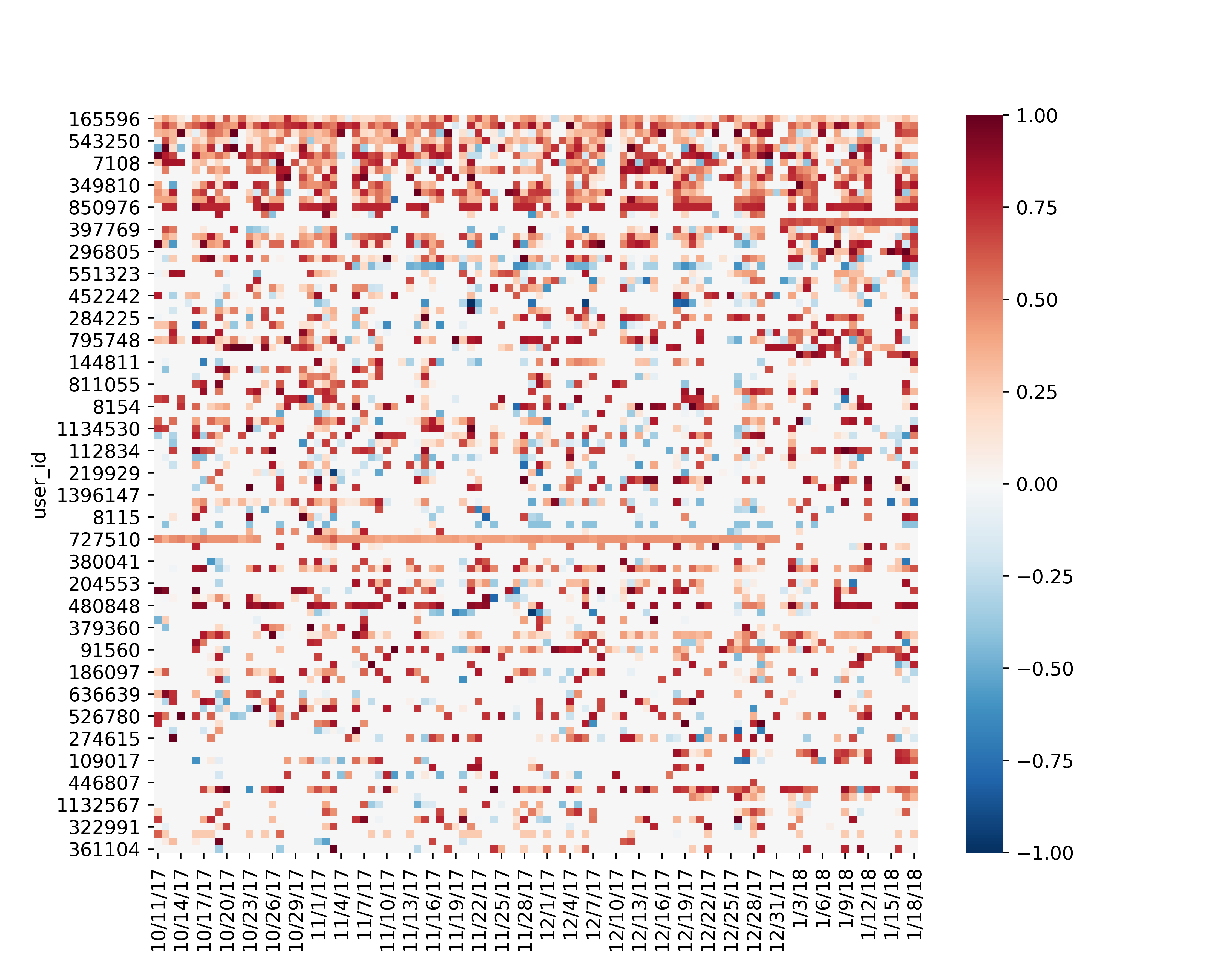}
		\caption{AAPL users}
		\label{fig:heat_aapl}
	\end{subfigure}
	\hfill
	\begin{subfigure}[b]{0.8\textwidth}
		\centering
		\includegraphics[width=\textwidth]{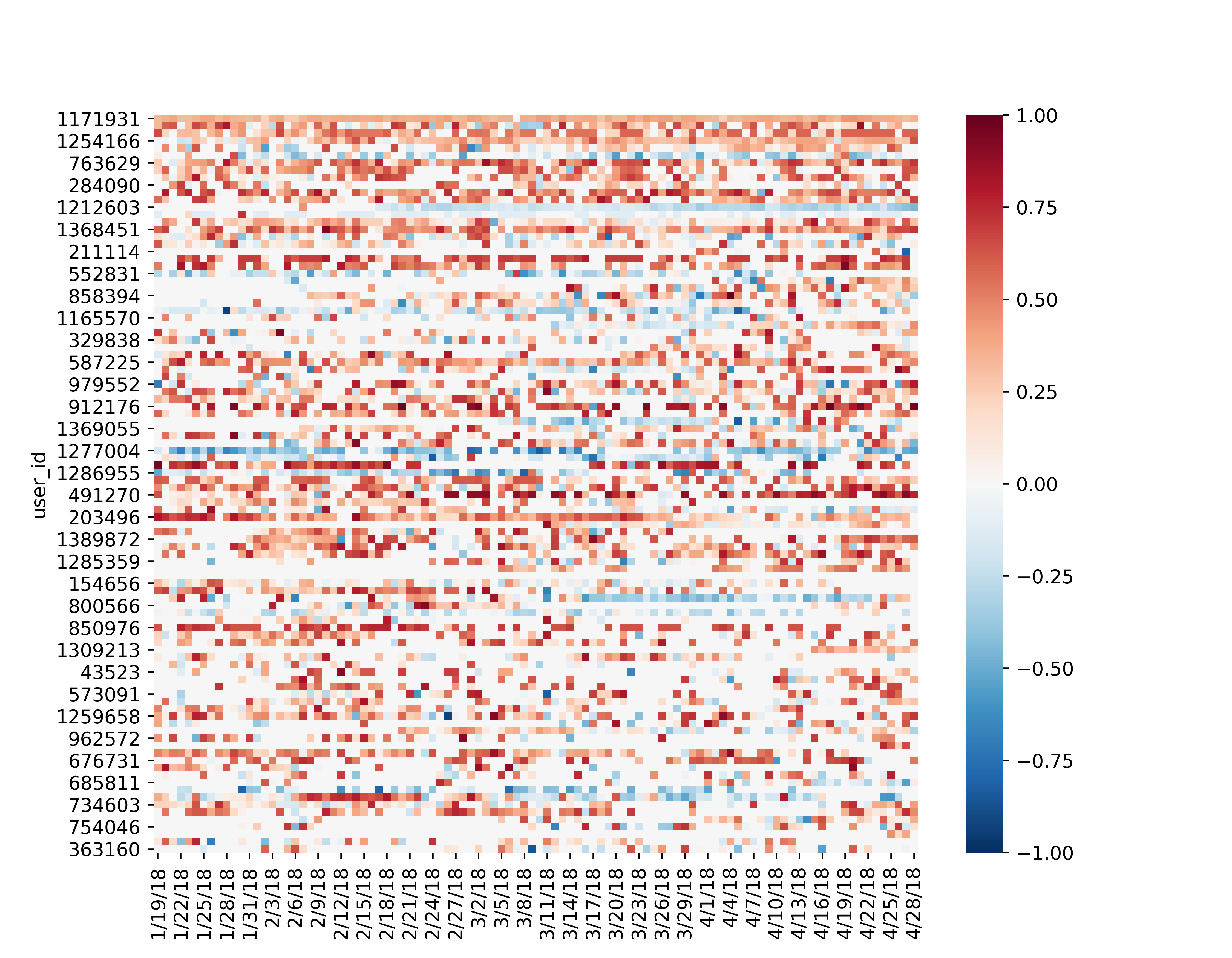}
		\caption{BTC users}
		\label{fig:heat_btc}
	\end{subfigure}
	\caption{Social media users' sentiment over time}
	{\footnotesize $y$-axis is the user's id, while $x$-axis is time stamp.\par}
	\label{fig:heatmap}
\end{figure}

\begin{figure}
	\centering
	\begin{subfigure}[b]{0.8\textwidth}
		\centering
		\includegraphics[width=\textwidth]{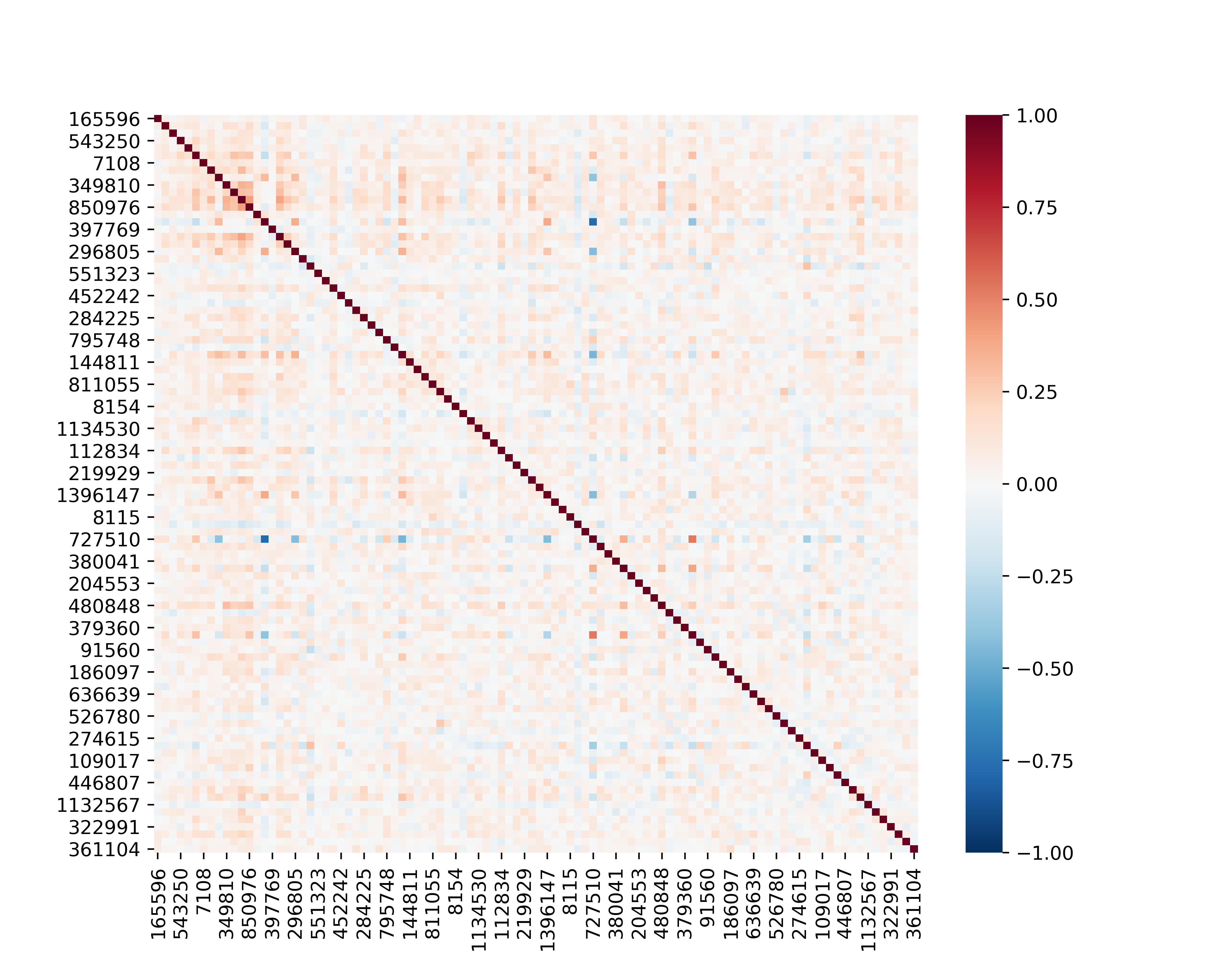}
		\caption{AAPL}
		\label{fig:corr_aapl}
	\end{subfigure}
	\hfill
	\begin{subfigure}[b]{0.8\textwidth}
		\centering
		\includegraphics[width=\textwidth]{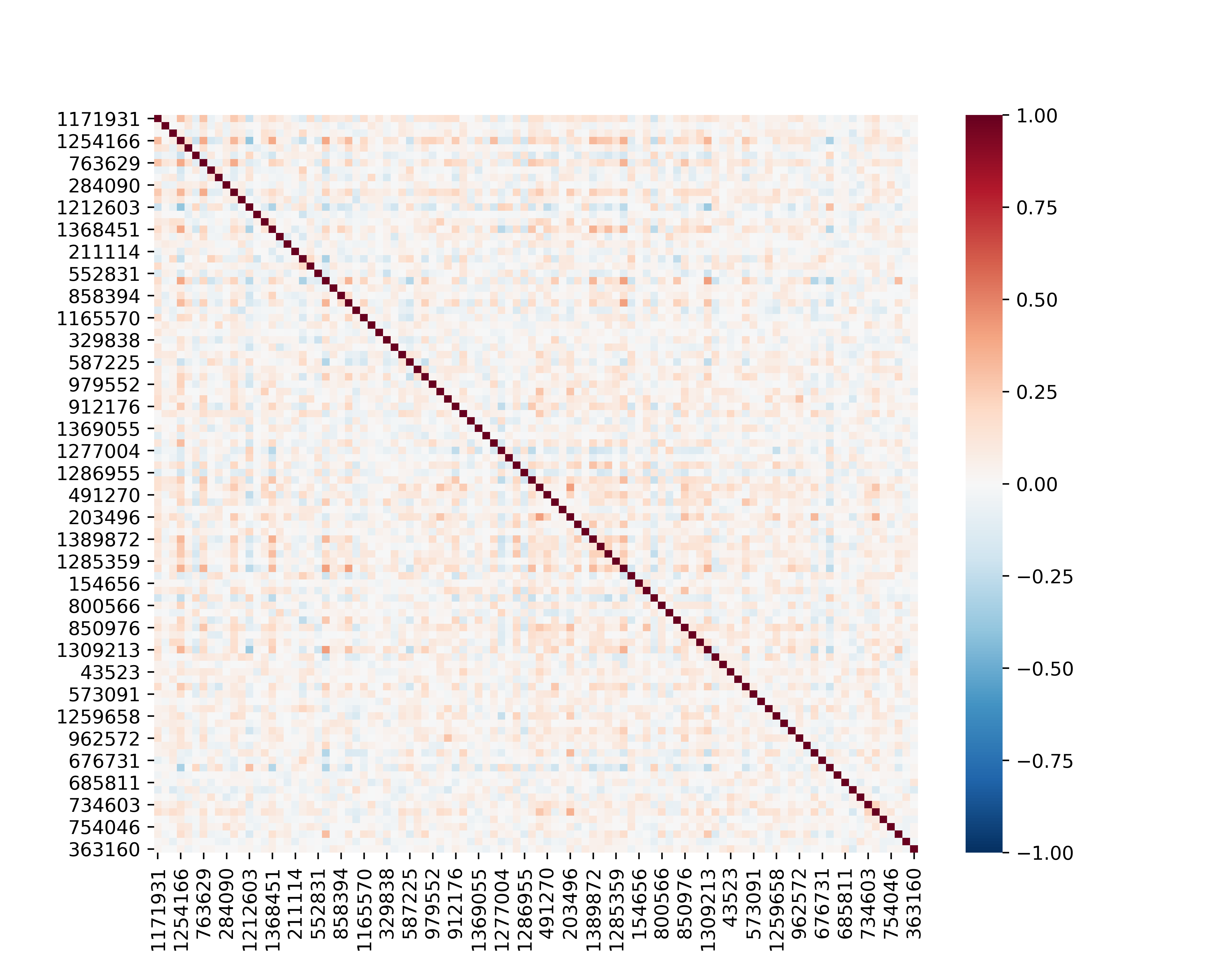}
		\caption{BTC}
		\label{fig:corr_btc}
	\end{subfigure}
	\caption{Correlation matrix of users' sentiment time series}
	\label{fig:corr}
\end{figure}

\begin{figure} 
	\centering
	\includegraphics[width=\textwidth]{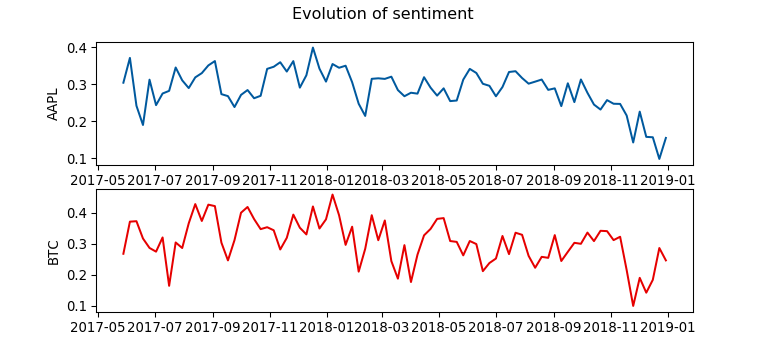}
	\caption{Sentiment evolve over time}
	{\footnotesize Daily aggregate sentiment indicator for each symbol is obtained by averaging, at 24-hour intervals, the sentiment scores of individual messages published per calendar day. \par}
	\label{fig:sentiment}
\end{figure}

\section{The SONIC model}\label{section:main}

\subsection{Notation} 

Let us first introduce some basic notations. Through the whole paper, \( N \) always denotes the size of the network. Denote by \( [N] \) the set of integers from \( 1 \) to \( N\), i.e.,  \([N] = \{ 1, \dots, N \} \). For a subset of indices \( \Lambda \subset [N]\) we denote its complement \( \Lambda^{c} = [N] \setminus \Lambda\). Moreover, if \( A \) is a \({N \times N} \) matrix and \( \Lambda_1, \Lambda_2 \subset [N]\) are two subsets of indices, we denote the submatrix \( A_{\Lambda_1, \Lambda_2} = (A_{ij})_{i \in \Lambda_1, j \in \Lambda_{2}}\). We also write for short \( A_{\Lambda, \cdot} = A_{\Lambda, [N]} \) and \( A_{\cdot, \Lambda} = A_{[N], \Lambda} \). 

Furthermore, for a vector \( \av \in \R^{d} \) denote a square matrix \( \diag\{ \av \} \in\R^{d \times d} \) that has the values \( a_1, \dots, a_d \)  on the diagonal and zeros elsewhere. For a square matrix \( A \in \R^{d \times d} \) we denote \( \Diag(A) \in \R^{d \times d} \) as a diagonal matrix of the same size that coincides with \( A \) on the diagonal, i.e., \( \Diag(A) = \diag(A_{11}, \dots, A_{dd}) \). For the off-diagonal part we use the notation \( \Off(A) = A - \Diag(A) \).  

For a real vector \( \xv \in \R^{d} \) and \( q\geq 1 \) or \( q = \infty \) denote the \(\ell_{q}\)-norm \( \| \xv\|_{q}  = (|x_1|^q + \dots + |x_d|^q )^{1/q}\); for \( q = 2 \) we ignore the index, i.e., \( \| \xv \| = \| \xv \|_{2} \); we also denote the pseudo-norm \( \| \xv \|_{0} = \sum_{i} \Ind(x_i \neq 0) \). 
For \( A \in \R^{d_1 \times d_2} \), \(  \sigma_{1}(A) \geq \sigma_{2}(A) \geq \dots \geq \sigma_{\min(d_1, d_2)}(A) \) denote the non-trivial singular values of \(A\). We will also refer to \( \sigma_{\min}(A) \) as the least nontrivial eigenvalue, i.e., \( \sigma_{\min}(A) = \sigma_{\min(d_1, d_2)}(A) \). Furthermore, we write \( \trinorm A \trinorm_{\op} = \max_{j} \sigma_{j}(A) \) for the spectral norm and \( \trinorm A \trinorm_{\Frob} = \tr^{1/2}(A^{\T} A) = \left(\sum_{j = 1}^{\min(p, q)} \sigma_{j}(A)^2 \right)^{1/2} \) for the Frobenius norm. Additionally, we introduce element-wise norms \( \| A \|_{p, q} \) for \( p, q \geq 1 \) (including \(\infty\)) denotes \( \ell_{q} \) norm of a vector composed of \( \ell_{p} \) norms of rows of \(A\), i.e., \( \| A \|_{p, q} = \left( \sum_{i} \left(\sum_{j} |A_{ij}|^{p} \right)^{q/p} \right)^{1/q} \). Notice that \( \| A \|_{2, 2} = \trinorm A \trinorm_{\Frob} \). {Finally, let \(\ev_{1}, \dots, \ev_{N}\) denotes the standard basis in \rev{\(\R^{N}\)}, i.e. \( \ev_{i} = (0, \dots, 0, 1, 0, \dots) \) with element $1$ at the \(i\)-th position.}

\subsection{\rev{The structure of operator $\Theta$}: Influencers \& communities}
\label{section:clusters_influencers}

In our set-up, the behavior of each node \( i \in [N] \) is characterized by the coefficients \( \Theta_{i1}, \dots, \Theta_{iN} \), and when we group the nodes using their characteristics the notion of community is merged with the notion of cluster.
We assume that the nodes are separated into clusters, such that these coefficients remain quantitatively comparable for the nodes within each cluster. Let us first give a precise definition of a clustering.

\begin{definition}
A \emph{\( K \)-clustering} of the set of the nodes \( [N]  \) is called a sequence \( \CC = (C_1, \dots, C_K) \) of \(K\) subsets of \( [N] \), such that
\begin{itemize}
	\item any two subsets are disjoint \( C_i \cap C_j = \emptyset \) for \( i \neq j \);
	\item the union of subsets \( C_j \) gives all nodes,
	\[
		C_1 \cup \dots \cup C_K = \{1, \dots, N  \}.
	\]
\end{itemize}
Two clusterings \( \CC \) and \( \CC' \) are equivalent if the corresponding clusters are equal up to a relabeling, i.e., there is a permutation \( \pi \) on \( \{1, \dots, K  \} \), such that \( C_{j} = C_{\pi(j)}' \) for every \( j = 1, \dots, K \).

Furthermore, define a distance between two clusterings as
\[
	d(\CC, \CC') = \min_{\pi} \sum_{j = 1}^{K} |C_j \setminus C_{\pi(j)}'|.
\]
\end{definition}
\begin{remark}
The distance between clusterings is, in fact, the minimal amount of node transferring from one cluster to another, that is required to make the clusterings equivalent. To see this, notice that each clustering can be defined as a sequence \( (l_1, \dots, l_N) \) of \( N \) labels taking values in \( \{1, \dots, K  \} \), so that each cluster is defined as \( C_{j} = \{ i :\; l_i = j \} \). Then, if the clustering \( \CC' \) corresponds to the labels \( l_1', \dots, l_N' \), the distance between them equals to
\[
	d(\CC, \CC') = \min_{\pi} \sum_{i = 1}^{N} \Ind(l_i \neq \pi(l_i')) .
\]
\end{remark}

We specify our model by \rev{imposing assumptions} concerning the communities and the presence of influencers.

\begin{definition}
We say that \( \Theta \in \mathsf{SONIC}(s, K) \) (SOcial Network with Influencers and Communities) if
\begin{itemize}
	\item each user is influenced by at most \( s \) influencers, i.e.,
	\[
		\max_{i} \sum_{j = 1}^{N} \Ind(\Theta_{ij} \neq 0) \leq s;
	\]
	\item there is a \( K \)-clustering \( \CC = (C_1, \dots, C_K) \) such that
	\[
		\Theta_{ij} = \Theta_{i'j},
		\qquad
		j = 1, \dots, N
	\]
	whenever \( i, i' \) are from the same cluster \( C_{l} \), \( l = 1, \dots, K \).
\end{itemize}
We will also say that \( \Theta \) has clustering \( \CC \).
\end{definition}

Once \( \Theta \in \mathsf{SONIC}(s, K) \) has clustering \( \CC = (C_1, \dots, C_K) \), the following factor representation takes place
\begin{equation}\label{theta_factor}
	\Theta = Z_{\CC} V^{\T},
\end{equation}
where \( Z_{\CC}, V \) are \( N \times K \) matrices such that
\begin{itemize}
	\item \( Z_{\CC} = [\zv_{C_1}, \dots, \zv_{C_K}] \) is a normalized index matrix of clustering \( \CC \), where for any \( C \subset [N] \) we denote
	\[
		\zv_{C} = \frac{1}{\sqrt{|C|}} (\Ind(1 \in C), \dots, \Ind(N \in C) ) \in \R^{N}
	\]
	--- a normalized index vector for the cluster \( C \) and \( Z_{\CC}^{\T} Z_{\CC} = I_{K} \) ;
	\item \( V = [\vv_1, \dots, \vv_{K}] \) has sparse columns,
	\[
		\| \vv_{j} \|_{0} \leq s ,
	\]
	i.e., only a few nodes are active and carrying information;
\end{itemize}
\begin{figure}
	\centering
	\includegraphics[width=0.6\textwidth]{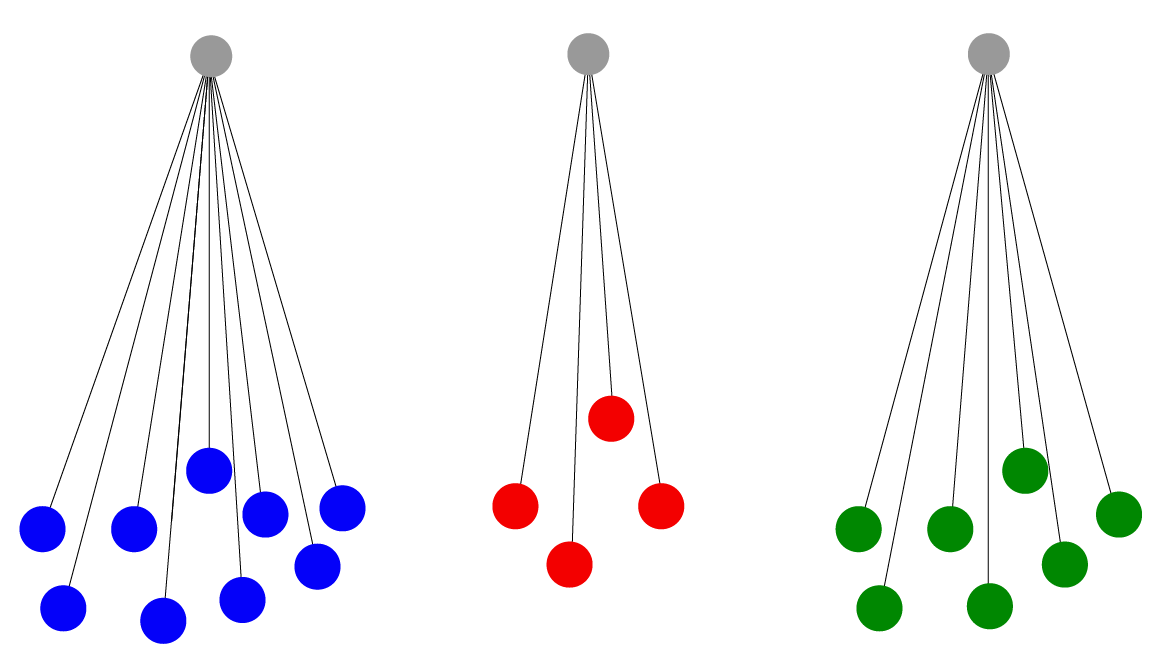}
	\caption{Example of a network with influencers for \( K = 3 \) and \( s = 1\).}
	\label{fig:influ_network}
\end{figure}

We present a schematic picture of what we expect in Figure~\ref{fig:influ_network}. Here, the nodes from the same clusters are subject to the same influencers (the grey nodes may be in any of the clusters), which also coincides with the idea of \cite{rohe2016co}, who looks for the right-hand side singular vectors of the Lagrangian in a directed network, grouping the nodes affected by the same group of nodes.

The equation \eqref{theta_factor} is akin to bilinear factor models, which appear in the econometric literature as a model with factor loadings, see e.g., \cite{moon2018nuclear} and the references therein. It is also a popular machine learning technique for low-rank approximation, see a thorough review in \cite{Boyd2016lowrank}. \cite{Schienle2016co} use sparse factors for a closely related model. We also mention the line of work \citep{kapetanios2019detection, parker2016identification, pesaran2020econometric} with a similar notion of dominant units, but in contrast with our analysis, they are defined through modeling cross-sectional dependencies. 

\subsection{Missing observations}\label{section:missing}

A network of size \(N\) represents a multivariate time series \( Y_{t} = (Y_{1t}, \dots, Y_{Nt})^{\T} \in \R^{N} \), where \( Y_{it} \) is the response of a node \(i = 1, \dots, N \) at a time \( t = 1, \dots, T \) and contaminated with missing observations. {Instead of specifying the exact distribution under the parametric model \eqref{VAR:def}, we assume there is a true parameter \( \Theta^{*} \in \R^{N \times N} \) and some unknown probability measure \( \P \) with the expectation \( \E \), such that under this measure the time series follows the autoregressive equation
\begin{equation}\label{auto_reg_for_cov}
	Y_{t} = \Theta^{*} Y_{t - 1} + W_{t},
\end{equation}
with \( \E[W_{t} \cond \FF_{t-1}] = 0 \) for \( \FF_{t-1} = \sigma(W_{t-1}, W_{t-2}, \dots) \). For the sake of simplicity,  we additionally  assume that \( W_{t} \) are independent and have \( \Var(W_{t}) = S \) under \(\P\).}  Once \( \trinorm \Theta^{*} \trinorm_{\op} < 1 \) the process exists as a converging series
\begin{equation}\label{ar_eq}
	Y_{t} = \sum_{k \geq 0} (\Theta^{*})^{k} W_{t-k},
\end{equation}
and the covariance of the process reads as
\begin{equation}\label{sigma_through_theta_s}
	\Sigma = \Var(Y_{t}) = \sum_{k \geq 0} (\Theta^{*})^{k} S \{(\Theta^{*})^{k}\}^{\T} .
\end{equation}
For simplicity, we consider \rev{\emph{sub-Gaussian}} vectors \( W_t \), as it allows us to have deviation bounds for covariance estimation with exponential probabilities. Recall the following definition, that appears, e.g., in \cite{vershynin2016high}.
\begin{definition}\label{subgaus_def}
A random vector \( W \in \R^{d} \) is called \(L\)-\rev{sub-Gaussian} if for every \( \uv \in \R^{d} \) it holds
\[
	\| \uv^{\T} W \|_{\psi_{2}} \leq L \| \uv^{\T} X \|_{L_2},
\]
where for a random variable \( X \in \R \) we denote
\begin{align*}
	\| X \|_{\psi_2} & = \inf\left\{ C> 0: \; \E \exp\left\{\left(\frac{|X|}{C}\right)^{2}\right\} \leq 2 \right\}, \\
	\| X \|_{L_2} &= \E^{1/2} |X|^{2} .
\end{align*}
\end{definition}

\rev{Estimating} SONIC is not impeded by the presence of missing data that appear to be one of the features of social media data. We adopt the framework of \cite{Lounici14} for vectors with missing observations, assuming that each variable \( Y_{it} \) is independent and only partially observed with some probability. Formally speaking, instead of having a realization of the whole vector \( Y_{t} \), we only observe the masked process \( Z_{t} \) defined as
\begin{equation}\label{covest_missing:def}
Z_{t} = (\delta_{1t} Y_{1t}, \dots, \delta_{Nt} Y_{Nt})^{\T},
\qquad
t = 1, \dots, T,
\end{equation}
where \( \delta_{it} \sim \mbox{Be}(p_i) \) are independent Bernoulli random variables for every \( i = 1, \dots, N \) and some \( p_{i} \in (0, 1] \), which means that each variable \( Y_{it} \) is only observed with probability \( p_{i} \) independently from other variables, with \( \delta_{it} = 1 \) corresponding to the observed \( Y_{it} \) and \( \delta_{it} = 0 \) to the {unobserved} \( Y_{it} \). Obviously, the case \( p_i = 1 \) for every \( i = 1, \dots, N \) corresponds to the process without missing observations. Therefore, the framework constituted by (\ref{covest_missing:def}) serves as a generalization of dynamic network models.

\begin{remark}
In terms of the StockTwits world, we interpret the process \( Y_{t} \) as an unobserved underlying \emph{opinion process}. Such an opinion process quantified from the messages is subject to random arrival of messages, as users disclose their opinions randomly on social media. Although one may restrict the sample to the case of full observation, the statistical inference may be questionable. Also, discarding nodes with very few missing observations is a waste of available information. Given the fact that some users are more active than others, we need to account for different probabilities \( p_{i} \).
\end{remark}

\def\pv{\mathbf{p}}
\def\ph{\hat{p}}
\def\phv{\hat{\pv}}
\def\rr{\tilde{\mathbf{r}}}

Notice that in general the probabilities \( p_{i} \) are not known, but can be easily estimated through the frequencies \( \hat{p}_i = T^{-1} \sum_{t = 1}^{T} \Ind[Y_{it} \neq 0] \). Set \( \phv = (\ph_{1}, \dots, \ph_{N})^{\T} \). Following \cite{Lounici14}, we {denote} the observed empirical covariance \( \Sigma^{*} = T^{-1}\sum_{t = 1}^{T} Z_t Z_t^{\T} \) and consider the following covariance estimator,
\begin{equation*}
	\hat{\Sigma} = \diag\{\phv \}^{-1} \Diag(\Sigma^{*}) + \diag\{ \phv \}^{-1} \Off(\Sigma^{*}) \diag\{ \phv \}^{-1}.
\end{equation*}
This estimator is motivated by the fact that \( \E \Sigma^{*}_{ii} = p_i \Sigma_{ii} \) and \( \E \Sigma^{*}_{ij} = p_i p_j \Sigma_{ij}  \) for \( i \neq j\) in the case of independent observations.
The state-of-the-art bound for the error of such covariance estimator is inspired by \cite{klochkov2018uniform}, Theorem~4.2. In the case of independent vectors \( Y_{t} \) and equal probabilities of observations \( p_1 = \dots = p_{N} = p \) they show that for any \( u \geq 1 \) with probability at least \( 1 - e^{-u} \) it holds
\[
	\trinorm \Sigmah - \Sigma \trinorm_{\op} \leq C \trinorm \Sigma \trinorm_{\op} \left( \sqrt{\frac{\rr(\Sigma) \log \rr(\Sigma)}{Tp^2}} \bigvee \sqrt{\frac{u}{Tp^2}} \bigvee \frac{\rr(\Sigma) \{\log\rr(\Sigma) + u\} \log T}{T p^2}  \right),
\]
where \( \rr(\Sigma) = \frac{\tr(\Sigma)}{\trinorm \Sigma \trinorm_{\op}} \) denotes the \emph{effective rank} of the covariance \( \Sigma \). Similarly, the effective rank appears as well in the classic covariance estimation problem (i.e., \(p=1\)), see, e.g., \cite{Koltch17} who even provide a matching lower bound. Notice that the effective rank takes values between \(1\) and the rank of \( \Sigma \). However, if there is no specific restriction on the spectrum of \( \Sigma \), the effective rank can grow as large as the full dimension \( N \), which means that the bound above can only guarantee the error of order \( \sqrt{\frac{N}{T p^2}} \), not taking into account the logarithms. 

On the other hand, one only needs to bound the error within specific low-dimensional subspaces. Say, given two projectors \( P \), \( Q \) of rank lower than \( N \), one needs to bound the error
\[
	\trinorm P(\hat{\Sigma} - \Sigma) Q \trinorm_{\op},
\]
which can be significantly smaller than the total error \( \trinorm \hat{\Sigma} - \Sigma \trinorm_{\op} \). For example, if we are interested in the error of estimation of \( \Sigma_{\Lambda, \Lambda} \), where \( \Lambda \subset [N] \), the corresponding projectors would have the form \( P = Q = \sum_{i \in \Lambda} \ev_{i} \ev_{i}^{\T} \). Notice that this projector will be sparse, in the sense that most of its values will be zeros, when \( |\Lambda| \) is much smaller than \(N\). In fact, due to the unknown probabilities \( p_i \), which we estimate via the frequencies, the ``sparsity'' of projectors \(P, Q\) will play an important role as well. We define it below.

\begin{definition}\label{def_projector_sparsity}
Let \( P \in \R^{N \times N} \) be a symmetric projector, i.e. \( P^2 = P\). Let \( \Lambda \subset [N] \) be the smallest set such that \( P_{ij} \) is nonzero only for indices \( i, j \in \Lambda\). Then, we refer to the value \( | \Lambda | \) as the \emph{sparsity of} \( P\).
\end{definition}

\begin{remark}
{
We {employ} this technical condition to state bounds for the error of the covariance estimator with {missing observations}. The corresponding diagonal projector \( \Pi_{\Lambda} = \sum_{i \in \Lambda} \ev_i \ev_i^{\T} \) commutes not only with \( P\), but also with any other diagonal operator, in particular, with \( \diag\{ \hat{p} \}^{-1} \). Thus, with the help of this larger projector (obviously, \(\mathrm{Rank}(P) \leq |\Lambda| \)) we can take into account the error that comes from the estimated frequencies.
}
\end{remark}

The following theorem provides a deviation bound for the autoregressive process \eqref{auto_reg_for_cov}. Unlike the bound of \cite{klochkov2018uniform}, it accounts for possibly distinct probabilities \( p_i \).

\begin{theorem}\label{missing_cov_est:prop}
	Assume the vectors \( W_{t} \) are independent \( L \)-\rev{sub-Gaussian} and also
	\[
	\trinorm \Theta^* \trinorm_{\op} \leq \gamma < 1,
	\qquad
	p_{i} \geq p_{\min} > 0.
	\] 
	Let \( P, Q \in \R^{N \times N} \) be two arbitrary orthogonal projectors of ranks \( M_1, M_2 \) and with sparsities \( K_1, K_2 \) respectively. Suppose, that \( u > 0 \) is such that
	\begin{equation}\label{sparse_projector}
		\max\{ 2, K_1,  K_2, \sqrt{K_1 K_2} \log T \} \frac{\log (4N) + u}{Tp_{\min}^{2}}  \leq 1 \, .
	\end{equation}
	Then, it holds with probability at least \( 1 - e^{-u} \) that
	\[
	\trinorm P(\Sigmah - \Sigma) Q \trinorm_{\op} \leq C \trinorm S \trinorm_{\op} \left(\sqrt{\frac{M_1 \vee M_2 (\log N + u)}{T p_{\min}^2}} \bigvee \frac{\sqrt{M_1 M_2}(\log N + u) \log T}{T p_{\min}^2}   \right),
	\]
	where \( C = C(\gamma, L) \) only depends on \( L \) and \( \gamma \).
\end{theorem}
See proof of this result in Section~\ref{section:covariance}.

Additionally, we are interested in estimating lag-\(1\) cross-covariance under the same scenario. Namely, based on the sample \( Z_1, \dots, Z_T \) and given the estimated probabilities \( \ph_1, \dots, \ph_N \), we wish to estimate the matrix \( A = \E Y_{t} Y_{t + 1}^{\T}\).
Since \( \E [Y_{t + 1} \vert \, \mathcal{F}_{t}] = \Theta^* Y_{t} \) for the linear process \eqref{ar_eq}, the corresponding cross-covariance reads as
\[
	A = \Sigma (\Theta^{*})^{\T} .
\]
Consider the following estimator
\[
	\hat{A} = \diag\{\phv\}^{-1} A^{*} \diag\{ \phv \}^{-1},
\]
where \( A^{*} \) is the observed empirical cross-covariance
\[
A^{*} = \frac{1}{T - 1} \sum_{t = 1}^{T - 1} Z_{t} Z_{t + 1}^{\T}.
\]
For this estimator, we provide an upper-bound, again with a restriction to some low-dimensional subspaces.

\begin{theorem}\label{cros_cov_missing:prop}
	Under conditions of Theorem~\ref{missing_cov_est:prop}, it holds, with probability at least \( 1 - e^{-u} \), that
	\[
	\trinorm P(\hat{A} - A) Q \trinorm_{\op} \leq C \trinorm S \trinorm_{\op} \left(\sqrt{\frac{(M_1 \vee M_2) (\log N + u)}{T p_{\min}^2}} \bigvee \frac{\sqrt{M_1 M_2}(\log N + u) \log T}{T p_{\min}^2}   \right),
	\]
	where \( C = C(\gamma, L) \) only depends on \( \gamma \) and \( L \).
\end{theorem}
We postpone the proof to Section~\ref{section:covariance}.

\subsection{Alternating minimization algorithm}

In order to estimate the matrix \( \Theta = Z_{\CC} V^{\T} \), we need to estimate both \( \CC \) and \( V \) simultaneously. Suppose that we have some clustering \( \CC \) at hand and we aim to estimate the corresponding \( V \). The mean squared loss from the fully observed sample {is:}
\begin{equation}\label{risk_no_missing}
\begin{aligned}
	R^*(V; \CC) = & \frac{1}{2(T-1)} \sum_{t = 1}^{T-1} \| Y_{t + 1} - Z_{\CC} V^{\T} Y_{t} \|^{2} \\
	= & \frac{1}{2} \tr(V^{\T} \tilde{\Sigma} V) - \tr(V^{\T} \tilde{A} Z_{\CC}) + \frac{1}{2(T-1)} \sum_{t = 1}^{T-1} \| Y_{t + 1} \|^{2},
\end{aligned}
\end{equation}
where we used the fact that \( Z_{\CC}^{\T} Z_{\CC} = I_{K} \) and the trace of a matrix product is invariant with respect to transition \( \tr(AB) = \tr(BA)  \). Here, we also denote
\[
	\tilde{\Sigma} = \frac{1}{T-1} \sum_{t = 1}^{T-1} Y_{t}Y_{t}^{\T},
	\qquad
	\tilde{A} = \frac{1}{T-1} \sum_{t = 1}^{T-1} Y_{t} Y_{t+1}^{\T},
\]
to be empirical covariance and empirical lag-1 covariance built on a sample \( Y_1, \dots, Y_{T} \), respectively, which we observe only partially. In reality, the feasible estimators are \( \hat{\Sigma} \) and \( \hat{A} \), which we have introduced in the previous section. A natural solution is to plug-in these estimators into the expression \eqref{risk_no_missing} instead of the unobserved \( \tilde{\Sigma}\) and \( \tilde{A} \). The last term \( \frac{1}{2(T-1)} \sum_{t = 1}^{T-1} \| Y_{t + 1} \|^{2} \) does not depend on the parameters \( \CC \) and \(V \) at all; therefore, we can drop it. We end up with the following risk function that we need to minimize,
\[
	{R}(V; \CC) = \frac{1}{2} \tr(V^{\T} \hat{\Sigma} V) - \tr(V^{\T} \hat{A} Z_{\CC}) .
\]
In particular, it is not hard to derive from Theorems~\ref{missing_cov_est:prop} and~\ref{cros_cov_missing:prop} that for any fixed pair \( \CC, V \) the values of \( {R}(V; C) \) and \( {R^*}(V; \CC) - \dfrac{1}{2(T-1)} \sum_{t = 1}^{T-1} \| Y_{t + 1}\|^{2} \) are close with high probability. 

As we are searching for a sparse matrix \( V \), we additionally impose a LASSO regularization and end up with the following convex optimization,
\begin{align*}
	\hat{V}_{\CC, \lambda} = \arg\min {R}_{\lambda}(V; \CC),
	\qquad
	R_{\lambda}(V; \CC) =& R(V; \CC) + \lambda \| V \|_{1, 1} \\
	=& \frac{1}{2} \Tr(V^{\T} \Sigmah V) - \Tr( V^{\T} \hat{A} Z_{\CC}) + \lambda \| V \|_{1, 1},
\end{align*}
where \( \| V \|_{1, 1} = \sum_{ij} |v_{ij}| \), and tuning parameter \( \lambda > 0 \) depends on the dimension \( N \) and number of observations \( T \). Concerning this minimization problem, we have the following observations: 
\begin{itemize}
\item the problem reduces to simple quadratic programming and therefore can be efficiently solved;
\item since \( \| V \|_{1,1} = \sum_{j = 1}^{K} \| \vv_{j} \|_{1} \) we can rewrite
\begin{align*}
	R_{\lambda}(V; \CC) = & \frac{1}{2} \Tr\left( V^{\T} \Sigmah V \right) - \Tr\left( V^{\T} \hat{A}  Z_{\CC} \right) + \lambda \| V \|_{1, 1} \\
	= & \sum_{ j = 1}^{K}  \frac{1}{2} \vv_{j}^{\T} \Sigmah \vv_{j} - \vv_{j}^{\T} \hat{A} \zv_{j} + \lambda \| \vv_{j} \|_{1} .
\end{align*}
Therefore, we need to solve \( K \) independent problems of size \( N \), which reduces computational complexity, and {may therefore be implemented} in parallel.
\end{itemize}
%
Ideally, we want to solve the following problem (note that the number of clusters \(K\) and the tuning parameter \( \lambda \) are fixed)
\begin{equation}\label{ideal_problem}
	F_{\lambda}(\CC) \rightarrow \min_{\CC},
	\qquad
	F_{\lambda}(\CC) = \min_{V} R_{\lambda}(V; \CC) .
\end{equation}
We can employ a simple greedy procedure. In the beginning, we initialize \( \CC^{(0)} = (l_{1}, \dots, l_{N}) \) randomly; each label takes values \( 1, \dots, K \). Then, at a step \( t \), we try to change one label of a node that reduces the risk the most, in other words, we try all the clusterings in the nearest vicinity of the current solution \( \CC^{(t)} \), i.e.,
\[
	\CC^{(t + 1)} = \arg\min_{d(\CC, \CC^{(t)}) \leq 1} F_{\lambda}(\CC) .
\]
At each such step, we would need to calculate \( F_{\lambda}(\CC) \) for \( \mathcal{O}\{N(K-1)\} \) different candidates.

\begin{remark}
In general, it is impossible to optimize an arbitrary function \( f(\CC) \) with respect to a clustering. The \(K\)-means is well-known to be NP-hard, however, different solutions are widely used in practice, see \cite{shindler2011fast} and \cite{likas2003global}.
\end{remark}

To speed up the trials of the greedy procedure, we utilize an alternating minimization strategy. Suppose, in the beginning, we initialize the clustering by \( \CC^{(0)} \) and compute the LASSO solution \( V^{(0)} = V_{\CC^{(0)}, t} \). When updating the clustering, we fix the matrix \( V = V^{(t)} \) and solve the problem
\[
	R_{\lambda}(V; \CC) = \frac{1}{2} \tr(V^{\T} \Sigmah V) - \tr(V^{\T}\hat{A} Z_{\CC}) + \lambda \| V \|_{1, 1} \rightarrow \min_{\CC} ,
\]
where only the term \( - \tr(V^{\T} \hat{A} Z_{\CC}) \) depends on \( \CC \). Minimizing by conducting a few steps of the greedy procedure we obtain the next clustering update \( \CC^{(t + 1)} \). Then, we again update the \(V\)-factor by setting \( V^{(t + 1)} = V_{\CC^{(t+1)}, \lambda} \). We continue so until the clustering does not change or the number of iterations exceeds a specific limit. The pseudo-code in Algorithm~\ref{alter_algo} summarizes this procedure.

\par 
\begin{algorithm}[H]
	\label{alter_algo}
	\SetAlgoLined
	\KwResult{a pair \( (\hat{\CC}, \hat{V}) \) }
	initialize \( \CC^{(0)} = (l_{1}^{(0)}, \dots, l_{N}^{(0)}) \) randomly\;
	\( t \leftarrow 0 \)\;
	\While{\( t < max\_iter \)}{
		update \( \hat{V}^{(t)} \leftarrow \arg\min R_{\CC^{(t)}, \lambda}(V) \)\;
		\For{\( i = 1, \dots, N \)}{
			\For{\( l = 1, \dots, K \)}{
				consider candidate \( \CC' = (l_{1}^{(t)}, \dots, l_{i-1}^{(t)}, l, l_{i+1}^{(t)}, \dots, l_{N}^{(t)})\)\;
				\( r_{il} \leftarrow -\tr(V^{(t)}\hat{A} Z_{\CC'}) \)\;
			}
		}
		\( (i^{*}, l^{*}) = \arg\min r_{il} \)\;
		update  \( \CC^{(t+1)} \leftarrow (l_{1}^{(t)}, \dots, l_{i^{*}-1}^{(t)}, l^{*}, l_{i^{*}+1}^{(t)}, \dots, l_{N}^{(t)})\)\;
		\eIf{\( \CC^{(t+1)} = \CC^{(t)}\)}{
			return \( (\CC^{(t)}, V^{(t)}) \)\;
		}{
			\(t \leftarrow t+1\)\;
		}
	}
	\caption{Alternating greedy clustering procedure.}
\end{algorithm}


\subsection{Local consistency result}

In this section, we show the existence of a locally optimal solution in the neighborhood of the true parameter with high probability. We call a clustering solution \( \hat{\CC} \) \emph{locally optimal} if the functional \( F_{\lambda}(\cdot) \) in (\ref{ideal_problem}) has the minimum value at point \( \hat{\CC} \) among its nearest neighbours \( d(\CC, \hat{\CC}) \leq 1 \). In particular, Algorithm~\ref{alter_algo} stops at such a solution. 

\subsubsection*{Conditions}

Here we describe the conditions that we need for the consistency result. The first condition concludes the requirements of Theorems~\ref{missing_cov_est:prop} and~\ref{cros_cov_missing:prop}.

\begin{assumption}\label{time_series:assume}
There is some \( \Theta^{*} \in \R^{N \times N} \) such that \( \trinorm \Theta^{*} \trinorm_{\op} \leq \gamma \) for some \( \gamma < 1 \) and the time series \(Y_t\) follows \eqref{ar_eq}. The innovations \( W_{t} \) are independent with \( \E W_{t} = 0 \) and \( \Var(W_t) = S\). Moreover, each \( W_{t} \) is \( L \)-\rev{sub-Gaussian}.
\end{assumption}

Furthermore, we impose assumptions on \rev{the structure of the true} parameter \( \Theta^{*} \) described in Section~\ref{section:clusters_influencers}.

\begin{assumption}\label{structured:assume}
The true VAR operator admits decomposition with \( K \)-clustering \( \CC^{*} \)
\[
	\Theta^{*} = Z_{\CC^{*}} V^{*},
\]
and meets the following conditions:
\begin{enumerate}
	\item \( \trinorm \Theta^{*} \trinorm_{\op}  = \trinorm V^{*} \trinorm_{\op} \leq \gamma < 1  \) for some constant \( \gamma \in (0, 1)\);
	\item cluster separation
	\begin{equation}\label{cl_sep}
		\sigma_{\min}([V^{*}]^{\T} \Sigma V^{*}) \geq a_{0}
	\end{equation}
	for some \( a_{0} > 0 \); 
	\item sparsity: for every \( j = 1, \dots, K \) the active set \( \Lambda_{j} = \supp(\vv_{j}^{*}) \) satisfies
	\[
		| \Lambda_{j}| \leq s ;
	\]
	\item active coefficients \rev{separated from zero}: there is \( \tau_{0} > 0 \) such that 
	\begin{equation}\label{coeff_separ:assume}
		| v^{*}_{ij} | \geq \tau_{0} s^{-1/2},
		\qquad
		i \in \Lambda_{j},
		\qquad
		j = 1, \dots, K \, .
	\end{equation}
	Here each \( \| \vv_j^{*} \| \leq 1 \) has at most \( s \) nonzero values, hence the normalization;
	
	\item significant cluster sizes: for some \( \alpha \in (0, 1) \) it holds 
	\[
		\frac{\min_{j} |C_j^{*}|}{\max_{j} |C_j^{*}|} \geq \alpha .
	\]
\end{enumerate}
\end{assumption}

Notice that the condition \eqref{cl_sep} corresponds to an appropriate separation of clusters, i.e., each \( \vv_{j}^{*} \) is far enough from a linear combination of the rest. Another assumption imposes conditions on the population covariance \( \Sigma \).

\begin{assumption}\label{sigma:assume}
The covariance of \( Y_t \) reads as
\[
	\Sigma = \sum_{k = 0}^{\infty} (\Theta^{*})^{k} S [(\Theta^{*})^{k}]^{\T},
\]
where \( S = \Var(W_{t}) \), and it is assumed that
\begin{enumerate}
	\item bounded operator norm
	\[
		\trinorm \Sigma \trinorm_{\op} \leq \sigma_{\max} ;
	\]
	\item restricted least eigenvalue
	\[
		\sigma_{\min}( \Sigma_{\Lambda_j, \Lambda_j}) \geq \sigma_{\min},
		\qquad
		j = 1, \dots, K \, .
	\]
\end{enumerate}
\end{assumption}

Note that we do not require that the smallest eigenvalue of \( \Sigma \) is bounded away from zero, but only those corresponding to the small subsets of indices are. Such assumption is not too restrictive. In fact,  \( \Sigma^{-1}_{\Lambda_j \Lambda_j}\) would correspond to the Fisher information if we were estimating the vector \( \vv_j \) knowing the cluster \( C_j^*\) and the sparsity pattern \( \Lambda_j \) in advance. 

For the sake of simplicity, we additionally assume that the ratio
\[
	\frac{\sigma_{\max}}{\sigma_{\min}} \leq \kappa,
\]
is bounded by some constant \( \kappa \geq 1 \). Additionally, we can treat the values \(L\),  \( \gamma\), \( a_{0}\), \(\tau_{0} \), and \( \alpha \) as constants. Below we focus on to what extent the relationship between \( N, T, s, K \), and the probabilities of the observations \( p_{i} \), \( i = 1, \dots, N\) allows consistent estimation of the parameter \( \Theta \).

Finally, we present the assumption that allows controlling the exact recovery of sparsity patterns for the LASSO estimator.

\begin{assumption}\label{ERC:assume}
For every \( j = {1}, \dots, {K} \) it holds
\[
	\| \Sigma_{\Lambda_j^{c}, \Lambda_j} \Sigma_{\Lambda_j, \Lambda_j}^{-1} \|_{1, \infty} \leq \frac{1}{4}. 
\]
Recall that \( \Lambda^{c}\) is the complement of \( \Lambda \subset [N] \) in \( [N] \).
\end{assumption}
\begin{remark}
\cite{zhao2006model} call the inequality \( \| \Sigma_{\Lambda_j^{c}, \Lambda_j} \Sigma_{\Lambda_j, \Lambda_j}^{-1} \|_{1, \infty} < \eta \) with constant \( \eta \in (0, 1)\) the strong Irrepresentable Condition. To avoid technical burden, we pick a concrete constant \( \eta = 1/4\). In a special case with fixed design and no noise, \cite{tropp2006just}  shows that the inequality \( \| \Sigma_{\Lambda_j^{c}, \Lambda_j} \Sigma_{\Lambda_j, \Lambda_j}^{-1} \|_{1, \infty} < 1 \) is necessary in order to be able to recover {the} sparsity pattern of \( \vv_j \). In Section~\ref{tropp_exact_recovery:sec}, we show a straightforward extension of Tropp's sparsity recovery results to the case with random design and missing observations.
\end{remark}

We are now ready to state our main theorem.

\begin{theorem}\label{main_thm}
Suppose that Assumptions~\ref{time_series:assume}-\ref{ERC:assume} hold. There are constants \( c, C > 0 \) that depend on \( L, \gamma  \) such that the following holds. Suppose,
\begin{equation}\label{n_star_cond}
	\sqrt{\frac{s n^* \log N}{Tp_{\min}^{2}}} \bigvee \sqrt{
	\frac{s \log N \log^2 T}{T p_{\min}^2}} \leq c,
\end{equation}
where \( n^* = \max_{j \leq K}|C_j^*| \) and, additionally, \( N \geq (C \alpha^{2} \vee \kappa) K \).
Then, with probability at least \( 1 - 1/N \) for any \( \lambda \) in the range 
\begin{equation}\label{lambda_range}
	C \sigma_{\max} \sqrt{\frac{\log N}{T p_{\min}^2}} \leq \lambda \leq c \left\{ \kappa^{-4} (a_{0}^{2} / \sigma_{\max}) K^{-2} s^{-1} \bigwedge \sigma_{\min} \tau_{0} s^{-1} \right\},
\end{equation}
and, additionally, \( \lambda \geq C \alpha^{2} K / N \),
there is a locally optimal solution \( \hat{\CC} \) satisfying
\[
	\trinorm Z_{\hat{\CC}} \hat{V}_{\hat{\CC}, \lambda}^{\T} - \Theta^{*} \trinorm_{\mathsf{F}} \leq \left\{ 3 \sigma_{\min}^{-1} \sqrt{Ks} + \frac{C \gamma}{a_{0}} \left(\frac{\sigma_{\max}}{\sigma_{\min}} \right)^{2} K\sqrt{s} \right\} \lambda \, .
\]
Moreover, the exact support recovery takes place, i.e., \( \supp(\hat{V}_{\hat{\CC}, \lambda}) = \supp(V^{*}) \).
\end{theorem}

\begin{remark}
In the above theorem we only show the existence of a local minimum of the functional \( F_{\lambda}(\mathcal{C}) \) defined in \eqref{ideal_problem} near the true clustering \(\mathcal{C}^*\) and, in addition, the statistical properties of the corresponding estimator \( \hat{\Theta}_{\lambda} \). This is not uncommon in the machine learning literature when dealing with non-convex bilinear models, see e.g. \cite{gribonval2015sparse}. In addition, we do not guarantee that the algorithm converges to the global minimum. Similarly, \cite{chen2020nonlinear, chen2014} offers a procedure that is only guaranteed to arrive at a local minimum, and suggest to run the algorithm several times to ensure that the global solution is covered.
\end{remark}

Let us discuss this result. According to the theorem, a greater \( \lambda \) gives greater error once it is in the required range. This comes naturally, as the result is based on the exact recovery, see e.g., \cite{tropp2006just}. Ideally, we want to choose the smallest available value,
\begin{equation}\label{lambda_opt_theo}
	\lambda^{*} = C \sigma_{\max}  \sqrt{\frac{\log N}{T p_{\min}^{2}}} .
\end{equation}
In this case, the error of the estimator reads as
\[
	\trinorm \hat{\Theta}_{\lambda^*} - \Theta^{*} \trinorm_{\Frob} \leq C' K \sqrt{\frac{s\log N}{T p_{\min}^{2}}} ,
\]
where \( C' \) does not depend on \( N, T, K, s \). Notice that in a hypothetical situation where the clustering \( \CC^{*} \) is known precisely, we only need to estimate the matrix \( V \) that consists of at most \( K s \) non-zero parameters. Therefore, according to Lemma~\ref{exact:lem}, the LASSO estimator must give us
\[
    \trinorm Z_{\CC^{*}} \hat{V}_{\CC^{*}, \lambda^{*}}^{\T} - \Theta^{*} \trinorm_{\Frob} = \trinorm \hat{V}_{\CC^{*}, \lambda^{*}} - V^{*} \trinorm_{\Frob} \leq C' \sqrt{\frac{Ks \log N}{T p_{\min}^{2}}},
\]
where we used the fact that \( Z_{\CC^*} \) has orthonormal columns; see also \cite{melnyk2016estimating} and \cite{han2015direct}. We may say in a loose way that not knowing the exact clustering provides an estimator that is at most \( \sqrt{K} \) times worse. 

Let us take a closer look at condition \eqref{n_star_cond}. Under the cluster size restriction from Assumption~\ref{structured:assume}, we have that all clusters have the size of order \( N / K \), since
\[
	\alpha \frac{N}{K} \leq |C_{j}^*| \leq \alpha^{-1} \frac{N}{K},
	\qquad
	j = 1, \dots, K .
\]
Therefore, if we ignore missing observations, we only need
\begin{equation}\label{network_size_2}
	\frac{(s N / K) \log N}{T} \leq c,
\end{equation}
with some constant \( c \) depending on \(\alpha\), enabling the estimation toward the parameters. So, once \( K \) is large enough, the estimator works with the corresponding error. Notice that the \( \ell_1 \)-regularisation alone requires the number of the observations to be at least the number of edges times \( \log N \), see \cite{fan2009network}. In our setting, the number of connections is up to \( Ns \), hence such a condition reads as
\[
	\sqrt{\frac{sN \log N}{T}} \leq 1 .
\]
Therefore, the SONIC model is an improvement in this regard. Finally, we point out that the conditions of Theorem~\ref{main_thm} imply some limitations on the size of the network concerning the number of observations. Indeed, using the first part of condition \eqref{n_star_cond} and comparing the lower- and upper-bounds of condition \eqref{lambda_range}, we can easily derive
\[
    \frac{N^{\frac{4}{5}} s^{\frac{6}{5}} \log N}{T p_{\min}^{2}} \leq c,
\]
where \( c > 0 \) is a constant that only depends on \(L\),  \( \gamma\), \( a_{0}\), \(\tau_{0} \), and \( \alpha \). Though we do not state that this condition is necesarry, it is clear that in some cases the estimation is possible even when \( N > T \).


\section{Simulation study}\label{section:simu}

The theoretical properties of the SONIC model and the developed theorems can be further supported via simulation. We check the discussed theorems and properties via relative estimation errors and cluster errors. We particularly discuss the choice of regularization parameter and number of clusters before turning to \rev{the} StockTwits applications. 


We set up the simulations as follows.
Take \( N = 100 \) and \( s = 1 \), while \( K \) will vary in the range \( 5...25 \). 
For every \( K = 5 , 10, 15, 20, 25 \) we construct the following matrix \( \Theta^{*} \),
\begin{itemize}
\item pick clusters \( C_j^{*} \) having approximately the same size \( \frac{N}{K}  \pm 1 \);
\item for every \( j = 1, \dots, K \) set
\[
	\vv_{j}^{*} = 0.5 \ev_{j} = (0, \dots, 0.5, \dots, 0)^{\T},
\]
with a single nonzero value at the place \( j \), so that \(s = 1\).
\item by construction we have,
\begin{align*}
	\trinorm \Theta^* \trinorm_{\op} &= \trinorm V^* \trinorm_{\op} = 0.5, \qquad
	\trinorm \Theta^* \trinorm_{\Frob} = \trinorm V^* \trinorm_{\Frob} = 0.5 \sqrt{K} .
\end{align*}
\end{itemize}
As for the sample size, we consider two scenarios: 
\begin{enumerate}
	\item[(a)] with \( T = 100\) and \( p_{i} = 1 \), i.e., no missing observations;
	\item[(b)] with \( T = 400 \) and \( p_{i} = 0.5 \), i.e., each \( Y_{it} \) is observed with probability \( 0.5 \).
\end{enumerate}
In order to generate the autoregressive process, we take i.i.d. \( W_{-19}, W_{-18}, \dots, W_{T} \sim \mbox{N}(0, I) \) and set
\[
	Y_{t} = \sum_{k = 0}^{20} (\Theta^*)^k W_{t - k},
	\qquad
	t = 1, \dots T ,
\]
where due to \( 0.5^{-20} \approx 10^{-6}  \) the terms for \( k > 20  \) can be neglected. In Figure~\ref{fig:alpha_to_loss} we show the relative error \( \E \trinorm \hat{\Theta} - \Theta^{*} \trinorm_{\Frob} / \trinorm \Theta^* \trinorm_{\Frob}  \) along the regularization paths for different choices of \( K \). Picking the best \( \lambda \), we show the relative error against the number of clusters in Figure~\ref{fig:K_to_loss}. We also show that the clustering error \( \E d(\hat{\CC}, \CC^*) \) in Figure~\ref{fig:K_to_cluster_loss} is subject to the choice of \( K \). All expectations are estimated based on \(20\) independent simulations.

{
Evidently, within the considered range of cluster numbers, larger ones lead to a smaller relative error as well as smaller clustering error. The simulations partially confirm the discussion in the end of the previous section, namely, that the conditions of Theorem~\ref{main_thm} can be met when \(K\) is large enough, although not too large.
In addition, we can see that the graphs for the scenario (a) with \( T = 100\) and \(p_{i} = 1 \), and the graphs for the scenario (b) with \( T = 400 \) and \( p_{i} = 0.5 \) are almost identical, except, perhaps, for the small \(\lambda\) in Figure~\ref{fig:K_to_loss}. This is consistent with the results of Section~\ref{section:missing} and with the Theorem~\ref{main_thm}, where the value \( T p_{\min}^{2} \) plays the role of the \emph{effective} number of observations.}
%
%
%
\begin{figure}[t]
	\begin{subfigure}[b]{0.5\textwidth}
		\centering
		\includegraphics[width=1.1\textwidth]{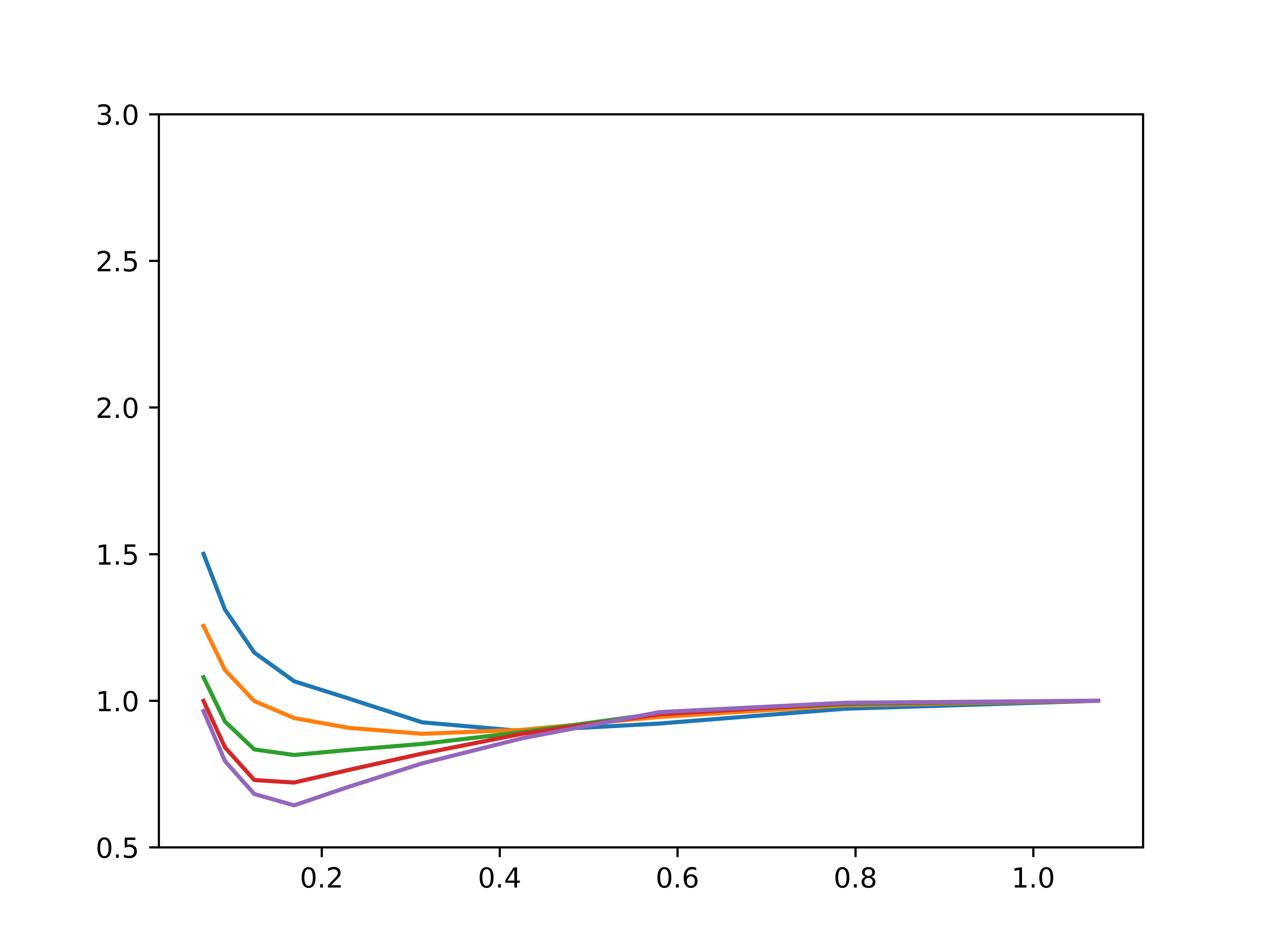}
		\caption{\( T = 100\) and  \(p_{\min} = 1 \).}
		\label{fig:alpha_to_loss_pmin_1}
	\end{subfigure}
	~
	\begin{subfigure}[b]{0.5\textwidth}
		\centering
		\includegraphics[width=1.1\textwidth]{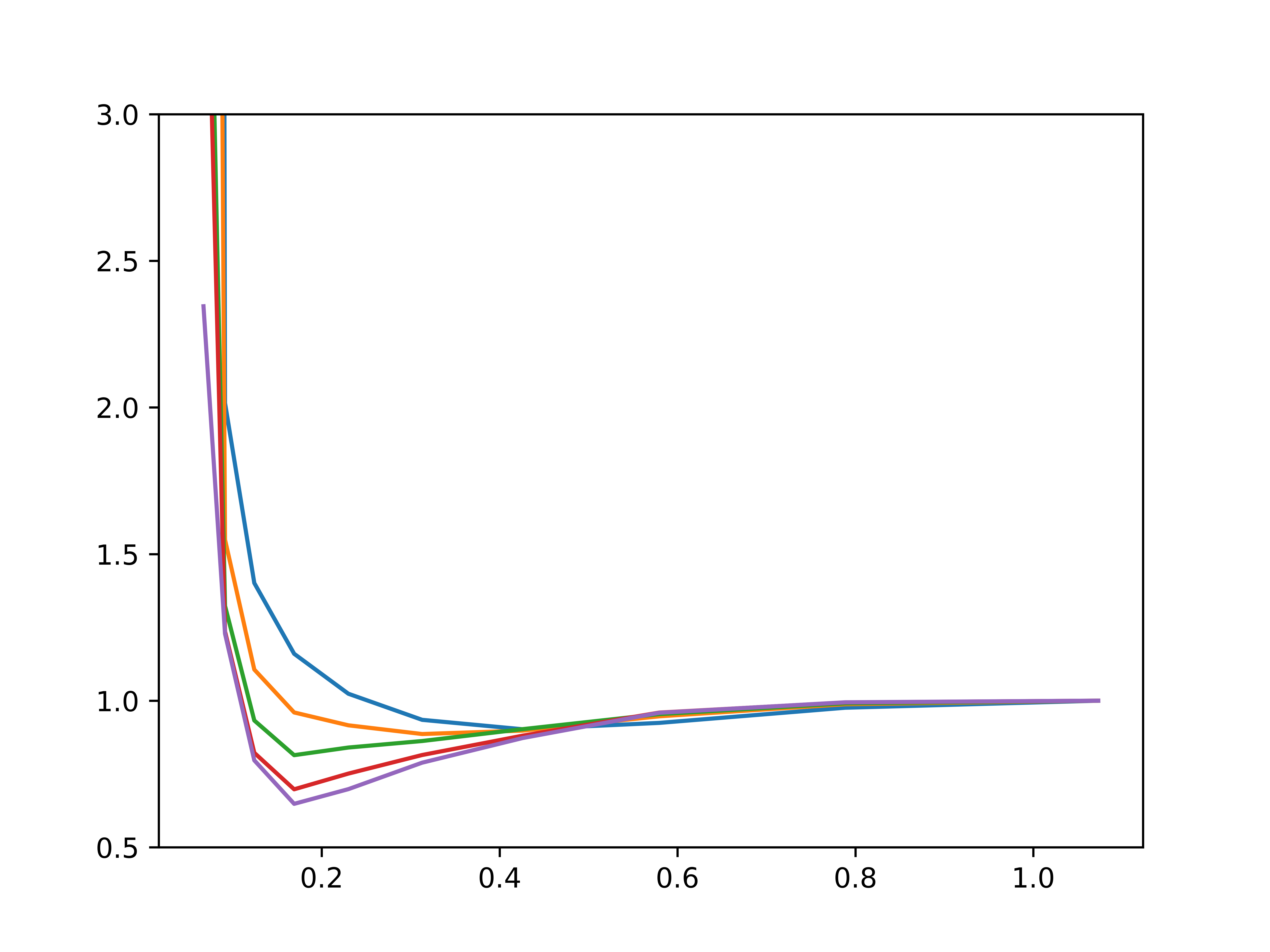}
		\caption{\(T = 400 \) and \(p_{\min} = 0.5\).}
		\label{fig:alpha_to_loss_pmin_05}
	\end{subfigure}
	\caption{Expected relative loss \( \E \frac{\trinorm \hat{\Theta} - \Theta^* \trinorm_{\Frob}}{\trinorm \Theta^* \trinorm_{\Frob} } \) for different \( \lambda \). \( N = 100\), and \( K = \textcolor{myblue}{5}, \textcolor{myorange}{10}, \textcolor{darkspringgreen}{15}, \textcolor{myred}{20}, \textcolor{mypurple}{25} \).}
	\label{fig:alpha_to_loss}
	\qlet 
	\href{https://github.com/QuantLet/SoNIC/tree/master/SoNIC_simulation_study}{SoNIC\_simulation\_study}
\end{figure}

\begin{figure}[t]
	\begin{subfigure}[b]{0.5\textwidth}
		\centering
		\includegraphics[width=1.1\textwidth]{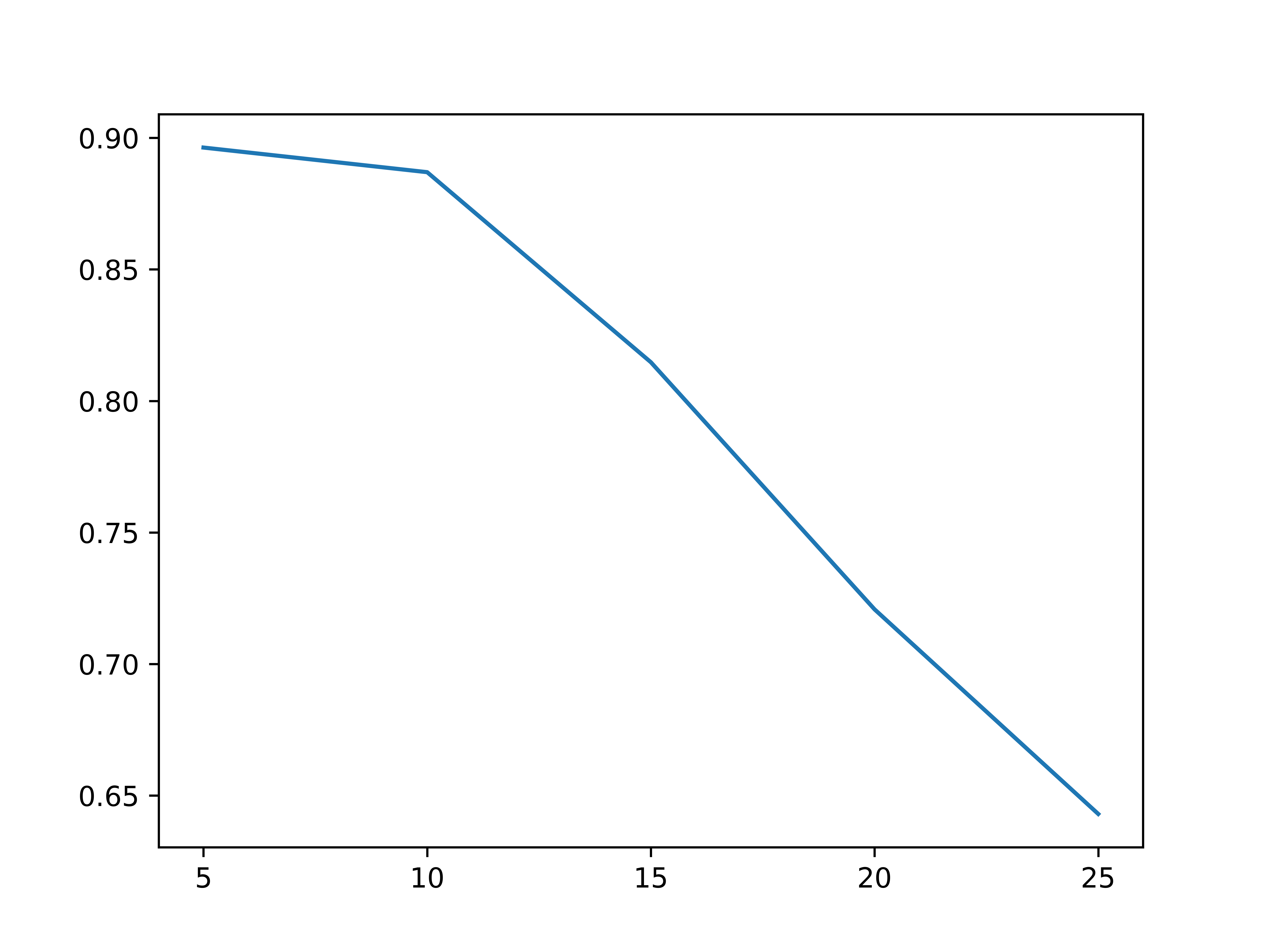}
		\caption{\( T = 100\) and  \(p_{\min} = 1 \).}
		\label{fig:K_to_loss_pmin_1}
	\end{subfigure}
	~
	\begin{subfigure}[b]{0.5\textwidth}
		\centering
		\includegraphics[width=1.1\textwidth]{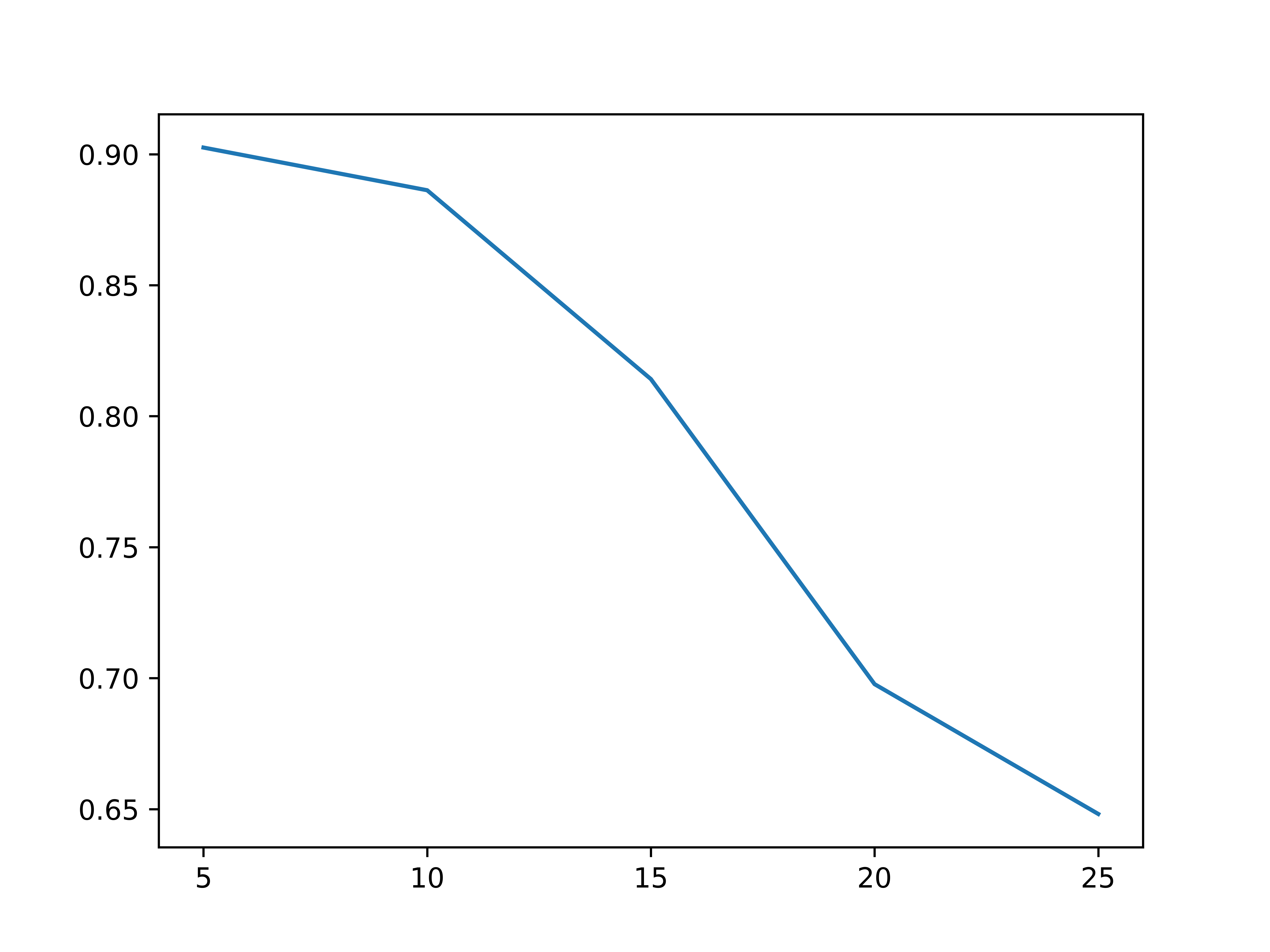}
		\caption{\(T = 400 \) and \(p_{\min} = 0.5\).}
		\label{fig:K_to_loss_pmin_05}
	\end{subfigure}
	\caption{Expected relative loss \( \E \frac{\trinorm \hat{\Theta} - \Theta^* \trinorm_{\Frob}}{\trinorm \Theta^* \trinorm_{\Frob} } \) for optimal \( \lambda \), \( N = 100\), and \( K = 5,...,25 \).}
	\label{fig:K_to_loss}
	\qlet 
	\href{https://github.com/QuantLet/SoNIC/tree/master/SoNIC_simulation_study}{SoNIC\_simulation\_study}
\end{figure}

\begin{figure}[t]
	\begin{subfigure}[b]{0.5\textwidth}
		\centering
		\includegraphics[width=1.1\textwidth]{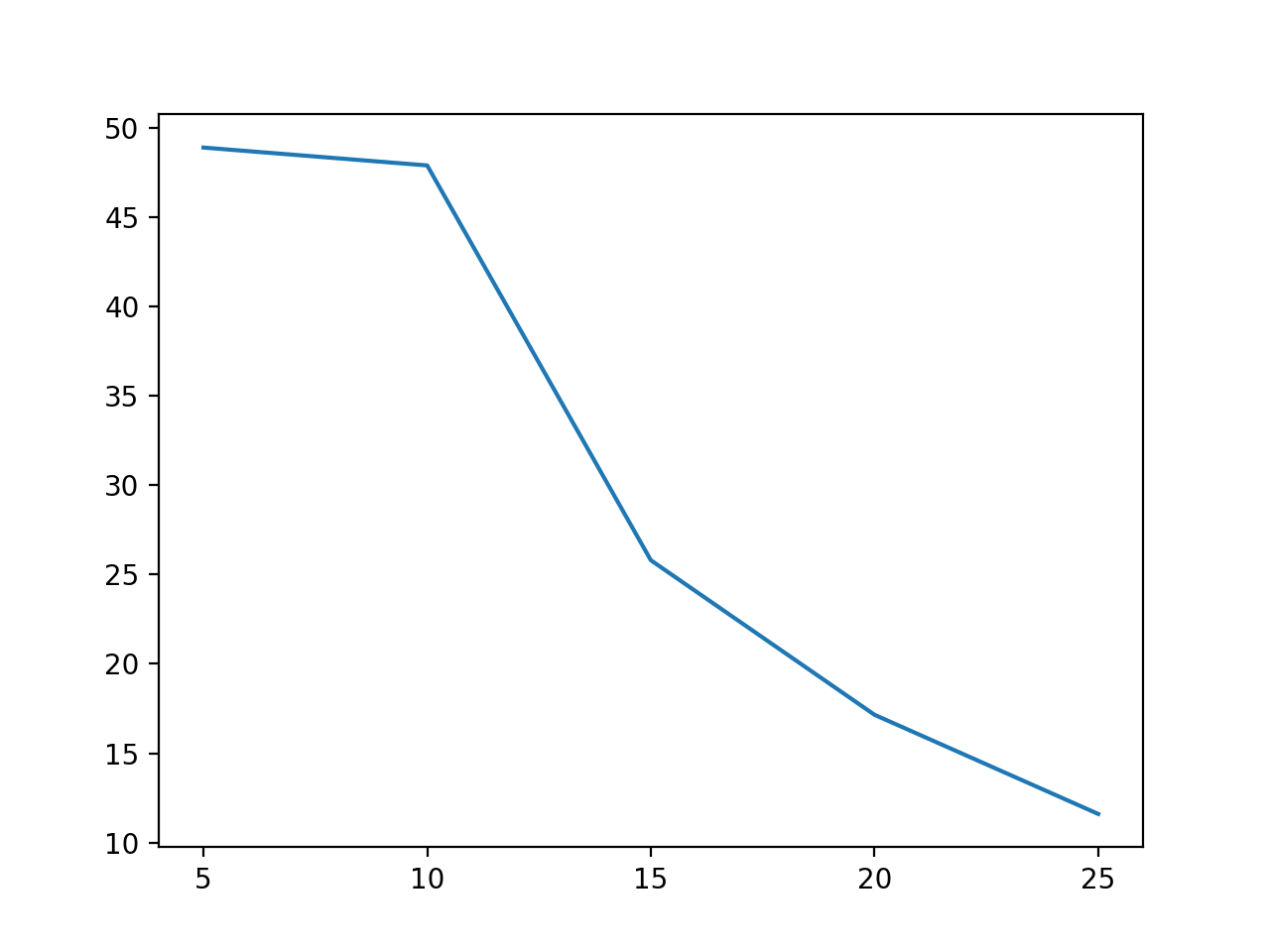}
		\caption{\( T = 100\) and  \(p_{\min} = 1 \).}
		\label{fig:K_to_cluster_loss_pmin_1}
	\end{subfigure}
	~
	\begin{subfigure}[b]{0.5\textwidth}
		\centering
		\includegraphics[width=1.1\textwidth]{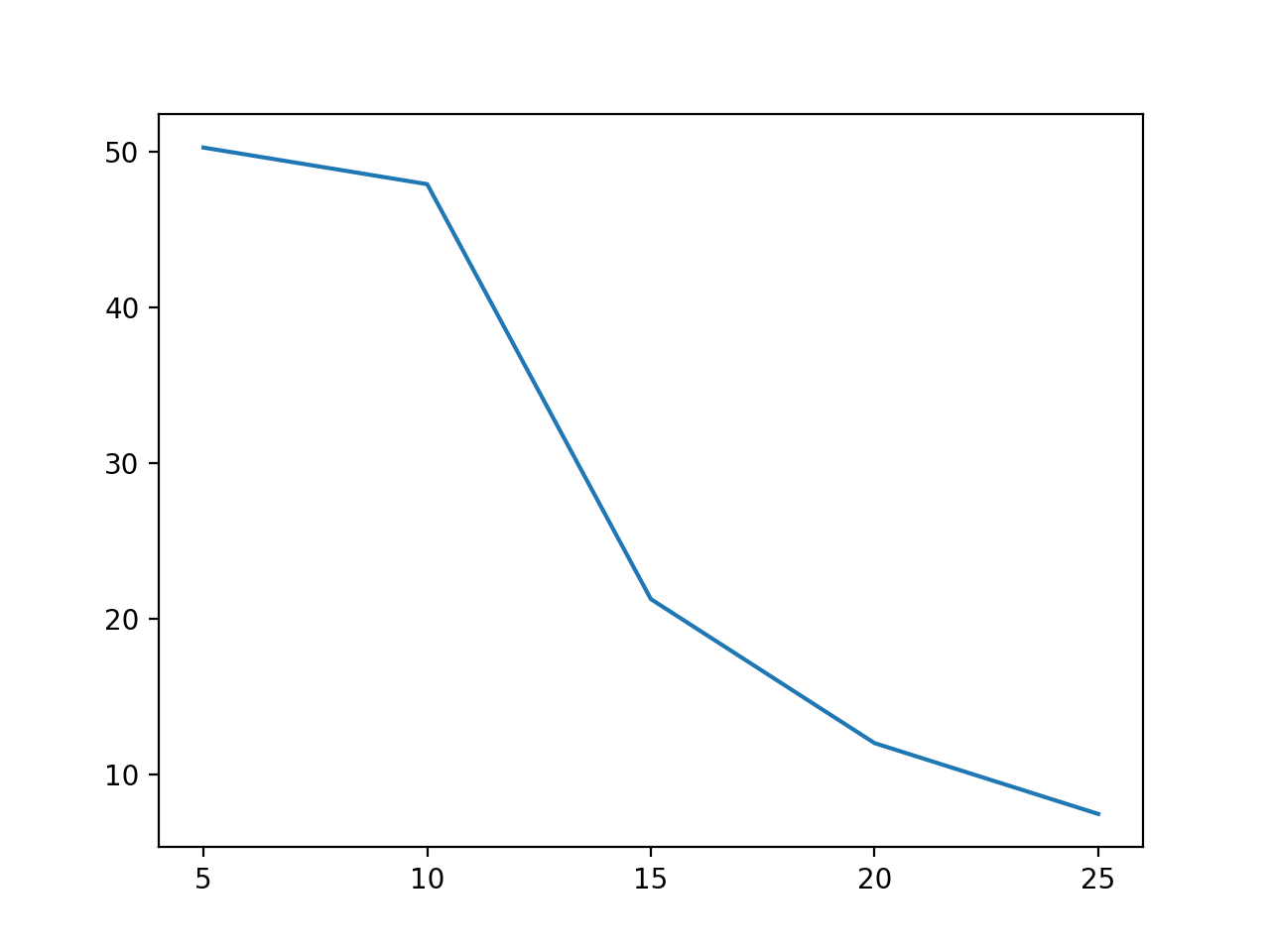}
		\caption{\(T = 400 \) and \(p_{\min} = 0.5\).}
		\label{fig:K_to_cluster_loss_pmin_05}
	\end{subfigure}
	\caption{Expected clustering error \( \E d(\hat{\CC}, \CC^*) \) for optimal \( \lambda \), \( N = 100\), and \( K = 5,...,25 \).}
	\label{fig:K_to_cluster_loss}
	\qlet 
	\href{https://github.com/QuantLet/SoNIC/tree/master/SoNIC_simulation_study}{SoNIC\_simulation\_study}
\end{figure}

\begin{figure}[t]
	\begin{subfigure}[b]{0.5\textwidth}
		\centering
		\includegraphics[width=1.1\textwidth]{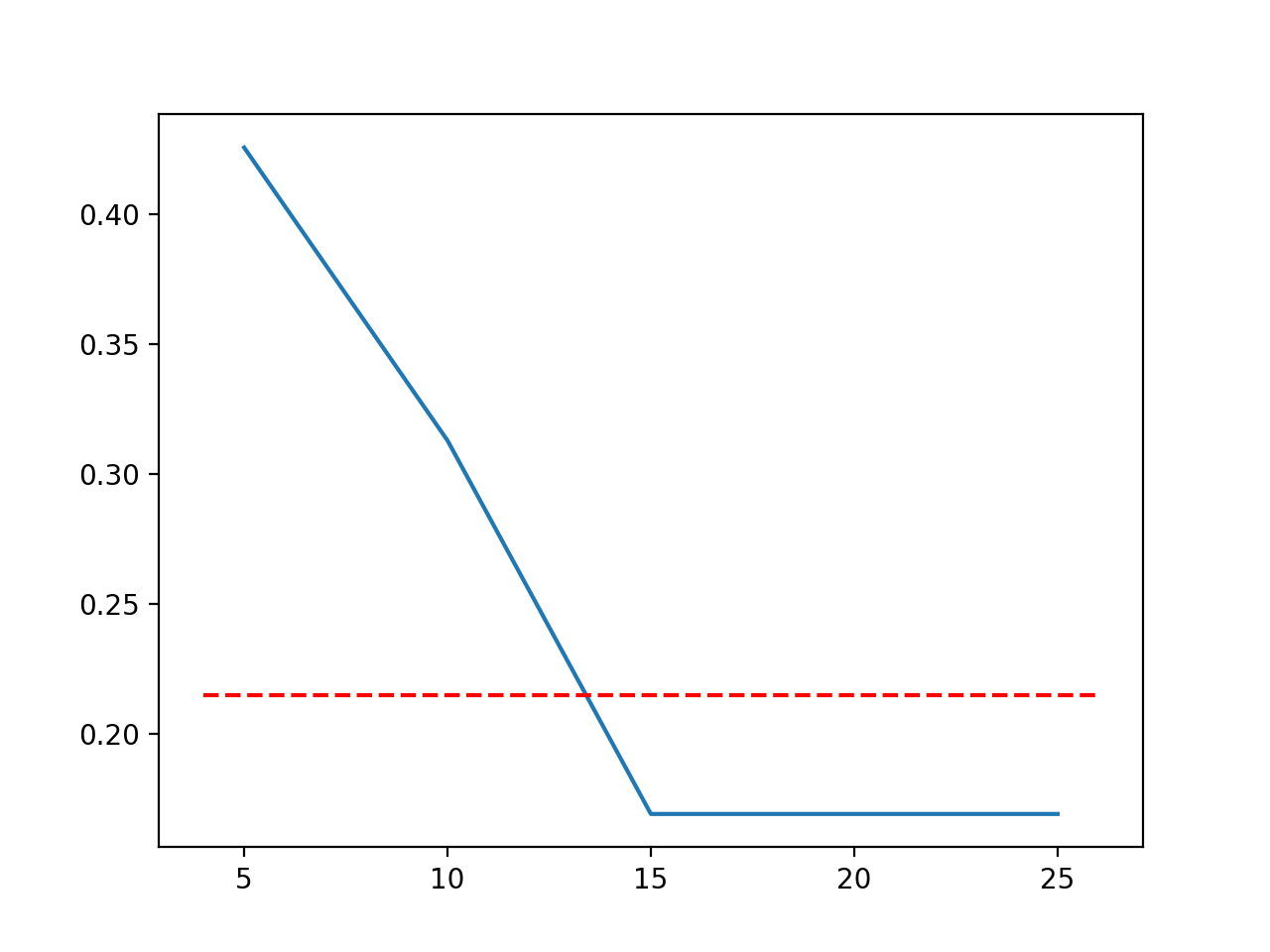}
		\caption{\( T = 100\) and  \(p_{\min} = 1 \).}
		\label{fig:alpha_opt_pmin_1}
	\end{subfigure}
	~
	\begin{subfigure}[b]{0.5\textwidth}
		\centering
		\includegraphics[width=1.1\textwidth]{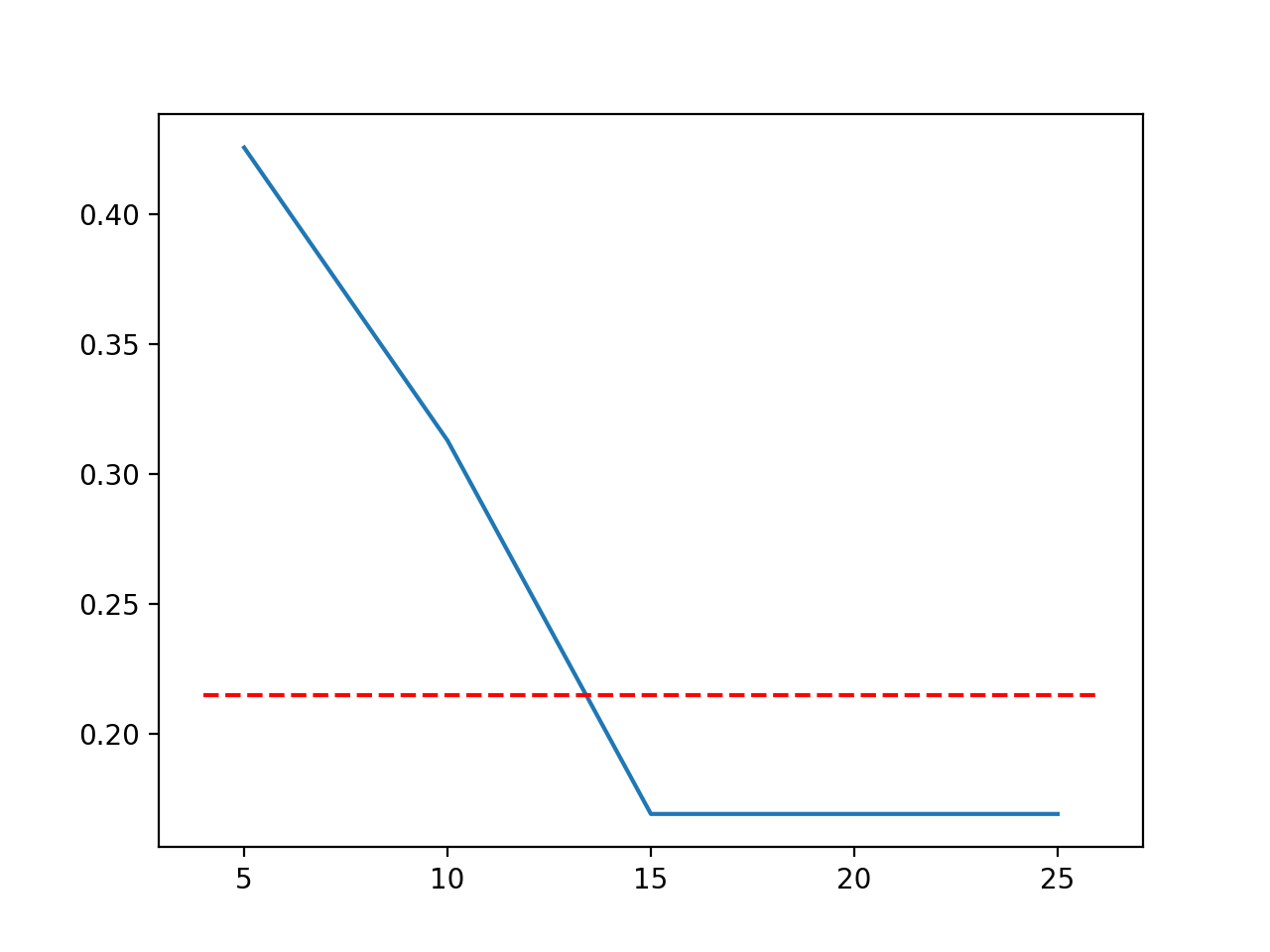}
		\caption{\(T = 400 \) and \(p_{\min} = 0.5\).}
		\label{fig:alpha_opt_pmin_05}
	\end{subfigure}
	\caption{The optimal value of \( \lambda \) for \( N = 100\) and \( K = 5,...,25 \). The red line corresponds to the value \( \lambda = \sqrt{\frac{\log N}{T p_{\min}^2}} \).}
	\label{fig:alpha_opt}
	\qlet 
	\href{https://github.com/QuantLet/SoNIC/tree/master/SoNIC_simulation_study}{SoNIC\_simulation\_study}
\end{figure}

\subsection{Choice of the regularization parameter \( \lambda \)}\label{choice_of_lambda}


It is often suggested to use the regularisation \( \lambda = \sigma \sqrt{{\log N} / T} \) in the LASSO literature, where \( \sigma \) stands for the noise level \citep{belloni2013least, van2008high, bickel2009simultaneous, van2014asymptotically}. In the example above, we have \( \sigma = 1 \). In our case of missing observations, the value \( T \) must be replaced by \( T p_{\min}^{2} \), the effective number of observations. Furthermore, \cite{wang2018high} recommend to disregard multiplicative constants that appear in theory in front of \( \sigma \sqrt{{\log N} / (Tp_{\min}^{2})} \) (see equation \eqref{lambda_opt_theo}) since it leads to consistent, but rather conservative estimation. 

The simulation results support this choice. Let us take a look at the regularisation paths in Figure~\ref{fig:alpha_to_loss} for different values of \( K \). All of the graphs that we show exhibit similar behavior: with \( \lambda \) increasing, the evaluated expected relative loss drops until it reaches its minimum, then it starts to increase until it reaches the constant value that corresponds to \( \hat{\Theta}_{\lambda} = 0 \), which obviously happens once the regularization is big enough. Typically, the ``oracle'' choice corresponds to the minimizer of the expected loss \( \E \trinorm \hat{\Theta}_{\lambda} - \Theta^{*} \trinorm_{\Frob} \). In order to compare it with the recommended choice above, for each \( K = 5,...,25 \), we pick the tuning parameter (among the available choices on the graph) that delivers the minimum to the evaluated expected loss. In Figure~\ref{fig:alpha_opt} we show the values of the best \(\lambda\) for each \( K = 5,...,25 \) (blue line) and compare it to the heuristic value \( \sqrt{\frac{\log N}{T p_{\min}^2}} \) (red line). We observe that once the number of clusters is large enough (\( K \geq 15 \)), the corresponding optimal choice of $\lambda$ approximately equals to \( \sqrt{\frac{\log N}{T p_{\min}^2}} \). On the other hand, as the graph in Figure~\ref{fig:K_to_cluster_loss} suggests, for \( K \leq 10 \) the number of nodes assigned to a wrong cluster grows significantly, and one cannot estimate the model with any given regularization parameter.

\begin{remark}\label{lambda_choice}
In practice, one must evaluate the noise level \( \sigma \) in a data-driven way \citep{belloni2013least}. We suggest to evaluate it using the spectrum of the covariance estimator \( \Sigmah \). One obvious choice can be \( \hat{\sigma} = \| \Sigmah \| \). However, this may lead to an overestimated noise level. We suggest using the following strategy. Since \( \Sigma = \Theta^{*} \Sigma (\Theta^{*})^{\T} + S \), we expect the original covariance to have either \( K \) or \( K -1 \) spikes (one cluster could be zero). In particular, this is true whenever \( S = \sigma I \). We therefore suggest using the singular value
\(
	\hat{\sigma} = \sigma_{K} (\Sigmah)
\),
which means that we skip the first \( K -1  \) components. The resulting regularisation parameter reads as
\[
	\lambda = \sigma_{K} (\Sigmah) \sqrt{\frac{\log N}{T p_{\min}^{2}}}.
\]
In the next section, we stick to this strategy.
\end{remark}

\subsection{Choice of number of clusters \(K\) via stability analysis}\label{choice_of_K}

In the simulation study above we fixed a priori the number of clusters. When applying SONIC to empirical data, this is rarely the case. One possible way to decide the number \( K\) is to analyze the \emph{stability} of the clustering algorithm \citep{rakhlin2007stability, le2018notion}. The idea is that if we guess the number of clusters correctly, then on different subsamples we should get similar results. On the other hand, if our guess is wrong, we can end up with randomly split or glued clusters. In other words, the resulting clustering will be unstable with respect to the change of the sample. We therefore propose the following procedure. Consider a sequence of intervals \( I_{1}, \dots, I_{l} \subset \{ 1, \dots, T \} \) of the same length and let us estimate the clusterings \( \hat{\CC}_{I_j} \) using the observations \(  (Y_{t})_{t \in I_{j}} \) for each \( j = 1, \dots, l \). If the number of clusters is correct, we expect that the pairwise distances \( \hat{\CC}_{j} \) are small. We take \( l = 6 \) intervals of length \( 3 T  /4 \pm 1 \), each of the form 
\begin{equation}\label{window_stability}
I_{j} = \left[\frac{j-1}{20} T + 1, \frac{j + 14}{20}  T  \right],
\qquad
j = 1, \dots, 6,
\end{equation}
so that we include all available observations. We then calculate the distances \( d(\hat{\CC}_{I_{1}}, \hat{\CC}_{I_{j}}  ) \) for each \( j = 2, \dots, l \) and for different choices of \( K \). We suggest to choose the number of clusters that has small distances \( d(\hat{\CC}_{I_{1}}, \hat{\CC}_{I_{j}}  ) \) when compared to the total number of nodes in the network.

We demonstrate how the picture can look in the following simulation scenarios:
\begin{itemize}
	\item[(a)] \( N = 100\), \( K = 2\), \( p_{\min} = 1 \), and \( T = 100, 200, 500, 1000, 2000 \);
	\item[(b)] \( N = 100\), \( K = 2\), \( p_{\min} = 0.5 \), and \( T = 100, 200, 500, 1000, 2000 \);
	\item[(c)] \( N = 100\), \( K = 5\), \( p_{\min} = 1 \), and \( T = 100, 200, 500, 1000, 2000 \);
	\item[(d)] \( N = 100\), \( K = 5\), \( p_{\min} = 0.5 \), and \( T = 100, 200, 500, 1000, 2000 \).
\end{itemize}
On Figure~\ref{fig:simustab} we present the results obtained from one realisation for each scenario. Each graph (a)-(d) contains the corresponding scenario, with \( T = 100, 200, 500, 1000, 2000\) marked with different colors from left to right. In Figure~\ref{fig:simustab_K2_pmin_1}, in the case where the true number of clusters is \( K = 2 \) and \( p_{\min} = 1 \), at first we do not see any stability. Although, the clustering errors corresponding to the correct guess \( K = 2 \) may be smaller, they are still rather large when compared to the total number of nodes. Only for \( T = 2000 \) the clustering distances become small (up to 4), and we can clearly see that there is only two clusters. Figure~\ref{fig:simustab_K2_pmin_05} shows the results for \( K = 2\) and \( p_{\min} = 0.5 \). Since the effective number of observations is \( T p_{\min}^2 \), the considered numbers of observations are not enough in this case. Figure~\ref{fig:simustab_K5_pmin_1} shows the results for \( K = 5 \) and \( p_{\min} = 1 \). Here, we can see stable estimation of the clustering for \( T = 1000, 2000 \), which means that it requires twice as smaller the observations than in the case \( K = 2 \). Notice that in the case \( T = 2000 \), the correct case \( K = 5 \) shows the smallest distance between clusterings obtained from different windows. For \( K = 6 \) it is still rather small, but the choice is incorrect. Figure~\ref{fig:simustab_K5_pmin_05} shows the results for \( K = 5\) and \( p_{\min} = 0.5 \). Effectively, the number of observations reduces by four times, and we can see the similarity between the graph for \( T = 2000 \) and for \( T = 500 \) in Figure~\ref{fig:simustab_K5_pmin_1}, as well as somewhat resemblance between \( T = 1000 \) in Figure~\ref{fig:simustab_K5_pmin_05} and \(T = 200 \) in Figure~\ref{fig:simustab_K5_pmin_1}. We can see that none of the graphs in Figure~\ref{fig:simustab_K5_pmin_05} demonstrates stability due to the lack of simulated observations.

In conclusion, we suggest to look for the smallest number of clusters that shows a ``reasonably'' small clustering difference for different windows, in the sense that it is much smaller than the total amount of nodes. However, at this point we are not able to provide any statistical explanation of what is a ``reasonable'' clustering distance. The stability analysis we suggest should be used as a qualitative heuristic.

\begin{figure}[t]
	\centering
	\begin{subfigure}{\textwidth}
		\centering
		\includegraphics[width=\textwidth]{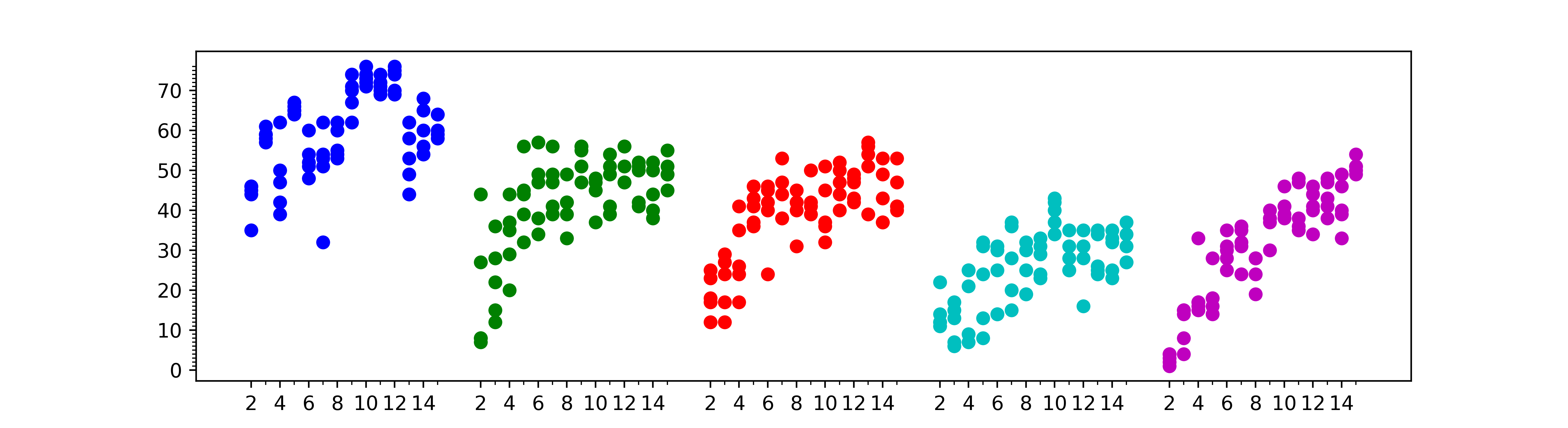}
		\caption{\(N = 100 \), \(K = 2\), and \(p_{\min} = 1\).}
		\label{fig:simustab_K2_pmin_1}
	\end{subfigure}
	\begin{subfigure}{\textwidth}
		\centering
		\includegraphics[width=\textwidth]{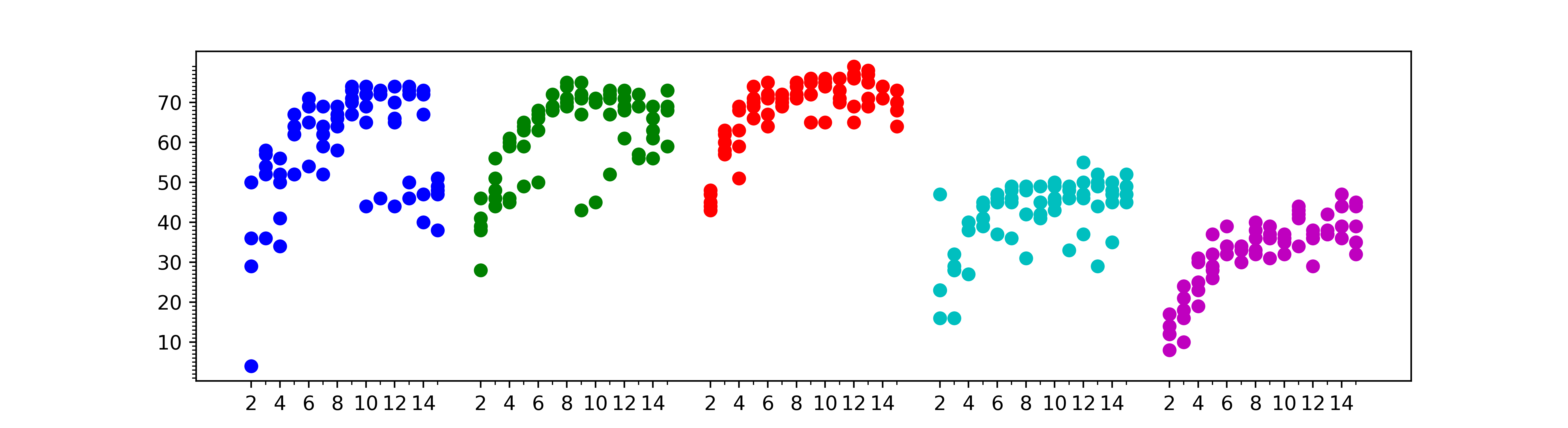}
		\caption{\(N = 100 \), \(K = 2\), and \(p_{\min} = 0.5\).}
		\label{fig:simustab_K2_pmin_05}
	\end{subfigure}
	\begin{subfigure}{\textwidth}
		\centering
		\includegraphics[width=\textwidth]{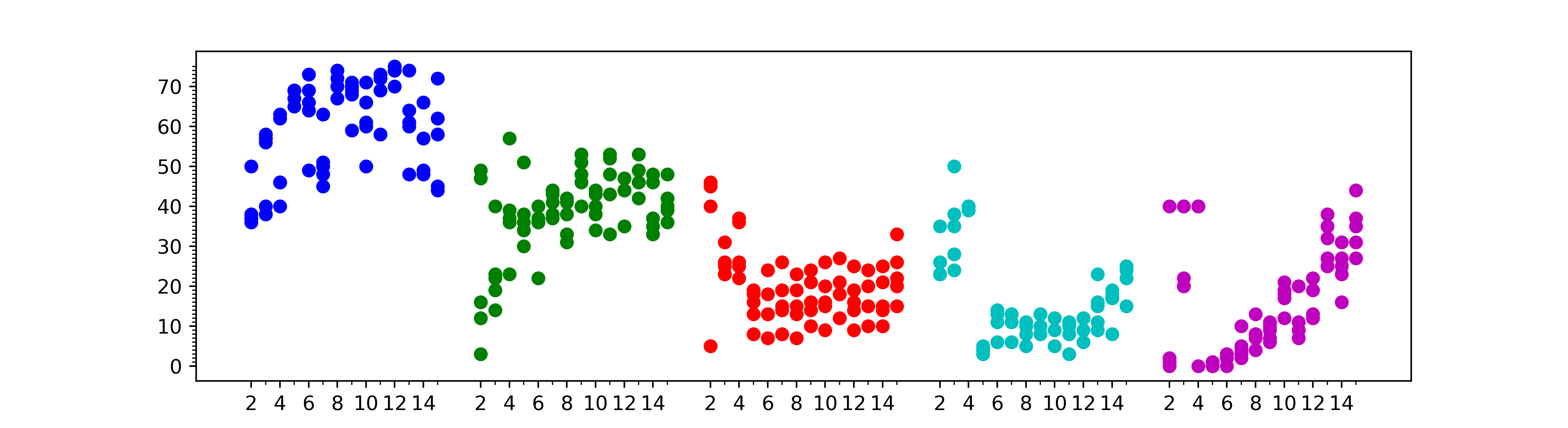}
		\caption{\(N = 100 \), \(K = 5\), and \(p_{\min} = 1\).}
		\label{fig:simustab_K5_pmin_1}
	\end{subfigure}
	\begin{subfigure}{\textwidth}
		\centering
		\includegraphics[width=\textwidth]{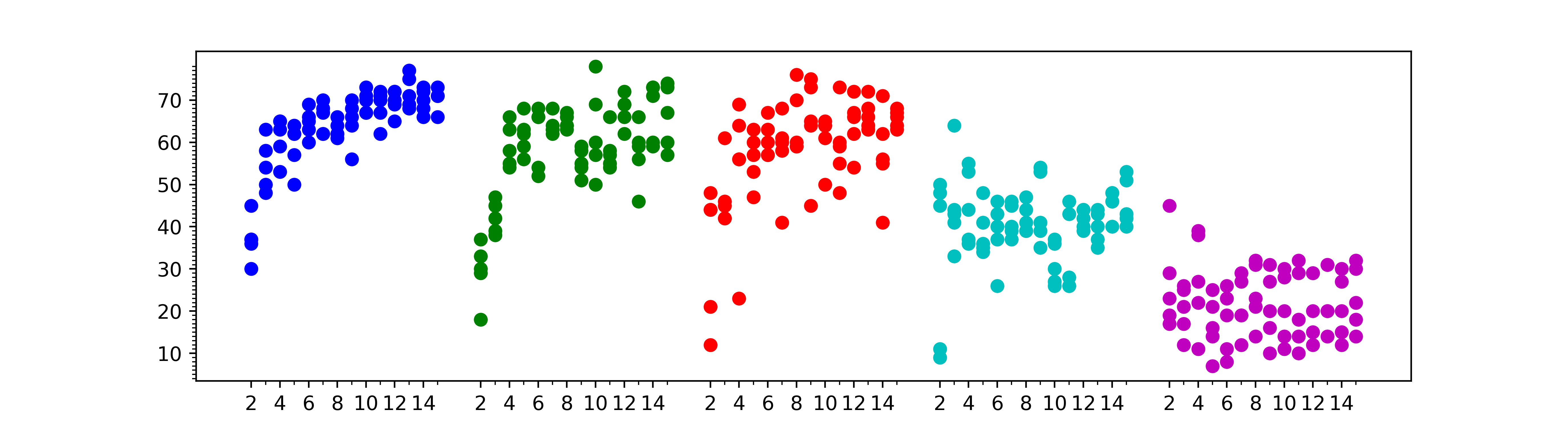}
		\caption{\(N = 100 \), \(K = 5\) and \(p_{\min} = 0.5\).}
		\label{fig:simustab_K5_pmin_05}
	\end{subfigure}
	\caption{Analysis of stability for four different scenarios and \( T = \textcolor{blue}{100}, \textcolor{dark-green}{200}, \textcolor{red}{500}, \textcolor{cyan}{1000}, \textcolor{magenta}{2000}\). For each \(T\), each point represents one of the five distances \( d(\hat{C}_{I_1}, \hat{C}_{I_j})\), where the clusterings are estimated based on the moving window \eqref{window_stability}. On the \(x\)-axis we have different guesses for the number of clusters {\(2,...,15\)}, the \(y\)-axis represent the clustering distance.}
	\label{fig:simustab}
	\qlet 
	\href{https://github.com/QuantLet/SoNIC/tree/master/SoNIC_stability_simulation}{SoNIC\_stability\_simulation}
\end{figure}


\section{Application to StockTwits sentiment}\label{section:application}
Here we present the applicability of SONIC to the dataset described in Section~\ref{section:stocktwits}. We look at the two (AAPL and BTC)  networks comprising of users' sentiment time series. These two symbols, representing the most popular security and cryptocurrency respectively, may reveal disparate characteristics, thereby distinct network dynamics featured with different communities and influencers.

To ensure that the model is applicable in the real world, we require that the observations are persistent with the same probability \( p_{i} \) over the considered time period. Moreover, since in Theorems~\ref{missing_cov_est:prop} and~\ref{cros_cov_missing:prop} the amount of observations scales with the factor \( p_{\min}^{2} \), we need to avoid the users whose \(p_i\) is too small. We propose the following criteria in sample selection to account for missing observations. 

\begin{enumerate}
	\item pick users with estimated probability \( \hat{p}_{i} \geq 0.6 \) for BTC and \( \hat{p}_{i} \geq 0.8 \) for AAPL to reflect the fact that the message volume of AAPL per day doubles that of BTC in Table \ref{tab:sum_stat}; 
	\item select the most extended historical interval over which the user exhibits persistent probability of observation. One can look at a moving average estimation and ensure that for any window it remains within the appropriate confidence interval;
	\item take only the users whose historical interval from step 2 is at least \(50\) weeks.
\end{enumerate}
Equipped with these criteria, for the AAPL dataset, we are left with 36 users and 82 weeks, while for BTC, we have 53 users and 78 weeks. Note that  concerning missing observations and the presence of outliers or noisy, the weekly sentiment series averaging out the daily sentiment series is employed. 

\def\stuser#1{\href{http://stocktwits.com/#1}{#1}}
We apply our SONIC model to the AAPL dataset. We set  \( \lambda = 0.08 \)  according to Section~\ref{choice_of_lambda} and Remark~\ref{lambda_choice}. As for the number of clusters, we perform the analysis described in Section~\ref{choice_of_K} and present the results in Figure~\ref{fig:aapl_K_adapt} for \( K = 2, 3, 4, 5, 6\). Based on these results we suggest to pick \( K = 2\) with maximum clustering distance \( 3\) out of \(36\) users in total. We present a heatmap visualization for the estimated matrix \( \hat{\Theta} \) in Figure~\ref{fig:aapl_theta}, where we identify the candidates of influencers with the identification numbers  \stuser{619769}, \stuser{850976}, \stuser{5}, \stuser{962572}, \stuser{526780}, \stuser{473512}. To parallel our identification with the indicators from social conventions in terms of what ought to possess as influencers e.g. the number of followers, we analyze the social network profiles of selected users including the register date of membership,  the number of followers, the number of ideas, liked count, etc. 

To retrieve users' social profiles, we use the StockTwits API toolkit to request the users' message streams and \rev{profiles}.
We stratify the retrieved data and particularly focus on the number of followers, the number of ideas, and the liked count, in hopes of these selected characteristics to comply with the social consensus in terms of the notion of influencers. Table~\ref{tab:user} summarizes influencers' social profile and reports the corresponding percentile rank among a pool of users. 

The identified users appear to either attract many followers or behave actively, provided with tremendous ideas (posts) or liked count. The first three influencers represent the trading companies offering technical and fundamental analysis for the symbols of interest. It shows that investment companies or financial industry entrepreneurs target their potential customers appearing on social media and influence them strategically. The latter two are financial analysts or trading consultants, and they may serve for a small group of users. 

As to the BTC dataset, applying the proposed strategy, we end up with \( \lambda = 0.21 \) and, using the results in Figure~\ref{fig:btc_K_adapt}, we choose \( K = 2 \). Figure~\ref{fig:btc_theta} displays the estimated matrix \( \hat{\Theta} \) and identifies the influencers \stuser{398367} and \stuser{969971}. Likewise, we elicit their social profile data and document the relevant features in Table~\ref{tab:user}. 
The first one is an investment company with a specialization on crypto assets, while the second one is a crypto specialist updating price information and producing the technical analytics to cryptocurrency traders. Both broadcast tactical trading information and update these frequently.

We notice that for anyone relying on these social characteristics may oversimplify the task of identifying influencers. One should be aware that some users with much more followers or ideas may not be able to surpass those being identified via our approach in terms of opinions' importance. In the case of BTC, those who have specialized themselves in crypto-assets may lend themselves to serve a relatively smaller group of people with specific trading preference, albeit not attracting granular followers.


\begin{table}[htp!]
	\begin{center}
		\begin{tabularx}{\linewidth}{@{\extracolsep{\fill}}rrrr}
			\hline\hline
			\textit{user ID}                                      & Followers & Ideas  & like count \\
			 \hline
			\multicolumn{4}{c}{\multirow{1}{*}{\textbf{AAPL}}} \\
			\hline
			\quad 619769         & 9,962 (0.69)   &  27,729 (0.53) & 7,041 (0.75) \\
			\quad 850976        & 37,426 (0.81)   &  42,817 (0.70) & 43,639 (0.94) \\
			\quad 5                & 225,575 (0.97)   &  166,591 (0.92) & 79,568 (1.00) \\
			\quad 962572        & 46,455 (0.86)   &  104,564 (0.81) & 1,897 (0.73) \\
			\quad 526780        & 823 (0.42)   &  3,942 (0.10) &1,473 (0.47) \\
			\quad 473512        & 306 (0.28)   &  25,853 (0.42) & 992 (0.39) \\
			\hline
			\multicolumn{4}{c}{\multirow{1}{*}{\textbf{BTC}}} \\
			\hline
			\quad 398367         & 232 (0.53)   &  11,852 (0.60) & 2,506 (0.43) \\
			\quad 969971        &  345 (0.66)   &  11,135 (0.59) & 111,423(1.00) \\
			\hline \hline
		\end{tabularx}
	\end{center}
	\caption{Influencers' metadata}\label{tab:user}
	\small
	\footnotesize\parindent=2em We report the number of followers, the number of ideas and the like count tagged to each specific user ID. The value in parenthesis is the corresponding percentile rank among a pool of users. 
\end{table}

\begin{figure}
	\centering
	\begin{subfigure}[b]{\textwidth}
		\centering
		\includegraphics[width=0.85\textwidth]{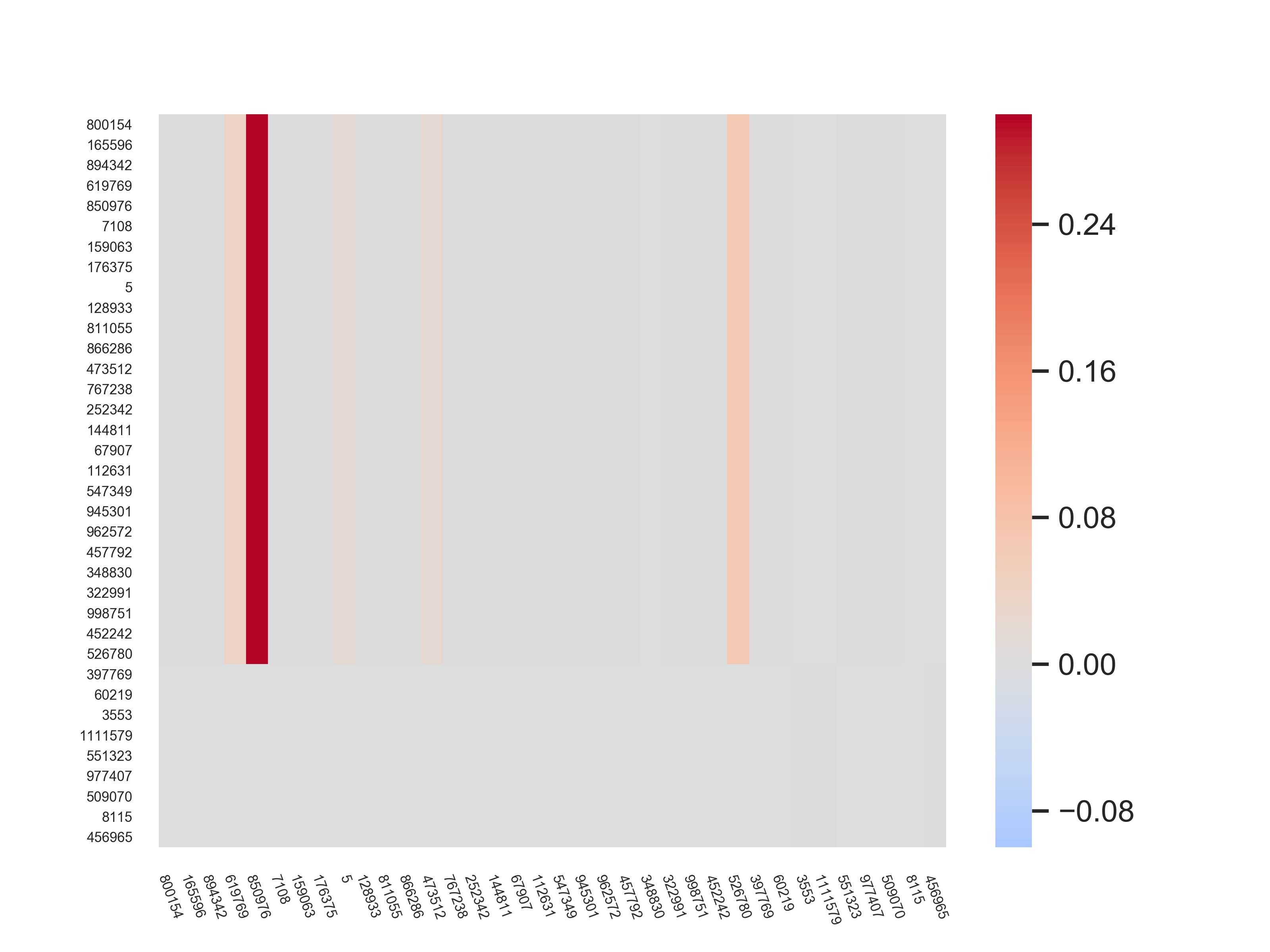}
		\caption{AAPL dataset with \( N = 36 \), \( T = 82 \) and \( K = 2 \).}
		\label{fig:aapl_theta}
	\end{subfigure}
	\hfill
	\begin{subfigure}[b]{\textwidth}
		\centering
		\includegraphics[width=0.85\textwidth]{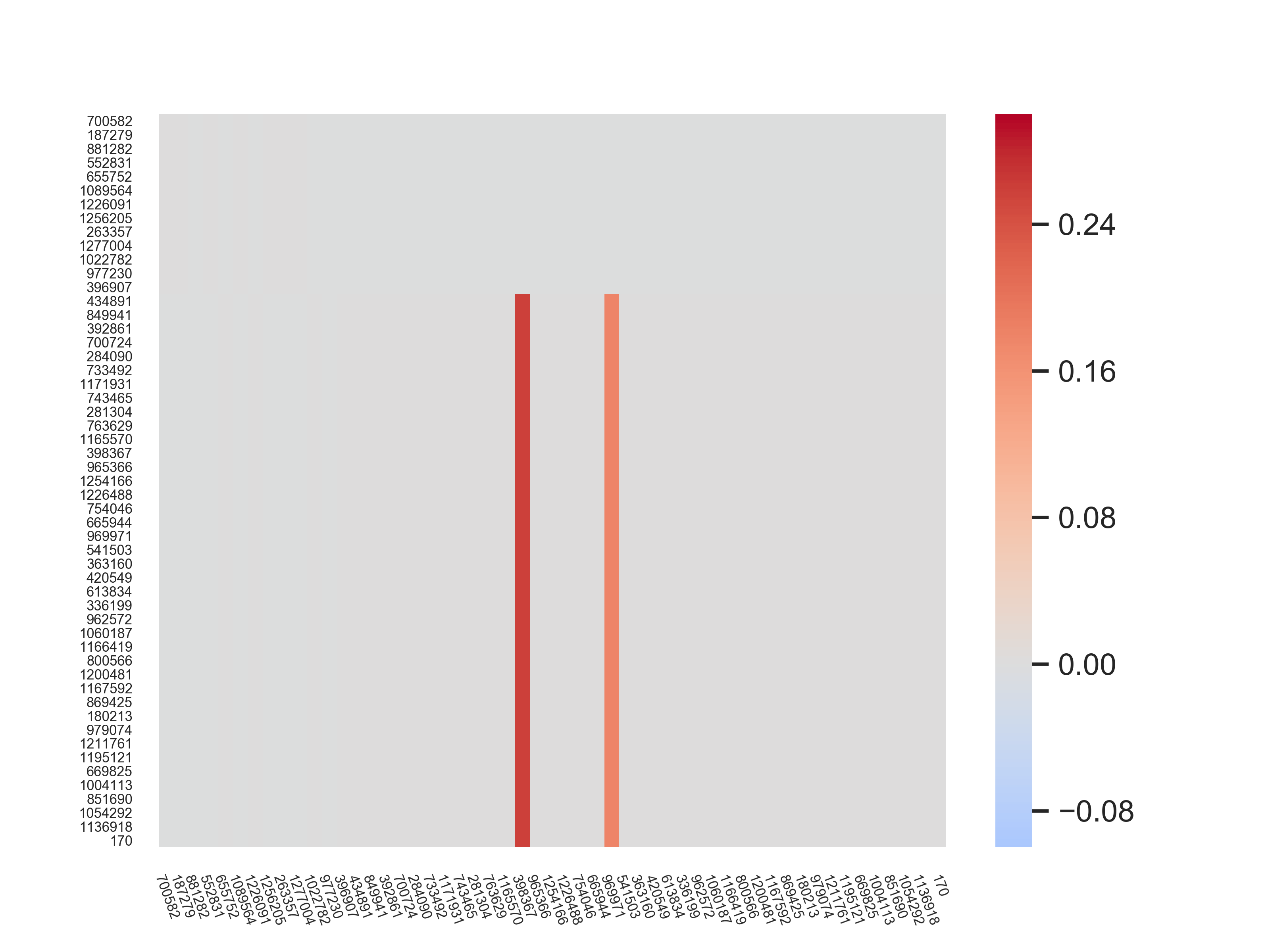}
		\caption{BTC dataset with \( N = 53 \), \( T = 78 \) and \( K = 2 \).}
		\label{fig:btc_theta}
	\end{subfigure}
	\label{fig:application}
	\caption{Estimated \( \hat{\Theta} \) for AAPL and BTC datasets. The axes correspond to users' id's and are rearranged with respect to the estimated clusterings. }
	\qlet
	\href{https://github.com/QuantLet/SoNIC/tree/master/SoNIC_AAPL_BTC}{SoNIC\_AAPL\_BTC}
\end{figure}

\begin{figure}
	\centering
	\begin{subfigure}[b]{0.8\textwidth}
		\centering
		\includegraphics[width=0.8\textwidth]{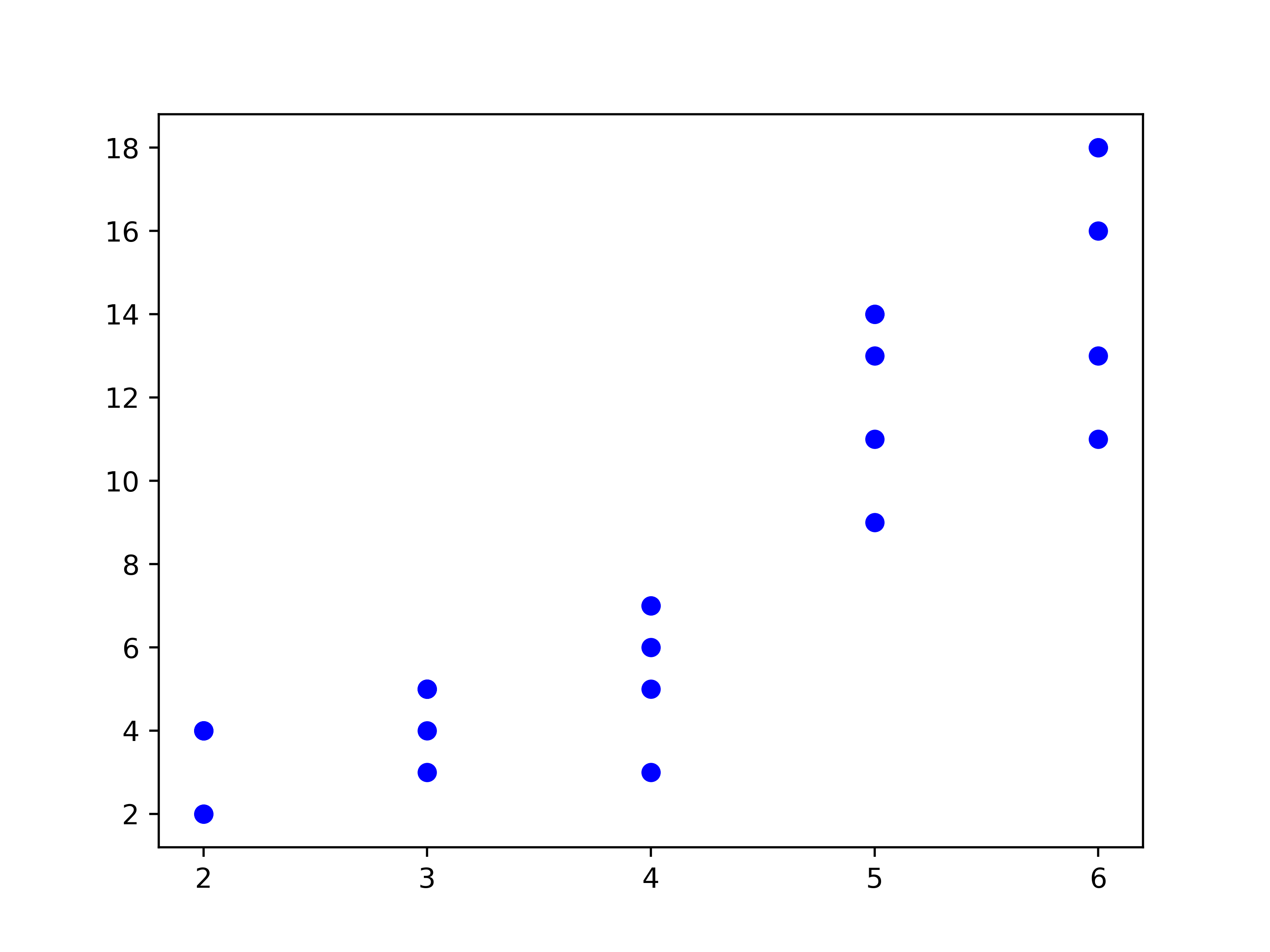}
		\caption{AAPL dataset}\label{fig:aapl_K_adapt}
	\end{subfigure}
	\hfill
	\begin{subfigure}[b]{0.8\textwidth}
		\centering
		\includegraphics[width=0.8\textwidth]{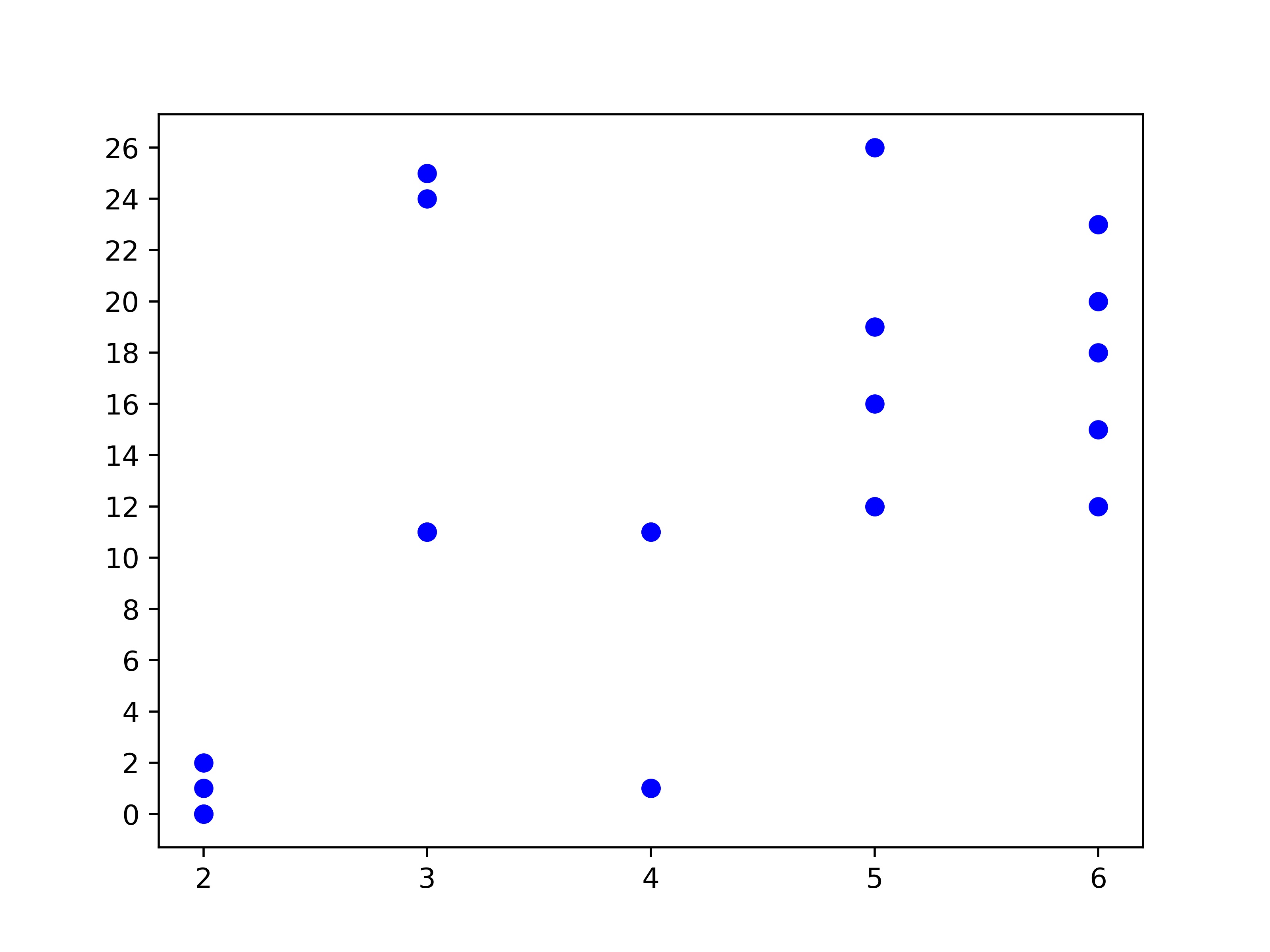}
		\caption{BTC dataset}\label{fig:btc_K_adapt}
	\end{subfigure}
	\caption{Stability analysis for AAPL and BTC datasets. For each \(K = 2, 3, 4, 5, 6\) on the $x$-axis, we plot five points corresponding to the clustering distances \( d(\hat{\CC}_{I_1}, \hat{\CC}_{I_j}) \) on the $y$-axis, with sampling windows as in \eqref{window_stability}.}\label{fig:K_adapt}
	\qlet
	\href{https://github.com/QuantLet/SoNIC/tree/master/SoNIC_AAPL_BTC_stability}{SoNIC\_AAPL\_BTC\_stability}
\end{figure}


\subsubsection*{Prediction performance compared with other methods}

To highlight the advantages of the proposed model, we compare the prediction accuracy of our method with other benchmarks. We consider the following prevalent benchmarks, each of which considers missing observations.
\begin{itemize}
	\item VAR with missing observations:
	\[
		\hat{\Theta} = \arg\min_{\Theta \in \R^{N\times N}} \frac{1}{2} \Tr(\Theta \hat{\Sigma} \Theta^{\T}) - \Tr(\Theta \hat{A}),
	\]
	where \( \hat{\Sigma}\) and \( \hat{A} \) are the covariance and cross-covariance estimators, respectively (recall the definition from Section~\ref{section:missing});
	\item Lasso VAR with missing observations:
	\[
		\hat{\Theta} = \arg\min_{\Theta \in \R^{N \times N}} \frac{1}{2} \Tr(\Theta \hat{\Sigma} \Theta^{\T}) - \Tr(\Theta \hat{A}) + \lambda \| \Theta \|_{1, 1},
	\]
	where we choose the same \( \lambda \) as in SONIC;
	\item Constant estimator \( \Theta = 0 \), which corresponds to \rev{no correlation across time}.
\end{itemize}
In our exercise, we split the available sample ---{82 weeks for AAPL and 78 weeks for BTC --- into the train and test subsamples, approximately \( 70\%\) to \( 30\% \). Measuring prediction error on data with missing observations, we stumble into the same problem. Ideally, we want to access the value,
\[
	\frac{1}{T_{test} - 1} \sum_{t = 2}^{T_{test}} \| Y_{t} - \hat{\Theta} Y_{t-1} \|^{2},
\]
where \( T_{test}\) is the number of observations in the test sample and \( \hat{\Theta} \) is estimated on the train sample. Observe that (similar to \eqref{risk_no_missing}),
\begin{align*}
	\frac{1}{T_{test} - 1} \sum_{t = 2}^{T_{test}} \| Y_{t} - \hat{\Theta} Y_{t-1} \|^{2} =& \Tr\left(\frac{1}{T_{test} - 1} \sum_{t = 2}^{T_{test}} Y_t Y_t^{\T} - 2 \hat{\Theta} \left[\frac{1}{T_{test} - 1} \sum_{t = 2}^{T_{test}} Y_{t-1} Y_t^{\T} \right] \right. \\
	& + \left.  \Theta \left[\frac{1}{T_{test} - 1} \sum_{t = 1}^{T_{test}-1} Y_{t} Y_t^{\T} \right] \Theta^{\T}   \right) \, ,
\end{align*}
and we suggest to replace \( \tfrac{1}{T_{test} - 1} \sum_{t = 2}^{T_{test}} Y_{t} Y_t^{\T} \) and \( \tfrac{1}{T_{test} - 1} \sum_{t = 2}^{T_{test}} Y_{t-1} Y_t^{\T}  \) with \( \hat{\Sigma}_{test} \) and \( \hat{A}_{test} \), respectively, which are the covariance and cross-covariance estimated from the test sample. To sum up, we evaluate the prediction performance by,
\[
	\Tr(\hat{\Sigma}) - 2 \Tr(\Theta\hat{A}) + \Tr(\Theta \hat{\Sigma} \Theta^{\T}) \,.
\]

The results are presented in Table~\ref{tab:comparison}. We find that, in terms of the prediction performance, SONIC is slightly better that the sparse VAR for the AAPL dataset and as good as the sparse VAR for the BTC dataset. The regular VAR blows up in both cases, which is not surprising given the dimension and the sample sizes in each case. The similarity of the SONIC and sparse VAR shows that the number of clusters \(K = 2\) is too small to benefit from our model in terms of performance. Notice that the condition \eqref{n_star_cond} of Theorem~\ref{main_thm} is likely to break for small number of clusters. However, the fact that SONIC is not worse than sparse VAR confirms that the model we propose indeed reflects the dynamics of a real sentiment based network. In addition, we compare the results with the constant estimator \( \hat{\Theta} = 0 \), which corresponds to the no causality case. We see that in both cases the loss is higher than that of the SONIC model.

\begin{table}[htp!]
	\begin{center}
		\begin{tabularx}{0.6\linewidth}{@{\extracolsep{\fill}}rrr}
			\hline\hline
			& {\bf AAPL} & {\bf BTC}   \\
			\hline 
			\quad SONIC        & 2.609   &  4.332  \\ 
			\quad VAR        & 1.302 \(\times 10^{27}\)   &  4.996 \(\times 10^{27}\)  \\ 
			\quad Sparse VAR        & 2.659   &  4.332  \\ 
			\quad \( \Theta = 0\) (no causality)        & 5.719   &  8.995  \\ 
			\hline \hline
		\end{tabularx}
	\end{center}
	\caption{Prediction error for SONIC and alternative methods.}
	\label{tab:comparison}
	\qlet
	\href{https://github.com/klochkov123/SoNIC/tree/master/SoNIC_AAPL_BTC_benchmark}{SoNIC\_AAPL\_BTC\_benchmark} 
\end{table}

\par

\section{{Conclusion}}\label{sec:conclusion}

Nowadays the interest in dynamics of interaction among the users emerging in social media is dramatically growing. Social media become an attractive venue where users can easily and instantly interact with others. The research in this strand is, however, challenging. From an econometric point of view,  these dynamics require effective state-of-the-art methodologies that cope with the curse of dimensionality, as well as characterize psychological interdependence. From a quantitative perspective, with the textual analysis, the text-based information distilled from Twitter or StockTwits social networks boils down to a numerical expression of sentiment or opinions.  The joint evolvement of sentiment variables from individuals constitutes a dynamic network with a possibly growing dimension. 

In order to cope with dimensionality in a limited observation setting, we propose SONIC (SOcial Network analysis with Influencers and Communities). SONIC characterizes the social network dynamics and interdependence featured with identified influencers and detectable communities. We provide and discuss several theoretical results on the asymptotic consistency of the dynamic network parameters, even when observations are missing. We propose an estimation procedure based on a greedy algorithm and LASSO regularization that we extensively test in simulations. 

Using StockTwits data and the lexicon-based sentiments, we deploy a SONIC analysis and display an opinion network for Apple and Bitcoin users (nodes). We detect \(K = 2\) communities using stability analysis and identify the influencers subsequently. We discuss the choice of the regularization parameter \(\lambda\) of LASSO and the choice of the number of clusters. 



\section{Proof of the main result}\label{section:proof}

This section is devoted to the proof of Theorem~\ref{main_thm}. We start with some preliminary lemmata and then proceed with the proof that consists of several steps. Following the ideas in \cite{gribonval2015sparse}, the proof relies on explicit representation of the loss function.

We exploit the following simplified notation. Denote, \( \zv_j^{*} = \zv_{C_j^*} \) to be the columns of \( Z^{*} = Z_{\CC^{*}} \) and we also denote \( n_j^* = |C_j^*| \) for every \( j = 1, \dots, K\). When the clustering \( \CC = (C_1, \dots, C_K) \) is clear from the context we will also write \( Z \) for \( Z_{\CC} \), \( \zv_{j} \) for \( \zv_{C_j} \), and \( n_j = |C_j| \) for every \( j = 1, \dots, K \).

\subsection{Preliminary lemmata}

\begin{lemma}\label{n_j_n_j_star:lem}
Suppose that \( C_j \) is such that \( \| \zv_{C_j} - \zv_j^* \| \leq 0.3 \). Then,
\[
	\frac{1}{1.1} |C_j^*| \leq |C_j|  \leq 1.1 |C_j^*| .
\]
\end{lemma}
\begin{proof}
Suppose, \( n_j = |C_j| > n_{j}^{*} = |C_j^*| \), then
\[
	r^{2} = \| \zv_j - \zv_j^{*} \|^{2} = 2 - \frac{2}{\sqrt{n_j n_j^{*}}} |C_j \cap C_j^{*} | \geq 2 - 2\sqrt{\frac{n^{*}_{j}}{n_{j}}} ,
\]
since \( |C_j \cap C_j^*| \leq n_j^* \). Thus, \( \sqrt{n_j} - \sqrt{n_j^*} \leq (r^2 / 2) \sqrt{n_j} \), which due to \( r \leq 0.3 \) implies by rearranging and taking square \( n_j \leq 1.1 n_j^{*} \).

If \( n_j < n_j^{*} \) we have,
\[
	r^{2} \geq \| \zv_j - \zv_{j}^{*} \|^{2} = 2 - \frac{2|C_j \cap C_j'|}{\sqrt{n_j n_j^{*}}} \geq 2 - 2 \sqrt{\frac{n_j}{n_j^{*}}},
\]
and the fact that \( r \leq 0.3 \) implies \( n_j^{*} \leq 1.1 n_j \).

\end{proof}

\begin{lemma}\label{l1_l2_comp:lem}
Let \( \| \zv_{C_1} - \zv_{C_2} \| \leq 0.3 \). Then,
\[
	\| \zv_{C_1} - \zv_{C_2} \|_{1} \leq 1.65 \sqrt{N_1} \| \zv_{C_1} - \zv_{C_2} \|^{2} \, .
\]
\end{lemma}
\begin{proof}
Let \( N_j = |C_j| \) and \( a = |C_1 \cap C_2| \), \( b = |C_1 \setminus C_2| \), \( c = |C_2 \setminus C_1| \), so that \( N_1 = a + b \), \( N_2 = a + c \), and \( |C_1 \triangle C_2| = b + c \). We have,
\[
	\| \zv_{C_1} - \zv_{C_2} \|^{2} = \left( \frac{1}{\sqrt{N_1}} - \frac{1}{\sqrt{N_2}}  \right)^{2} a + \frac{b}{N_1} + \frac{c}{N_2} \geq  \frac{b}{N_1} + \frac{c}{N_2} \, .
\]
On the other hand,
\begin{align*}
	\| \zv_{C_1} - \zv_{C_2} \|_{1} &= \left| \frac{1}{\sqrt{N_1}} - \frac{1}{\sqrt{N_2}}  \right| a + \frac{b}{\sqrt{N_1}} + \frac{c}{\sqrt{N_2}} \\
	& \leq
	\left| \frac{1}{\sqrt{N_1}} - \frac{1}{\sqrt{N_2}}  \right| a + \sqrt{N_1 \vee N_2} \| \zv_{C_1} - \zv_{C_2} \|^{2} \, .
\end{align*}
Since \( |N_1 - N_2| \leq b + c \) we obviously have,
\begin{align*}
	\left| \frac{1}{\sqrt{N_1}} - \frac{1}{\sqrt{N_2}}  \right| a & = \frac{|N_1 - N_2| a}{\sqrt{(a + b)(a + c)}(\sqrt{a + b} + \sqrt{a + c})} \\
	& \leq
	\frac{(b + c) a}{\sqrt{N_1 \vee N_2} \sqrt{a} (2\sqrt{a})} \\
	& \leq
	\sqrt{N_1 \vee N_2} \| \zv_{C_1} - \zv_{C_2} \|^{2} / 2 ,
\end{align*}
and it is left to apply Lemma~\ref{n_j_n_j_star:lem}.
\end{proof}

\begin{lemma}\label{l1_l2_z_comp:lem}
Suppose, \( \frac{\min_{j} n_{j}^*}{\max_{j} n_{j}^{*} } \geq \alpha \) for some \( \alpha \in (0,1]  \) and let \( \| \zv_{j} - \zv_{j}^{*} \| \leq r \). Suppose, \( r \leq 0.3 \). Then,
\[
	\| [Z^{*}]^{\T} (\zv_j - \zv_j^{*}) \|_{1} \leq 3.05 \alpha^{-1/2} r^{2} .
\]
\end{lemma}
\begin{proof}
1) We first consider the case \( |C_j| = n_j^{*} \). It holds then
\[
	[\zv_{j}^{*}]^{\T} (\zv_j^{*} - \zv_{j}) = \frac{1}{n_j^*} (n_j^{*} - |C_j \cap C_{j}^{*}|) = \frac{1}{n_{j}^{*}} |C_j^{*} \setminus C_j |  .
\]
Moreover, for every \( k \neq j \) it holds
\[
	|[\zv_{k}^{*}]^{\T} (\zv_j^{*} - \zv_{j})| = | [\zv_k^*]^{\T} \zv_j | = \frac{1}{\sqrt{n_k^* n_j^*}} |C_k^{*} \cap C_j | \leq \frac{\alpha^{-1/2}}{n_j^{*}} |C_k^{*} \cap C_j | .
\]
Summing up, we get
\begin{align*}
	\| [Z^{*}]^{\T} (\zv_j - \zv_j^{*}) \|_{1} & \leq \frac{\alpha^{-1/2}}{n_j^{*}} \left( |C_j^{*} \setminus C_j | + \sum_{k \neq j} |C_k^{*} \cap C_j | \right) \\
	& \leq \frac{\alpha^{-1/2}}{n_j^{*}} \left( |C_j^{*} \setminus C_j | + | C_j \setminus C_j^{*} | \right)  \\
	& = \frac{\alpha^{-1/2}}{n_j^{*}} |C_j \triangle C_j^{*}| .
\end{align*}
It is left to notice that in the case \( |C_j| = |C_j^*| = n_j^* \) we have exactly \( \| \zv_j - \zv_j^{*}\|^{2} = \frac{1}{n_j^*} |C_j \triangle C_j^{*}|  \).

2) Suppose, \( n_j = |C_j| > n_{j}^{*} \). Obviously, we can decompose \( C_{j} = C_{j}' \cup B \) such that \( |C_j'| = n_j^{*} \) and \( B \cap C_j^{*} = \emptyset \). Setting \( \zv_j' = \zv_{C_j'} \) we get by the above derivations that \( \| [Z^{*}]^{\T} (\zv_j' - \zv_j^{*}) \|_{1} \leq \alpha^{-1/2} \| \zv_j' - \zv_j^{*} \|^{2} \). Since \( C_j' \cap C_j^{*} = C_j \cap C_j^{*} \) we can compare the distances
\[
	\| \zv_j - \zv_j^{*} \|^{2} = 2 - \frac{2}{\sqrt{n_j n_j^{*}}} |C_j \cap C_j^{*} | > 2 - \frac{2}{n_j^{*}} |C_j \cap C_j^{*}| = \| \zv_j' - \zv_j^{*} \|^{2} .
\] 
Taking the remainder \( \bv = \zv_j - \zv_j' \) we have, 
\[
	b_i = \left\{
	\begin{aligned}
		&{n_j}^{-1/2} - (n_j^*)^{-1/2}, \qquad
		&i \in C_j', \\
		&{n_j}^{-1/2},
		\qquad & i \in B, \\
		& 0 \qquad & \text{otherwise}.
	\end{aligned}
	\right.
\]
Setting \( d = n_j - n_j^* = |B| \) it is easy to obtain \( |{n_j}^{-1/2} - (n_j^*)^{-1/2}| \leq \frac{d}{n_j} \frac{1}{\sqrt{n_j^*}} \). Thus, we get
\begin{align*}
	\sum_{k = 1}^{K} |[\zv_k^{*}]^{\T} \bv | & \leq \sum_{i = 1}^{k} \frac{1}{\sqrt{n_k^*}} \left( \frac{d}{n_j} \frac{1}{\sqrt{n_j^*}} |C_j' \cap C_k^*| + |B \cap C_k^{*}| \frac{1}{\sqrt{n_j}} \right)  \\
	&\leq \frac{\alpha^{-1/2} d}{n_j^* n_j} |C_j'| + \frac{\alpha^{-1/2}}{\sqrt{n_j^{*} n_j}} d \\
	& < \frac{2 \alpha^{-1/2} d}{\sqrt{n_j n_j^{*}}} .
\end{align*}
We show that the latter is at most \(2.05 \alpha^{-1/2} r^2\). Indeed, it is not hard to show that from \( n_j \leq 1.1 n_j^{*} \) (see Lemma~\ref{n_j_n_j_star:lem}) it follows
\[
	\frac{n_j - n_j^{*}}{\sqrt{n_j n_j^*}} \leq 2.05\left(1 - \frac{n_{j}^{*}}{\sqrt{n_j n_j^*}} \right) \leq 2.05 \times \frac{r^{2}}{2} ,
\]
thus \(	\| [Z^*]^{\T}(\zv_j - \zv_j^{*}) \|_{1} \leq 3.05 \alpha^{-1/2} r^{2} \) and the result follows.

3) The case \( n_j < n_j^{*} \) can be resolved similarly to the previous one. Since \( |C_{j}^{*} \setminus C_{j} | \geq n_{j}^{*} - n_{j} \) we can pick a subset \( B \subset C_{j}^{*} \setminus C_{j}  \) of size \( d = n_{j}^* - n_{j} \) and set \( C_{j}' = B \cup C_{j} \) with \( |C_j'| = n_{j}^{*} \); set also \( \zv_j' = \zv_{C_{j}'} \). Then, we have
\[
	\| \zv_j' - \zv_j^{*} \|^{2} = 2 - 2 \frac{|C_{j}' \cap C_{j}^{*}|}{n_{j}^{*}} \leq 2 - \frac{2|C_j \cap C_j'|}{\sqrt{n_j n_j^{*}}} = \| \zv_j - \zv_{j}^{*} \|^{2}.
\]
Thus, by the first part of this proof it holds
\(
	\| [Z^{*}]^{\T} (\zv_j' - \zv_j^{*}) \|_{1} \leq \alpha^{-1/2} r^{2}
\) . Setting \( \bv = \zv_{j}' - \zv_{j} \) we have,
\[
b_i = \left\{
\begin{aligned}
	& (n_j^*)^{-1/2} - {n_j}^{-1/2}, \qquad &i \in C_j, \\
	&{n_j^{*}}^{-1/2},                  \qquad & i \in B, \\
	& 0                             \qquad & \text{otherwise}.
\end{aligned}
\right.
\]
Since \( |{n_j}^{-1/2} - (n_j^*)^{-1/2}| \leq \frac{d}{n_j^{*}} \frac{1}{\sqrt{n_j}} \) we obtain,
\begin{align*}
	\sum_{k = 1}^{K} |[\zv_k^{*}]^{\T} \bv | 
	& \leq 
	\sum_{i = 1}^{k} \frac{1}{\sqrt{n_k^*}} \left( \frac{d}{n_j^{*}} \frac{1}{\sqrt{n_j}} |C_j \cap C_k^*| + |B \cap C_k^{*}| \frac{1}{\sqrt{n_j^*}} \right)  \\
	&\leq 
	\frac{\alpha^{-1/2} d}{(n_j^*)^{3/2} n_j^{1/2}} |C_j| + \frac{\alpha^{-1/2}}{n_j^{*}} d \\
	& < 
	\frac{2 \alpha^{-1/2} d}{n_j^{*}} .
\end{align*}
It is left to notice that
\[
	r^{2} \geq 2 - \frac{2 n_j}{\sqrt{n_j n_j^*}} = \frac{2(\sqrt{n_j^*} - \sqrt{n_j})}{\sqrt{n_j}} = \frac{2(n_j^* - n_j)}{n_j^* + \sqrt{n_j n_j^*}} \geq \frac{2 d}{2 n_j^*} ,
\]
therefore \( \| [Z^*]^{\T} \bv \|_{1} \leq 2 \alpha^{-1/2} r^{2} \), thus \( \| [Z^*]^{\T} (\zv_j - \zv_j^*) \|_{1} \leq 3 \alpha^{-1/2} r^{2} \).
\end{proof}

\begin{lemma}\label{proj_diff_r_square:lem}
	Let \( r = \trinorm Z_{\CC} - Z^{*} \trinorm_{\Frob} \) and suppose that \( r \leq 0.3 \). Then \( \trinorm P_{\CC} - P_{\CC^{*}} \trinorm_{\Frob}^2 \geq 2 r^2 (1 - 10 \alpha^{-1} r^{2} ) \).
\end{lemma}
\begin{proof}
	Denote \( \zv_{j} = \zv_{C_j}  \) and \( r_{j} = \| \zv_j - \zv_j^{*} \| \). It holds,
	\[
	\trinorm P_{\CC} - P_{\CC^{*}} \trinorm_{\Frob}^{2} = 2 K - 2 \tr(P_{\CC} P_{\CC^{*}}) = 2 K - \sum_{j, k} (\zv_{j}^{\T} \zv_{k}^{*})^{2} .
	\]
	Notice, that \( 2 \zv_j^{\T} \zv_{j}^{*} = 2 - \| \zv_j \|^{2} - \| \zv_j^{*} \|^{2} + 2 \zv_j^{\T} \zv_j^{*} = 2 - \| \zv_j - \zv_{j}^{*} \|^{2} \), i.e., \( \zv_j^{\T} \zv_j^{*} = 1 - r_j^{2} / 2 \). In particular,
	\( 1 - (\zv_j^{\T} \zv_j^{*})^{2} = r_j^{2} - r_j^{4} / 4  \), whereas \( ([\zv_j^{*}]^{\T} (\zv_j - \zv_{j}^{*}))^{2} = r_j^4 / 4 \). Since we additionally have \( [\zv_k^{*}]^{\T} (\zv_j - \zv_j^{*}) = [\zv_k^{*}]^{\T} \zv_j \) for \( k \neq j \), it holds
	\begin{align*}
	2 K - 2 \sum_{j, k} (\zv_{j}^{\T} \zv_{k}^{*})^{2} &= 2 \sum_{j} r_{j}^{2} - r_{j}^{4} / 4 - 2 \sum_{j} \sum_{k \neq j} \left([\zv_k^{*}]^{\T} (\zv_j - \zv_j^{*})\right)^{2} \\
	&= 2r^{2} - 2 \sum_{j, k} \left([\zv_k^{*}]^{\T} (\zv_j - \zv_j^{*})\right)^{2} \\
	& = 2r^{2} - 2 \sum_{j} \| [Z^{*}]^{\T} (\zv_j - \zv_j^{*}) \|^{2}
	\end{align*}
	By Lemma~\ref{l1_l2_z_comp:lem} we have for every \( j = 1, \dots, K \)
	\[
	\| [Z^{*}]^{\T} (\zv_j - \zv_j^{*}) \| \leq \| [Z^{*}]^{\T} (\zv_j - \zv_j^{*}) \|_{1} \leq 3.05 \alpha^{-1/2} r_j^{2} ,
	\]
	therefore
	\[
	\sum_{j} \| [Z^{*}]^{\T} (\zv_j - \zv_j^{*}) \|^{2} \leq 10 \alpha^{-1} \sum_{j} r_j^{4} \leq 10 \alpha^{-1} r^{4} ,
	\]
	thus inequality follows.
\end{proof}

\begin{lemma}\label{lemma_one_change}
	Let \(C, C'\) be such that \( |C \triangle C'| = 1 \). Then \( \| \zv_{C} - \zv_{C'} \|^{2} \leq \frac{2}{|C| \vee |C'|} \).
\end{lemma}
\begin{proof}
Suppose, \( |C'| > |C| \) then \( C' = C \cup \{a\}  \) and denoting \( n = |C| \) we have
\begin{align*}
	\| \zv_{C} - \zv_{C'} \|^{2} =& n \left( \sqrt{\frac{1}{n+1}} - \sqrt{\frac{1}{n}}  \right)^{2} + \frac{1}{n+1} = \frac{(\sqrt{n+1} - \sqrt{n})^{2} + 1}{n + 1} \leq \frac{2}{n + 1} .
\end{align*}
\end{proof}

\subsection{Proof of Theorem~\ref{main_thm}}

The proof consists of several steps, each represented by a separate lemma.

\begin{lemma}
Suppose, Assumption~\ref{time_series:assume} holds and let \( N \geq 2 \).
There is a constant \( C = C(\gamma, L) \), so that if 
\begin{equation}\label{neglect_subexp}
	\max(2, s \log^2 T, n^{*}) \frac{\log N}{T p_{\min}^2} \leq \frac{1}{9},
\end{equation}
then with probability at least \( 1 - 1/N  \) and for with \( \Delta_{1} = C \sigma_{\max} \sqrt{\frac{\log N}{T p_{\min}^{2}}} \) the following inequalities take place for every \( j = 1, \dots, K \)

\begin{equation}\label{A_inf_inf}
	\| \hat{A} - A \|_{\infty, \infty} \leq \Delta_{1},
	\qquad
	\| \Sigma_{\Lambda_j, \Lambda_j}^{-1} (\hat{A}_{\Lambda_j, \cdot} - A_{\Lambda_{j}, \cdot}) \|_{\infty, \infty} \leq \sigma_{\min}^{-1} \Delta_{1} ;
\end{equation}

\begin{equation}\label{A_inf_zj}
	\| (\hat{A} - A) \zv_{j}^{*} \|_{\infty} \leq \Delta_{1},
	\qquad
	\| \Sigma_{\Lambda_j, \Lambda_j}^{-1} ( \hat{A}_{\Lambda_j, \cdot} - A_{\Lambda_j, \cdot}) \zv_j^{*} \|_{\infty} \leq \sigma_{\min}^{-1} \Delta_{1};
\end{equation}

\begin{equation}\label{Sigma_inf_inf}
	\| \Sigmah - \Sigma \|_{\infty, \infty} \leq \Delta_{1},
	\qquad
	\| (\hat{\Sigma}_{\Lambda_j, \cdot} - \Sigma_{\Lambda_j, \cdot}) \vv_{j}^{*} \|_{\infty} \leq \Delta_{1} ;
\end{equation}

\begin{equation}\label{Sigma_Lambda_vv}
	\| \Sigma_{\Lambda_j, \Lambda_j}^{-1} (\hat{\Sigma}_{\Lambda_j, \cdot} - \Sigma_{\Lambda_j, \cdot}) \vv_{j}^{*} \|_{\infty} \leq \sigma_{\min}^{-1} \Delta_{1} ;
\end{equation}

\begin{equation}\label{Sigma_Lambda_op}
\trinorm \Sigmah_{\Lambda_{j}, \Lambda_{j}} - \Sigma_{\Lambda_j, \Lambda_j} \trinorm_{\op} \leq \sqrt{s} \Delta_{1} .
\end{equation}

\end{lemma}
\begin{proof}
By Theorem~\ref{cros_cov_missing:prop} for any pair \( \av, \bv \in \R^N \) with \( \| \av \| \leq 1 \), \( \| \bv \| \leq 1 \) it holds with probability \( \geq 1- N^{-m} \),
\[
	|\av^{\T} (\hat{A} - A) \bv | \leq C \sigma_{\max} \left\{\sqrt{\frac{(m + 1)\log N}{T p_{\min}^2}} \bigvee \frac{(m + 1)\log N \log T}{T p_{\min}^2}   \right\} . 
\]
Suppose for a moment that \( m \) is such that
\begin{equation}\label{m_cond__}
	\sqrt{\frac{(m + 1) s \log N}{T p_{\min}^2}} \log T = \BigO(1),
\end{equation}
so that we can neglect the second term. In order to meet the condition \eqref{sparse_projector} we also need to have,
\[
	\max\{ 2, \| \av \|_{0},  \| \bv\|_{0}, \sqrt{\| \av\|_{0} \| \bv\|_{0}} \log T \} \frac{\log (4N) + m \log N}{Tp_{\min}^{2}}  \leq 1 \, .
\]
Set,
\[
	A_0 = \{ (\ev_i, \ev_{i'}) : \; i, i' \leq N  \},
	\qquad
	B_0 = \{ (\ev_i, \zv_{l}^{*}) : \; i \leq N, l \leq K \},
\]
as well as for every \( j = 1, \dots, K \)
\begin{align*}
	A_j &= \{ (\sigma_{\min} \Sigma_{\Lambda_{j}, \Lambda_{j}}^{-1} \ev_i,  \ev_{i'}) : \; i \in \Lambda_j, i' \leq N   \},
	\\
	B_j &= \{ (\sigma_{\min} \Sigma_{\Lambda_{j}, \Lambda_{j}}^{-1} \ev_i,  \zv_{l}^{*}) : \; i \in \Lambda_j, l \leq K   \} .
\end{align*}
We have \(  |A_0| \leq N^2, |B_0| \leq NK \) and \( |A_j| \leq sN, |B_j| \leq sK \) for \( j = 1, \dots, N \), so since \( s, K \leq N \) together they have not more than \( 4N^{3} \) pairs of vectors \( (\av, \bv) \), each having norm bounded by one. In addition, each \( (\av, \bv) \in A_j \) has \( \| \av\|_0 \leq s \) and \( \| \bv \|_{0} = 1 \), whereas each \( \| \av \|_{0} \leq s \), \( \| \bv\|_{0} \leq n^{*} \). In the worst case, we need
\[
	\max(2, s, n^{*}, \sqrt{n^* s} \log T) \frac{\log(4N) + m \log N}{T p_{\min}^{2}} \leq 1 .
\]
Taking a union bound, we have that the inequalities \eqref{A_inf_inf} and \eqref{A_inf_zj} hold with probability at least \( 1 - 4 N^{3 - m} \). By analogy, we can show that \eqref{Sigma_inf_inf} and \eqref{Sigma_Lambda_vv} hold with probability at least \( 1 - 4 N^{3- m} \).

As for the last inequality, for every \( j = 1, \dots, K \) pick \( P_j = \sum_{i \in \Lambda_j} \ev_i \ev_i^{\T} \), i.e., projectors onto the subspace of vectors supported on \( \Lambda_j \). Then by Theorem~\ref{missing_cov_est:prop} it holds with probability at least \( 1 - KN^{-m}  \) for every \( j = 1, \dots , K \) (taking into account \eqref{m_cond__})
\[
	\trinorm \hat{\Sigma}_{\Lambda_j, \Lambda_j} - \Sigma_{\Lambda_j, \Lambda_j} \trinorm_{\op} = \trinorm P_j(\hat{\Sigma} - \Sigma) P_j \trinorm_{\op} \leq C \sigma_{\max} \sqrt{\frac{s (m + 1)\log N}{T p_{\min}^2}} . 
\]
The sparsity condition is satisfied once
\[
	\max(2, s\log T) \frac{\log(4N) + u}{T p_{\min}^{2}} \leq 1.
\]
The total probability will be at least \( 1 - 8 N^{3 - m} - K N^{-m} \), which is at least \( 1 - 1/N \) whenever \( m \geq 7 \) and \( N \geq 2 \), and both sparsity conditions are satisfied for \( m = 7\).

\end{proof}

In what follows we use the additional notation. For a vector \( \vv \in \R^{N} \) let \( \sign(\vv) \in \{-1, 0, 1\}^{N} \) denotes the vector consisting of coordinates,
\[
	\sign(\vv)_j = \left\{
	\begin{aligned}
		-1, \qquad &v_j < 0, \\
		0, \qquad &v_j = 0, \\
		1, \qquad &v_j > 0
	\end{aligned}
	\right.
	\qquad
	j = 1, \dots, N \, .
\]
We write \( \bar{\sv}_j = \sign(\vv_j^{*}) \) for each \( j = 1, \dots, K\). In addition, \( \sv_{j}^{*} = (\bar{\sv}_j)_{\Lambda_j} \), which only consists of the values \( \pm 1 \) since \( \Lambda_j \) is the support of \( \vv_{j}^{*} \).

In the following, we apply the technique from \cite{gribonval2015sparse}. Suppose that the LASSO solution \( \hat{\vv}_{j} \) for a given clustering \( \CC \) is not only supported exactly on \( \Lambda_{j} \), but its signs are matching those of the true \( \vv_{j}^{*} \). Let \( \sv_{j}^{\T} \in \{ -1, 0, 1 \}^{N} \) be the vector consisting of the signs of coordinates of \( \vv_{j}^{*} \), i.e. \( -1\) for negative, \(1\) for positive, and zero for the zero coordinates of \( \vv_{j}^{*}\). Then, \( \| \hat{\vv}_{j} \|_{1} = \bar{\sv}_{j}^{\T} (\hat{\vv}_{j})_{\Lambda_j} \). Therefore, we can write
\begin{align*}
	(\hat{\vv}_{j})_{\Lambda_{j}} &=  \arg\min_{\vv \in \R^{\Lambda_{j}}} \frac{1}{2} \vv^{\T} \Sigmah_{\Lambda_j, \Lambda_j} \vv - \vv^{\T} \hat{A}_{\Lambda_j, \cdot} \zv_{j} + \lambda \bar{\sv}_{j}^{\T} \vv \\
	& =  \Sigmah_{\Lambda_j , \Lambda_{j}}^{-1} (\hat{A}_{\Lambda_j, \cdot} \zv_{j} - \lambda \bar{\sv}_{j}  ) ,
\end{align*}
and plugging this solution into the risk function we get that \( F_{\lambda}(\CC) = \Phi_{\lambda}(\CC) \), where the latter is defined explicitly
\[
	\Phi_{\lambda}(\CC) = -\frac{1}{2} \sum_{j = 1}^{K} (\hat{A}_{\Lambda_j, \cdot} \zv_{j} - \lambda \bar{\sv}_j)^{\T} \Sigmah_{\Lambda_j, \Lambda_j}^{-1} (\hat{A}_{\Lambda_j, \cdot} \zv_{j} - \lambda \bar{\sv}_j) .
\]
The next lemma shows that such representation takes place in the local vicinity of the true clustering \( \CC^* \).

\begin{lemma}\label{exact:lem}
Suppose, the inequalities \eqref{A_inf_inf}--\eqref{Sigma_Lambda_op} take place. Assume,
\begin{equation}\label{Delta_not_too_big:exact}
	s \Delta_1 \leq 1 / 16 ,
	\qquad
	{12} \Delta_{1} \leq \lambda \leq \frac{\sigma_{\min}}{4} \tau_{0} s^{-1} .
\end{equation}
Then, for any \( \CC = (C_1, \dots, C_K) \) satisfying
\begin{equation}\label{ineq_r_lambda:exact}
	\max_{j} \| \zv_{C_{j}} - \zv_{C^{*}_{j}} \| \leq 0.3 \wedge 0.22 \sqrt{\left(2 \sigma_{\max} \alpha^{-1/2} + \sqrt{n^{*}}\Delta_1  \right)^{-1} \lambda} 
\end{equation}
it holds
\[
	\trinorm \hat{V}_{\lambda, \CC} - V^{*} \trinorm_{\Frob} \leq 3 \sigma_{\min}^{-1} \sqrt{Ks} \lambda ,
\]
and the equality \( F_{\lambda}(\CC) = \Phi_{\lambda}(\CC) \) takes place.
\end{lemma}
\begin{proof}
Taking into account \( Z^{\T} Z = \Id_{K} \), it holds
\begin{align*}
	R_{\lambda}(V; \CC) = & \frac{1}{2} \Tr\left( V^{\T} \Sigmah V \right) - \Tr\left( V^{\T} \hat{A} Z \right) + \lambda \| V \|_{1, 1} \\
	= & \sum_{ j = 1}^{K}  \frac{1}{2} \vv_{j}^{\T} \Sigmah \vv_{j} - \vv_{j}^{\T} \hat{A} \zv_{j} + \lambda \| \vv_{j} \|_{1} ,
\end{align*}
so that the optimization problem separates into \(K\) independent subproblems. Solving each of the problems
\[
	\frac{1}{2} \vv_{j}^{\T} \Sigmah \vv_{j} - \vv_{j}^{\T} \hat{A} \zv_{j} + \lambda \| \vv_{j} \|_{1} \rightarrow \min_{\vv_{j}}
\]
corresponds to Corollary~\ref{lasso_exact:cor} with \( \Dh = \Sigmah \) and \( \cvh = \hat{A} \zv_{j} \), whereas the ``true'' version of the problem corresponds to \( \Db = \Sigma \) and \( \cvb = A \zv_{j}^{*} = \Sigma (\Theta^{*})^{\T} \zv_{j}^{*} = \Sigma \vv_{j}^{*} \). We need to control the differences between \( \cvh \) and \( \cvb \), and between \( \Dh \) and \( \Db \). It holds,
\begin{align*}
	 \| \hat{A} \zv_{j} - {A} \zv_{j}^{*} \|_{\infty} \leq & \| A(\zv_{j} - \zv_{j}^{*}) \|_{\infty} + \| (\hat{A} - A) \zv_{j}^{*} \|_{\infty} + \| (\hat{A} - A) (\zv_{j} - \zv_{j}^{*}) \|_{\infty} \, .
\end{align*}
Since \( A = \Sigma V^{*} [Z^{*}]^{\T} \), we bound the first term using Lemma~\ref{l1_l2_z_comp:lem}
\[
	\| A (\zv_j - \zv_j^{*}) \|_{\infty} \leq \| \Sigma V^{*} \|_{\infty, \infty} \| [Z^*]^{\T} (\zv_j - \zv_j^{*}) \|_{1} \leq 3.05 \alpha^{-1/2} \| \Sigma V^{*} \|_{\infty, \infty} r^{2}_{j}.
\]
The second term is bounded by \( \Delta_{1} \), whereas the fourth term satisfies
\[
	\| (\hat{A} - A) (\zv_{j} - \zv_{j}^{*}) \|_{\infty} \leq \| \hat{A} - A \|_{\infty, \infty} \| \zv_{j} - \zv_{j}^{*} \|_{1} \leq 1.65 \Delta_{1} \sqrt{n^{*}} r_{j}^{2},
\]
where we also used Lemma~\ref{l1_l2_comp:lem}. 
Summing up,we get,
\[
	\| \hat{\cv} - \cv \|_{\infty} \leq  1.65 (2\sigma_{\max} \alpha^{-1/2} + \sqrt{n^{*}_j} \Delta_1) r_j^{2} + \Delta_1 \, .
\]
Similarly, we bound \(  \| \Sigma_{\Lambda_j, \Lambda_j} (\cvh_{\Lambda_j} - \cvb_{\Lambda_j})\|_{\infty} \) as follows
\begin{align*}
	\| \Sigma_{\Lambda_j, \Lambda_j}^{-1}(\hat{A}_{\Lambda_{j}, \cdot} \zv_{j} - {A}_{\Lambda_{j}, \cdot} \zv_{j}^{*}) \|_{\infty} 
	\leq 
	& 
	\| \Sigma_{\Lambda_j, \Lambda_j}^{-1} A(\zv_{j} - \zv_{j}^{*}) \|_{\infty} + \| \Sigma_{\Lambda_j, \Lambda_j}^{-1} (\hat{A}_{\Lambda_{j}, \cdot} - A_{\Lambda_{j}, \cdot}) \zv_{j}^{*} \|_{\infty} \\
	& 
	\, + \| \Sigma_{\Lambda_j, \Lambda_j}^{-1} (\hat{A}_{\Lambda_j, \cdot} - A_{\Lambda_j, \cdot}) (\zv_{j} - \zv_{j}^{*}) \|_{\infty} \\
	\leq 
	&
	\| \Sigma_{\Lambda_j, \Lambda_j}^{-1} A(\zv_{j} - \zv_{j}^{*}) \|_{\infty} + 
	1.65 \sigma_{\min}^{-1} \Delta_{1} \sqrt{n^{*}} r_{j}^{2} + \sigma^{-1}_{\min} \Delta_{1} \\
	\leq &
	1.65 \sigma_{\min}^{-1} (2\sigma_{\max} \alpha^{-1/2} + \sqrt{n^{*}_j} \Delta_1) r_j^{2} + \sigma_{\min}^{-1} \Delta_1
\end{align*}
To sum up, Corollary~\ref{lasso_exact:cor} is applied with
\begin{align*}
	\delta_{c} = & 1.65  (2\sigma_{\max} \alpha^{-1/2} + \sqrt{n^{*}} \Delta_1) r_j^{2} + \Delta_1 , \\
	\delta_{c}' = & 1.65 \sigma_{\min}^{-1} (2\sigma_{\max} \alpha^{-1/2} + \sqrt{n^{*}} \Delta_1) r_j^{2} + \sigma_{\min}^{-1} \Delta_1 \\
	\delta_{D}  = &\Delta_1,
	\qquad
	\delta_{D}' = \Delta_1,
	\qquad
	\delta_{D}'' = \sigma_{\min}^{-1}\Delta_1 .
\end{align*}
It requires the conditions,
\[
	3 \{1.65  (2\sigma_{\max} \alpha^{-1/2} + \sqrt{n^{*}} \Delta_1) r_j^{2} + 2 \Delta_1 \} \leq \lambda,
	\qquad
	s \Delta_{1} \leq \frac{1}{16},
\]
and due to the fact that \( \| D_{\Lambda_{j}, \Lambda_{j}}^{-1} \|_{1, \infty} \leq \sqrt{s} \trinorm D_{\Lambda_{j}, \Lambda_{j}}^{-1} \trinorm_{\op}   \) and Assumption~\ref{coeff_separ:assume},
\[
	2 \sigma_{\min}^{-1}( 1.65  (2\sigma_{\max} \alpha^{-1/2} + \sqrt{n^{*}} \Delta_1) r_j^{2} + 2 \Delta_1  + \sqrt{s} \lambda) < \tau_{0} s^{-1/2} ,
\]
which are not hard to derive from the given inequalities.
Together this yields that \( \hat{\vv}_{j} \) is supported on \( \Lambda_j \) and the solution satisfies
\[
	(\hat{\vv}_{j})_{\Lambda_{j}} = \Sigmah^{-1}_{\Lambda_j, \Lambda_j} \left( \hat{A}_{\Lambda_j, \cdot} \zv_{j} - \lambda \sv_{j}^{*} \right) ,
\]
and the corresponding minimum is equal to
\[
	\frac{1}{2} \hat{\vv}_{j}^{\T} \Sigmah \hat{\vv}_{j}^{\T} - \hat{\vv}_{j}^{\T} \hat{A} \zv_{j} + \lambda (\hat{\vv}_{j})_{\Lambda_{j}}^{\T}\sv_{j}^{*} = -\frac{1}{2} \left( \hat{A}_{\Lambda_j, \cdot} \zv_{j} - \lambda \sv_{j}^{*} \right)^{\T} \Sigmah_{\Lambda_j, \Lambda_j}^{-1} \left( \hat{A}_{\Lambda_j, \cdot} \zv_{j} - \lambda \sv_{j}^{*} \right) .
\]
Summing up, we get the corresponding expression for \( F_{\lambda}(\CC) \). Moreover, we have
\begin{align*}
	\| \hat{\vv}_{j} - \vv_{j}^{*} \| \leq & 2 \sqrt{s} \left\{ 2 \Delta_1 + 1.65 (2 \sigma_{\max} \alpha^{-1} + \sqrt{n^*}\Delta_1)r_j^2 + \lambda  \right\} \\
	\leq &
	2 \sigma^{-1}_{\min} \sqrt{s} \left( \frac{\lambda}{6} + \frac{1.65 \lambda}{20} + \lambda  \right) \\
	\leq & 3 \sigma_{\min}^{-1} \sqrt{s} \lambda ,
\end{align*}
and together it provides a bound on \( \trinorm \hat{V}_{\lambda, \CC} - V^{*} \trinorm_{\Frob} \).
\end{proof}

Consider the function,
\[
	\Phib_{\lambda}(\CC) = -\frac{1}{2} \sum_{j = 1}^{k} \left( {A}_{\Lambda_{j}, \cdot} \zv_{j} - \lambda \sv_{j}^{*} \right)^{\T} \Sigma_{\Lambda_{j}, \Lambda_{j}}^{-1} \left( {A}_{\Lambda_{j}, \cdot} \zv_{j} - \lambda  \sv_{j}^{*} \right) .
\]
The following lemma shows how this function grows with \( \CC \) retreating from the true clustering \( \CC^{*} \).

\begin{lemma}\label{l_r:lem}
Suppose, \( \CC \) is a clustering such that \( r = \trinorm Z_{\CC} - Z^{*} \trinorm_{\Frob} \leq 0.3 \). Then,
\[
	\Phib_{\lambda}(\CC) - \Phib_{\lambda}(\CC^{*}) \geq \frac{a_0}{2} r^{2}(1 - 10\alpha^{-1}r^{2}) - \lambda \sqrt{Ks} \trinorm V^{*} \trinorm_{\Frob}r .
\]
\end{lemma}
\begin{proof}
Denoting \( \Phib_{0}(\CC) = -\frac{1}{2} \sum_{j = 1}^{k} \zv_{j}^{\T} \hat{A}_{\Lambda_{j}, \cdot}^{\T} \Sigmah_{\Lambda_{j}, \Lambda_{j}}^{-1} \hat{A}_{\Lambda_{j}, \cdot} \zv_{j} \) (which indeed corresponds to \( \lambda = 0 \)), we have the decomposition
\[
	\Phib_{\lambda}(\CC) - \Phib_{\lambda}(\CC^{*}) = \Phib_{0}(\CC) - \Phib_{0}(\CC^*) - \lambda \sum_{j = 1}^{K} [\sv_{j}^{*}]^{\T} \Sigma_{\Lambda_j, \Lambda_j}^{-1} A_{\Lambda_j, \cdot} (\zv_{j} - \zv_{j}^{*}) .
\]
Let us first deal with the term \( \Phib_{0}(\CC) - \Phib_{0}(\CC^*) \). Note that since \( [\vv_{j}^{*}]_{\Lambda_{j}} = \Sigma_{\Lambda_{j}, \Lambda_{j}}^{-1} {A}_{\Lambda_{j}, \cdot} \zv_{j}^{*} \), we have
\[
	\Phib_{0}(\CC^{*}) = - \frac{1}{2} \sum_{j=1}^{K} [\vv_{j}^{*}]^{\T} \Sigma \vv_{j}^{*} = - \frac{1}{2} \Tr( [V^{*}]^{\T} \Sigma V^{*} ) = - \frac{1}{2} \Tr(\Theta^{*} \Sigma [\Theta^{*}]^{\T}) .
\]
whereas
\[
	\Phib_{0}(\CC) = \min_{V = [\vv_{1}, \dots, \vv_{k}]} \frac{1}{2} \Tr(V^{\T} \Sigma V) - \Tr(V^{\T} A Z_{\CC})  
\]
where the minimum is taken s.t. the restrictions \( \supp(\vv_j) \subset \Lambda_{j} \). Dropping the restrictions we get,
\begin{align*}
	\Phib_{0}(\CC) - \Phib_{0}(\CC^{*}) &\geq \min_{V} \frac{1}{2} \Tr(V^{\T} \Sigma V) - \Tr(V^{\T} A Z_{\CC}) + \frac{1}{2} \tr(\Theta^{*} \Sigma [\Theta^{*}]^{\T}) \\
	&=
	\min_{V} \frac{1}{2} \trinorm Z_{\CC} V^{\T} \Sigma^{1/2}   \trinorm_{\Frob}^{2} - \Tr(Z_{\CC} V^{\T} \Sigma [\Theta^{*}]^{\T}) + \trinorm \Theta^{*} \Sigma^{1/2} \trinorm_{\Frob}^{2} \\
	&=
	\min_{V} \frac{1}{2} \trinorm (Z_{\CC} V^{\T} - \Theta^{*}) \Sigma^{1/2} \trinorm_{\Frob}^{2} .
\end{align*}
It is not hard to calculate that the minimum is attained for \( V = [\Theta^{*}]^{\T} Z_{\CC} \) and therefore
\[
	\Phib_{0}(\CC) - \Phib_{0}(\CC^{*}) \geq \frac{1}{2} \trinorm (Z_{\CC}Z_{\CC}^{\T} - I) \Theta^{*} \Sigma^{1/2} \trinorm_{\Frob}^{2} \geq \frac{a_{0}}{2} \trinorm (Z_{\CC}Z_{\CC}^{\T} - I) Z^{*} \trinorm_{\Frob}^{2},
\]
where the latter follows using \( \Theta^{*} = Z^{*} [V^{*}]^{\T} \) and from the fact that \( \lambda_{\min} ([V^{*}]^{\T} \Sigma V^{*}) \geq \sigma_{0} \). Moreover,
\begin{align*}
	\trinorm (Z_{\CC}Z_{\CC}^{\T} - I) Z^{*} \trinorm_{\Frob}^{2} &= \Tr((P_{\CC} - \Id) P_{\CC^{*}} (P_{\CC} - \Id)) = \Tr(P_{\CC^{*}}) - \Tr( P_{\CC} P_{\CC^{*}} ) \\
	&= \frac{1}{2} \trinorm P_{\CC} - P_{\CC^{*}} \trinorm_{\Frob}^{2},
\end{align*}
where we used the fact that \( \Tr(P_{\CC}) = \Tr(P_{\CC^{*}}) =K\). It is left to recall the result of Lemma~\ref{proj_diff_r_square:lem}, so that we get
\[
	\Phib_{0}(\CC) - \Phib_{0}(\CC^{*}) \geq \frac{a_0 r^{2}}{2} (1 - 10 \alpha^{-1} r^{2}) .
\]

As for the linear term, it holds
\[
	\left(\sum_{j = 1}^{K} [\sv_{j}^{*}]^{\T} \Sigma_{\Lambda_j, \Lambda_j}^{-1} A_{\Lambda_j, \cdot} (\zv_{j} - \zv_{j}^{*}) \right)^{2} \leq \left( \sum_{j = 1}^{K}  \| [\sv_{j}^{*}]^{\T} \Sigma_{\Lambda_j, \Lambda_j}^{-1} A_{\Lambda_j, \cdot} \|^{2} \right) r^{2}
\]
Since \( A = \Sigma [\Theta^*]^{\T} \), we have 
\(
	A_{\Lambda_j, \cdot}^{\T} \Sigma_{\Lambda_j, \Lambda_j}^{-1} \sv_{j}^{*} = \Theta^{*} \Sigma_{\cdot, \Lambda_j} \Sigma_{\Lambda_j, \Lambda_j}^{-1} \sv_j^{*}
\). Denote, \( \xv = \Sigma_{\cdot, \Lambda_j} \Sigma_{\Lambda_j, \Lambda_j}^{-1} \sv_j^{*} \), then we have \( \xv_{\Lambda_j} = \sv_j \) and \( \| \xv_{\Lambda_j} \|_{\infty} = 1 \). Moreover, by the ERC property
\[
	\| \xv_{\Lambda^{c}_j} \|_{\infty} = \| \Sigma_{\Lambda_{j}^{c}, \Lambda_{j}} \Sigma_{\Lambda_j, \Lambda_j}^{-1} \sv_j \|_{\infty} \leq \| \Sigma_{\Lambda_{j}^{c}, \Lambda_{j}} \Sigma_{\Lambda_j, \Lambda_j}^{-1} \|_{1, \infty} \leq 1/2 .
\]
We have
\[
	\| A_{\Lambda_j, \cdot}^{\T} \Sigma_{\Lambda_j, \Lambda_j}^{-1} \sv_{j}^{*} \|^{2} = \| \sum \zv_{j}^{*} [\vv_{j}^{*}]^{\T} \xv \|^{2} = \sum_{k = 1}^{K} | [\vv_{k}^{*}]^{\T} \xv |^{2} ,
\]
where, since \( \vv_{k}^{*} \) is supported on \( \Lambda_{k} \) of size at most \( s \),
\[
	| [\vv_{k}^{*}]^{\T} \xv | \leq \| \vv_{k}^{*} \|_{1} \| \xv \|_{\infty} \leq \sqrt{s} \| \vv_{k}^{*} \| .
\]
Summing up, we get \( \| A_{\Lambda_j, \cdot}^{\T} \Sigma_{\Lambda_j, \Lambda_j}^{-1} \sv_{j}^{*} \|^{2} \leq s \trinorm V^{*} \trinorm_{\Frob}^{2} \), so that
\[
	\left| \sum_{j = 1}^{K} [\sv_{j}^{*}]^{\T} \Sigma_{\Lambda_j, \Lambda_j}^{-1} A_{\Lambda_j, \cdot} (\zv_{j} - \zv_{j}^{*}) \right| \leq \sqrt{Ks} \trinorm V^{*} \trinorm_{\Frob}r .
\]
The lemma now follows from the two terms put together.
\end{proof}

The next step is to bound the difference \( \Phi_{\lambda}(\CC) - \Phib_{\lambda}(\CC) \) uniformly in the neighbourhood of \( \CC^{*} \).

\begin{lemma}\label{ed2:lemma}
Suppose that the inequalities \eqref{A_inf_inf}--\eqref{Sigma_Lambda_op} hold and let
\begin{align*}
	\Delta_{1} &\leq \sigma_{\min}/(2\sqrt{s}) \vee \frac{\lambda}{12},
	\qquad
	\sigma_{\max}/\sigma_{\min} \leq n^{*}, \qquad
	\lambda \leq \sigma_{\min} s^{-1}
\end{align*}
Let some \( r \leq 0.3 \) satisfies \( \sqrt{sn^{*}} \Delta_{1} r^{2} \leq \sigma_{\max} \).
Then,
\begin{align*}
	\sup_{\trinorm Z - Z^{*} \trinorm_{\Frob} \leq r}  | \Phi_{\lambda}(\CC) & - \Phib_{\lambda}(\CC) - \Phi_{\lambda}(\CC^{*}) + \Phib_{\lambda}(\CC^{*}) | \\
	\leq &
	4\left ( \left(\frac{\sigma_{\max}}{\sigma_{\min}} \right)^{2} \sqrt{s} \trinorm V^{*} \trinorm_{\Frob}  + \frac{\sigma_{\max}}{\sigma_{\min}} \sqrt{K} \right)Delta_{1} r + 16 \frac{\sigma_{\max}}{\sigma_{\min}} \sqrt{sn^*} \Delta_1 r^{2} .
\end{align*}

\end{lemma}

\def\Phit{\tilde{\Phi}}
\begin{proof}
Denote,
\begin{align*}
	\Phit_{\lambda}(\CC) &= - \frac{1}{2} \sum_{j = 1}^{K} \left(A_{\Lambda_{j}, \cdot} \zv_{j} - \lambda \sv_{j}^{*}  \right)^{\T} \hat{\Sigma}_{\Lambda_{j}, \Lambda_j}^{-1} \left(A_{\Lambda_{j}, \cdot} \zv_{j} - \lambda \sv_{j}^{*}  \right) ,
\end{align*}
so that we have
\begin{align*}
	| \Phit_{\lambda}(\CC) &- \Phib_{\lambda}(\CC) - \Phit_{\lambda}(\CC^{*}) + \Phib_{\lambda}(\CC^{*}) | \\
	&\leq \frac{1}{2} \sum_{j = 1}^{K} \left| \left({A}_{\Lambda_{j}, \cdot} (\zv_{j} + \zv_{j}^{*}) - 2 \lambda \sv_{j}^{*}  \right)^{\T} (\Sigmah_{\Lambda_{j}, \Lambda_j}^{-1} - \Sigma_{\Lambda_{j}, \Lambda_{j}}^{-1}) {A}_{\Lambda_{j}, \cdot} (\zv_{j} - \zv_{j}^{*}) \right|
\end{align*}
First of all, due to \eqref{Sigma_Lambda_op} it holds,
\[
	\trinorm \Sigmah_{\Lambda_{j}, \Lambda_{j}}^{-1} - \Sigma_{\Lambda_j, \Lambda_j}^{-1} \trinorm_{\op} \leq \frac{\sigma_{\min}^{-2} \sqrt{s}\Delta_{1}}{1 - \sigma_{\min}^{-1} \sqrt{s} \Delta_{1}} \leq 2 \sigma_{\min}^{-2} \sqrt{s} \Delta_{1} .
\]
Since \( A = \Sigma [\Theta^{*}]^{\T} \), we have 
\begin{align*} 
	\| {A}_{\Lambda_{j}, \cdot} (\zv_{j} - \zv_{j}^{*}) \| & \leq \sigma_{\max} r_{j} \\
	\|  {A}_{\Lambda_{j}, \cdot} (\zv_{j} + \zv_{j}^{*}) - 2\lambda \sv_{j}^{*} \| & \leq \sigma_{\max} (2 \| \vv_j^{*} \| +  r_j) + 2 \lambda \sqrt{s} .
\end{align*}
Then by Cauchy-Schwartz,
\begin{align*}
	| \Phit_{\lambda}(\CC) - \Phib_{\lambda}(\CC) - \Phit_{\lambda}(\CC^{*}) + \Phib_{\lambda}(\CC^{*}) | 
	\leq &
	\sigma_{\min}^{-2} \sqrt{s} \Delta_{1} \left( \sum_{j = 1}^{K} \sigma_{\max} r_j \left\{  \sigma_{\max}(2 \| \vv_j\| + r_j) + 2 \lambda \sqrt{s} \right\}  \right) \\
	\leq&  2 \left( \frac{\sigma_{\max}}{\sigma_{\min}}  \right)^{2} \sqrt{s}  \trinorm V^{*} \trinorm_{\Frob} \Delta_{1} r + 2 \frac{\sigma_{\max}}{\sigma_{\min}^{2}}  \lambda s \sqrt{K} \Delta_{1} r \\
	& + \left( \frac{\sigma_{\max}}{\sigma_{\min}}  \right)^{2} \sqrt{s} \Delta_{1} r^{2} .
\end{align*}

Going further,
\[
	\Phi_{\lambda}(\CC) - \Phit_{\lambda}(\CC) = -\frac{1}{2} \sum_{j = 1}^{K} \left( (A_{\Lambda_{j}, \cdot} + \hat{A}_{\Lambda_{j}, \cdot}) \zv_{j} - 2 \lambda \sv^*_j  \right)^{\T} \Sigmah_{\Lambda_j, \Lambda_j}^{-1} (\hat{A}_{\Lambda_j, \cdot}  - A_{\Lambda_j, \cdot}) \zv_{j},
\]
which implies that
\begin{equation}\label{F_Ft_decomp}
\begin{aligned}
	| \Phi_{\lambda}(\CC) &- \Phit_{\lambda}(\CC) - \Phi_{\lambda}(\CC^{*}) + \Phit_{\lambda}(\CC^{*}) | \\
	\leq & \frac{1}{2} \sum_{j = 1}^{K} \left| \left( (A_{\Lambda_{j}, \cdot} + \hat{A}_{\Lambda_{j}, \cdot}) (\zv_{j} - \zv_{j}^{*}) \right)^{\T} \Sigmah_{\Lambda_j, \Lambda_j}^{-1} (\hat{A}_{\Lambda_j, \cdot}  - A_{\Lambda_j, \cdot}) \zv_{j} \right| \\
	& + \frac{1}{2} \sum_{j = 1}^{K} \left| \left( (A_{\Lambda_{j}, \cdot} + \hat{A}_{\Lambda_{j}, \cdot}) \zv_{j}^{*} - 2 \lambda \sv^*_j \right)^{\T} \Sigmah_{\Lambda_j, \Lambda_j}^{-1} (\hat{A}_{\Lambda_j, \cdot}  - A_{\Lambda_j, \cdot}) (\zv_{j} - \zv_{j}^{*}) \right|
\end{aligned}
\end{equation}
First notice, that due to Lemma~\ref{l1_l2_comp:lem} and \eqref{A_inf_inf} it holds,
\begin{align*}
	\| (\hat{A}_{\Lambda_{j}, \cdot} - {A}_{\Lambda_{j}, \cdot})(\zv_j - \zv_j^{*}) \| & \leq \sqrt{s} \| \hat{A}_{\Lambda_{j}, \cdot} - {A}_{\Lambda_{j}, \cdot} \|_{\infty, \infty} \| \zv_j - \zv_j^{*} \|_{1} \\
	& \leq 1.65 \sqrt{s n^{*}} \Delta_{1} r_{j}^{2} .
\end{align*}
Therefore, it follows
\[
	\| (\hat{A}_{\Lambda_{j}, \cdot} + {A}_{\Lambda_{j}, \cdot})(\zv_j - \zv_j^{*}) \| \leq 2 \sigma_{\max} r_j + 1.65 \sqrt{s n^{*}} \Delta_{1} r_{j}^{2} .
\]
Moreover, using \eqref{A_inf_zj} we get
\begin{align*}
	\| (\hat{A}_{\Lambda_{j}, \cdot} - {A}_{\Lambda_{j}, \cdot}) \zv_j \| & \leq \Delta_{1} + 1.65 \sqrt{s n^{*}} \Delta_{1} r_{j}^{2} \\
	\| (\hat{A}_{\Lambda_{j}, \cdot} + {A}_{\Lambda_{j}, \cdot}) \zv_j^{*}  - 2 \lambda \sv_{j}^{*} \| & \leq 2 \sigma_{\max} \| \vv_{j} \| + \Delta_{1} + 2 \lambda \sqrt{s} .
\end{align*}
and we also have \( \trinorm \Sigmah_{\Lambda_j, \Lambda_j}^{-1} \trinorm_{\op} \leq 2 \sigma_{\min}^{-1} \) due to the condition \( \sigma_{\min}^{-1} \sqrt{s} \Delta_{1} \leq 1/2 \).
Thus we get that the first sum of \eqref{F_Ft_decomp} is bounded by
\begin{align*}
 \sigma_{\min}^{-1} \sum_{j = 1}^{K} & \left( 2 \sigma_{\max} r_j + 1.65 \sqrt{s n^*} \Delta_{1} r_j^{2} \right) \left( \Delta_1 + 1.65 \sqrt{s n^{*}} \Delta_1 r_j^{2} \right) \\
 & \leq 2 \frac{\sigma_{\max}}{\sigma_{\min}} \Delta_{1} \sqrt{K} r + 1.65 \sigma^{-1}_{\min} \sqrt{s n^{*}} \Delta_{1}^{2} r^{2} + 3.3 \frac{\sigma_{\max}}{\sigma_{\min}} \sqrt{s n^{*}} \Delta_1 r^{3} + 2.8 \sigma_{\min}^{-1} sn^{*} \Delta_1^{2} r^{4} ,
\end{align*}
while the second sum is bounded by
\begin{align*}
	\sigma_{\min}^{-1} \sum_{j = 1}^{K} & \left( 2 \sigma_{\max} \| \vv_{j}^{*} \| + \Delta_{1} + 2 \lambda \sqrt{s} \right) \left( 1.65 \sqrt{s n^{*}} \Delta_{1} r_j^2 \right) \\
	& \leq  \frac{1.65}{\sigma_{\min}} \left(\sigma_{\max} \sqrt{s n^{*}}  + \sqrt{sn^{*}} \Delta_1 + 2 \lambda s \sqrt{n^{*}} \right) \Delta_1 r^{2} \\
	& \leq \frac{3.3}{\sigma_{\min}} \left(\sigma_{\max} \sqrt{s n^{*}}  + \lambda s \sqrt{n^{*}} \right) \Delta_1 r^{2}
\end{align*}
where we used the fact that \( \max_{j} \| \vv_{j}^{*} \| \leq \trinorm V^{*} \trinorm_{\op} = \trinorm \Theta^{*} \trinorm_{\op} < 1 \) together with the condition of the lemma \( \Delta_1 \leq \sigma_{\max}\). Combining all the bounds we get
\begin{align*}
	| \Phi_{\lambda}(\CC) &- \Phib_{\lambda}(\CC) - \Phi_{\lambda}(\CC^{*}) + \Phib_{\lambda}(\CC^{*}) | \\
	\leq & 2\left\{ \left(\frac{\sigma_{\max}}{\sigma_{\min}} \right)^{2} \sqrt{s} \trinorm V^{*} \trinorm_{\Frob}  + 2 \frac{\sigma_{\max}}{\sigma_{\min}^{2}} \lambda s\sqrt{K} + 2 \frac{\sigma_{\max}}{\sigma_{\min}} \sqrt{K} \right\} \Delta_{1} r \\
	& + \left\{  3.3 \frac{\sigma_{\max}}{\sigma_{\min}} \sqrt{sn^*} + 3.3 \sigma_{\min}^{-1} \lambda s \sqrt{n^*} + 1.65 \sigma_{\min}^{-1} \sqrt{sn^*} \Delta_{1} + \left(\frac{\sigma_{\max}}{\sigma_{\min}} \right)^{2} \sqrt{s}  \right\} \Delta_1 r^{2} \\
	& + 3.3 \frac{\sigma_{\max}}{\sigma_{\min}} \sqrt{sn^*} \Delta_{1} r^3 \\
	& + 2.8 \sigma_{\min}^{-1} sn^{*} \Delta_{1}^{2} r^{4} ,
\end{align*}
where by \( r \leq 0.3 \) and \( \sqrt{sn^{*}} \Delta_{1} \leq \sigma_{\max} \) we can neglect the third and the fourth power, respectively, and thus the required bound follows.
\end{proof}

\begin{lemma}
There are numerical constants \( c, C > 0 \) such that the following holds. Suppose, the inequalities take place:
\begin{align}
\label{s_n_star_logN_over_Tdelta}
	 \sqrt{\frac{s n^* \log N}{T p_{\min}^{2}}} & \leq c \frac{a_{0} \sigma_{\min}}{\sigma_{\max}^{2}},
	 \qquad
	 n^{*} \geq \sigma_{\max} / \sigma_{\min} .
\end{align}
Let
\(
	 C  \sigma_{\max} \sqrt{\frac{\log N}{T p_{\min}^{2}}} \leq \lambda  \leq c \sigma_{\min} \tau_{0} s^{-1} \), and set
\[
	\bar{r} = 0.3 \wedge 0.18\sqrt{\alpha} \wedge 0.22 \sqrt{\left(2 \sigma_{\max} \alpha^{-1/2} + \sqrt{n^{*}}\Delta_1  \right)^{-1} \lambda}  .
\]
Then under the inequalities \eqref{A_inf_inf}--\eqref{Sigma_Lambda_op} the clustering 
\[
	\hat{\CC} = \arg\min_{\trinorm Z_{\CC} - Z^{*} \trinorm_{\Frob} \leq r_{\max}} F_{\lambda}(\CC)
\]
satisfies 
\[
\trinorm Z_{\hat{\CC}} - Z^{*} \trinorm_{\Frob} \leq \frac{C}{a_{0}} \left(\frac{\sigma_{\max}}{\sigma_{\min}} \right)^{2} \lambda  K\sqrt{s} \, .
\]
\end{lemma}
\begin{proof}
It is not hard to see that for \( \Delta_{1} = \sqrt{\frac{\log N}{Tp_{\min}^{2}}} \) the inequalities required by Lemmata~\ref{exact:lem}--\ref{ed2:lemma} are satisfied for \( r \leq \bar{r} \) due to \eqref{s_n_star_logN_over_Tdelta} and conditions on \( \lambda \) and \( \bar{r} \). Since obviously \( \hat{\CC} \) satisfies \( F_{\lambda}(\hat{\CC}) \leq F_{\lambda}(\CC^{*})  \), we have for \( r = \trinorm Z_{\hat{\CC}} - Z_{\CC^{*}} \trinorm_{\Frob} \leq r_{\max} \)
\begin{align*}
	F_{\lambda}(\hat{\CC}) - F_{\lambda}(\CC^{*}) \geq & \Phib_{\lambda}(\CC) - \Phib_{\lambda}(\CC) - | F_{\lambda}(\CC) - \Phib_{\lambda}(\CC) - F_{\lambda}(\CC^{*}) + \Phib_{\lambda}(\CC^{*}) | \\
	\geq & \frac{a_0r^{2}}{2} \left( 1 - 10\alpha^{-1} r^{2} \right) - \lambda \sqrt{Ks} \trinorm V^{*} \trinorm_{\Frob} r \\
	& - 4\left\{ \left(\frac{\sigma_{\max}}{\sigma_{\min}} \right)^{2} \sqrt{s} \trinorm V^{*} \trinorm_{\Frob}  + \frac{\sigma_{\max}}{\sigma_{\min}} \sqrt{K} \right\} \Delta_{1} r - 15 \frac{\sigma_{\max}}{\sigma_{\min}} \sqrt{sn^*} \Delta_1 r^{2} \\
	= & \frac{a_0r^{2}}{2} \left( 1 - 10\alpha^{-1} r^{2} - \frac{30}{a_{0}}  \frac{\sigma_{\max}}{\sigma_{\min}} \sqrt{sn^*} \Delta_1 \right) \\
	& - \lambda \sqrt{Ks} \trinorm V^{*} \trinorm_{\Frob} r - 4\left\{ \left(\frac{\sigma_{\max}}{\sigma_{\min}} \right)^{2} \sqrt{s} \trinorm V^{*} \trinorm_{\Frob}  + \frac{\sigma_{\max}}{\sigma_{\min}} \sqrt{K} \right\} \Delta_{1} r \, .
\end{align*}
Since \( \bar{r} \leq 0.2 \sqrt{\alpha} \) implies \( 10 \alpha^{-1} r^{2} \leq \frac{1}{3} \), it holds by \eqref{s_n_star_logN_over_Tdelta} 
\[
	1 - 10\alpha^{-1} r^{2} - \frac{30}{a_{0}}  \frac{\sigma_{\max}}{\sigma_{\min}} \sqrt{sn^*} \Delta_1 \geq \frac{1}{2} .
\]
Therefore, after dividing by \( r \), we get that such optimal clustering must satisfy
\[
	\frac{a_0}{4} r \leq \lambda \sqrt{Ks} \trinorm V^{*} \trinorm_{\Frob} + 4\left\{ \left(\frac{\sigma_{\max}}{\sigma_{\min}} \right)^{2} \sqrt{s} \trinorm V^{*} \trinorm_{\Frob}  + \frac{\sigma_{\max}}{\sigma_{\min}} \sqrt{K} \right\} \Delta_{1} .
\]
Recalling that \( \trinorm V^{*} \trinorm_{\Frob} \leq \sqrt{K}  \), \( \Delta_{1} = C \sigma_{\max} \sqrt{\frac{\log N}{T p_{\min}^{2}}} \), and \( \Delta_{2} = C \sqrt{\frac{s \log N}{T p_{\min}^{2}}}  \) yields the result.

\end{proof}

Now we are ready to finalize the proof of Theorem~\ref{main_thm}. Firstly, we need to show that the clustering \( \hat{\CC} \) from the lemma above is locally optimal. By Lemma~\ref{lemma_one_change}, any neighbouring to it clustering \( \CC' \) satisfies
\( \trinorm Z_{\CC'} - Z_{\hat{\CC}} \trinorm_{\Frob} \leq \frac{2}{\sqrt{\alpha N / K}} \). Therefore,
\[
	\trinorm Z_{\CC'} - Z_{\CC^{*}} \trinorm_{\Frob} \leq \frac{C}{a_{0}} \left(\frac{\sigma_{\max}}{\sigma_{\min}} \right)^{2} \lambda  K\sqrt{s} + 2 \alpha^{-1/2} \sqrt{\frac{K}{N}} \, ,
\]
and it is enough to check that this value is at most \( \bar{r} \). We check that each of the terms is at most \( \bar{r} / 2 \). For the first one, it is sufficient to have
\begin{align*}
	\frac{C}{a_{0}} \left(\frac{\sigma_{\max}}{\sigma_{\min}} \right)^{2} \alpha^{-1/2} \lambda  K\sqrt{s} & \leq 0.09 , \\
	\frac{C^2}{a_{0}^2} \left(\frac{\sigma_{\max}}{\sigma_{\min}} \right)^{4} \lambda \left(2 \sigma_{\max} \alpha^{-1/2} + \sqrt{n^{*}}\Delta_1  \right)  K^2 s & \leq 0.012,
\end{align*}
and both are satisfied due to the upper bound \( \lambda \leq c \kappa^{-4} (a_{0}^{2} / \sigma_{\max}) K^{-2} s^{-1} \) and the requirement \( \sqrt{ \frac{s n^* \log N}{T p_{\min}^{2}}} \leq c \). For the second term we need
\begin{align*}
	\alpha^{-1} \frac{K}{N} \leq 0.008 \alpha,
	\qquad
	\alpha^{-1} \left(2 \sigma_{\max} \alpha^{-1/2} + \sqrt{n^{*}}\Delta_1  \right) \frac{K}{N} \leq \lambda ,
\end{align*}
both are satisfied once \( N \geq C \alpha^{2} K \) and \( \lambda \geq C \sigma_{\max} \alpha^{-3/2} \frac{K}{N} \).

Moreover, by Lemma~\ref{exact:lem} we have for \( \hat{\Theta} = Z_{\hat{\CC}} \hat{V}_{\hat{\CC}, \lambda} \)
\begin{align*}
	\trinorm \hat{\Theta} - \Theta^{*} \trinorm_{\Frob} & \leq \trinorm Z_{\hat{\CC}}(\hat{V}_{\hat{\CC}, \lambda} - V^{*})^{\T} \trinorm_{\Frob} + \trinorm (Z_{\hat{\CC}} - Z^{*}) V^{*} \trinorm_{\Frob} \\
	& \leq 3 \sigma_{\min}^{-1} \sqrt{Ks} \lambda + \frac{C}{a_{0}} \left(\frac{\sigma_{\max}}{\sigma_{\min}} \right)^{2} \gamma   K\sqrt{s} \lambda,
\end{align*}
which finishes the proof.

\bibliography{mybib}

\appendix

\section{Proof of Theorems~\ref{missing_cov_est:prop} and \ref{cros_cov_missing:prop}}
\label{section:covariance}
Recall that we have a time series,
\begin{equation}\label{lin_proc:def}
	Y_{t} = \sum_{k \geq 0} \Theta^{k} W_{t-k} ,
	\qquad
	t \in \Z,
\end{equation}
where \( W_{t} \in \R^{N} \), \( t \in \Z \) are independent vectors with \( \E W_{t} = 0 \) and \( \Var(W_{t}) = S \). 
We also have \( \trinorm \Theta \trinorm_{\op} \leq \gamma \) for some \( \gamma  < 1 \), and
the covariance \( \Sigma = \Var(Y_{t}) \) reads as
\begin{equation*}
	\Sigma = \sum_{k \geq 0} \Theta^{k} S [\Theta^{k}]^{\T} .
\end{equation*}
We have the observations
\begin{equation}
	Z_{t} = (\delta_{1t} Y_{1t}, \dots, \delta_{Nt} Y_{Nt})^{\T},
	\qquad
	t = 1, \dots, T,
\end{equation}
where \( \delta_{it} \sim \mbox{Be}(p_i) \) are independent Bernoulli random variables for every \( i = 1, \dots, N \) and \( t = 1, \dots, T \) and some \( p_{i} \in (0, 1] \). 

The proofs of both statements are based on the following version of the Bernstein matrix inequality, which does not require bounded summands. Recall, that for a random variable \( X \in \R \) the value
\[
	\| X \|_{\psi_{j}} = \inf\left\{ C > 0 : \E \exp\left(\left|\frac{X}{C}\right|^{j} \right) \leq 2 \right\}
\]
denotes a \( \psi_j \)-norm. For \(j=1\) the norm is referred to as \emph{subexponential} and for \( j =2\) as \rev{\emph{sub-Gaussian}}, see Definition~\ref{subgaus_def}.

\begin{theorem}[\cite{klochkov2018uniform}, Proposition~4.1]\label{nikita_thm}
Suppose, the matrices \(A_t \) for \( t = 1, \dots, T \) are independent and let \( M = \max_{t} \bigl\| \trinorm A_t \trinorm_{\op} \bigr\|_{\psi_1}  \) is finite. Then, \( S_T = \sum_{t = 1}^{T} A_t \) satisfies for any \( u \geq 1 \)
\[
	\P\left[ \trinorm S_T - \E S_T \trinorm_{\op} > C\left\{ \sqrt{\sigma^{2} (\log N + u)} + M \log T (\log N + u)  \right\}  \right] \leq e^{-u},
\]
where \( \sigma^{2} = \trinorm \sum_{t = 1}^{T} \E A_t^{\T} A_t \trinorm_{\op} \vee \trinorm \sum_{t = 1}^{T} \E A_t A_t^{\T} \trinorm_{\op} \) and \(C\) is an absolute constant.
\end{theorem}
\def\Ih{\hat{I}}

Both \cite{Lounici14} and \cite{klochkov2018uniform} assume that the probabilities of the observations are given. Using Chernov's bound for the difference 
\[
	\ph_i - p_i = \frac{1}{N} \sum_{t = 1}^{T} \delta_{it} - \E \delta_{it},
\]
and applying the union bound, we derive that with probability at least \( 1 - \frac{1}{2} e^{-u}\) it holds that
\begin{equation}\label{chernov_union}
	\max_{i \leq N} \left|1 - \frac{\hat{p_i}}{p_i} \right| \leq \sqrt{2 \frac{\log(4N) + u}{T p_{\min}}} + \frac{\log(4N) + u}{T p_{\min}} \, .
\end{equation}
Notice that what appears on the right-hand side of the above display is dominated by the error that appears in Theorems~\ref{missing_cov_est:prop} and \ref{cros_cov_missing:prop}. Consider the auxiliary estimators
\begin{align*}
	\tilde{\Sigma} &= \diag\{\pv\}^{-1} \Diag(\Sigma^{*}) + \diag\{\pv\}^{-1} \Off(\Sigma^{*}) \diag\{\pv\}^{-1},
	\\
	\tilde{A} &= \diag\{\pv\}^{-1} A^{*} \diag\{\pv\}^{-1} .
\end{align*}
Then, we have for \( \hat{I} = \diag\{ \phv \}^{-1} \diag\{\pv \} \) that
\begin{align*}
	\hat{\Sigma} &= \hat{I} \Diag(\tilde{\Sigma}) + \hat{I} \Off(\tilde{\Sigma}) \hat{I},
	\qquad
	\hat{A} = \hat{I} \tilde{A} \hat{I} .
\end{align*}
Given that \( \tfrac{\log(4N) + u}{T p_{\min}} \leq \frac{1}{2} \), we easily get that by \eqref{chernov_union},
\begin{equation}
	\trinorm \hat{I} - I \trinorm_{\op} \leq \delta = 3 \sqrt{\frac{\log(4N) + u}{T p_{\min}}} . 
\label{chernov_inverted}
\end{equation}
with the corresponding probability. In this case, we have
\begin{align*}
	\trinorm \hat{A} - A \trinorm_{\op} &\leq \trinorm \Ih \tilde{A} \Ih - A \trinorm_{\op} \\
	& \leq \trinorm \Ih \trinorm_{\op}^{2} \trinorm \tilde{A} - A \trinorm_{\op} + \trinorm \Ih A \Ih - A \trinorm_{\op} \\
	& \leq 
	(1 + \delta)^{2} \trinorm \tilde{A} - A \trinorm_{\op} + 2 (1 + \delta) \delta \trinorm A \trinorm_{\op} \, .
\end{align*}
Similarly,
\begin{align*}
	\trinorm \Diag(\hat{\Sigma}) - \Diag(\Sigma) \trinorm_{\op} & \leq (1 + \delta) \trinorm \Diag(\tilde{\Sigma}) - \Diag(\Sigma) \trinorm_{\op} + \delta \trinorm \Sigma \trinorm_{\op} \, ,
	\\
	\trinorm \Off(\hat{\Sigma}) - \Off(\Sigma) \trinorm_{\op} & \leq (1 + \delta)^{2} \trinorm \Off(\tilde{\Sigma}) - \Off(\Sigma) \trinorm_{\op} + 4 (1 + \delta) \delta \trinorm \Sigma \trinorm_{\op} \, .
\end{align*}
{Recall that \( S = \Var(W_t) \), and from \eqref{sigma_through_theta_s} we can easily derive that \( \trinorm \Sigma \trinorm_{\op} \leq \tfrac{1}{1 - \gamma^2} \trinorm S \trinorm_{\op} \), where \( \trinorm \Theta^* \trinorm_{\op} \leq \gamma < 1 \). Correspondingly, from \( A = \Theta^{*} \Sigma  \), it follows \( \trinorm A \trinorm_{\op} \leq \tfrac{\gamma}{1- \gamma^2} \trinorm S \trinorm_{\op} \). The condition \eqref{sparse_projector} ensures that \( \delta \leq 3 \), and both the theorems now follow from the proposition below.}

\begin{proposition}\label{prop_sigma_tilde}
Under the conditions of Theorems~\ref{missing_cov_est:prop} and \ref{cros_cov_missing:prop}, for any two projectors with \(P, Q\) with ranks \(M_1, M_2\), respectively, we have that for any \(u > 0\), with probability at least \( 1 - e^{-u}\),
\begin{equation}\label{diag_sigma_tilde}
	\begin{aligned}
	\trinorm P(&\Diag(\tilde{\Sigma}) - \Diag(\Sigma)) Q \trinorm_{\op}\\
	& \leq C \trinorm S \trinorm_{\op} \left(\sqrt{\frac{(M_1 \vee M_2) (\log N + u)}{T p_{\min}^2}} \bigvee \frac{\sqrt{M_1 M_2}(\log N + u) \log T}{T p_{\min}^2}   \right).
	\end{aligned}
\end{equation}
and
\begin{equation}\label{off_sigma_tilde}
	\begin{aligned}
	\trinorm P(&\Off(\tilde{\Sigma}) - \Off(\Sigma)) Q \trinorm_{\op} \\
	& \leq C \trinorm S \trinorm_{\op} \left(\sqrt{\frac{(M_1 \vee M_2) (\log N + u)}{T p_{\min}^2}} \bigvee \frac{\sqrt{M_1 M_2}(\log N + u) \log T}{T p_{\min}^2}   \right).
	\end{aligned}
\end{equation}
Moreover, with probability at least \( 1 - e^{-u}\) we have that,
\begin{equation}\label{A_tilde}
\begin{aligned}
\trinorm P(&\tilde{A} - A) Q \trinorm_{\op} \\
& \leq C \trinorm S \trinorm_{\op} \left(\sqrt{\frac{(M_1 \vee M_2) (\log N + u)}{T p_{\min}^2}} \bigvee \frac{\sqrt{M_1 M_2}(\log N + u) \log T}{T p_{\min}^2}   \right).
\end{aligned}
\end{equation}
Here, \( C= C(\gamma, L)\) only depends on \( \gamma \) and \(L \).
\end{proposition}

We first derive Theorems~\ref{missing_cov_est:prop} and \ref{cros_cov_missing:prop} from the above proposition. Then the rest of the section will be devoted to the proof of the above proposition.

\begin{proof}[Proof of Theorems~\ref{missing_cov_est:prop} and \ref{cros_cov_missing:prop}]
Observe that
\begin{align*}
	\trinorm P(\Diag(\hat{\Sigma}) - \Diag(\Sigma)) Q \trinorm_{\op} & = \trinorm P(\Ih \Diag(\tilde{\Sigma}) - \Diag(\Sigma)) Q \trinorm_{\op} \\
	& \leq
	\trinorm P (\Ih - I) \Diag(\tilde{\Sigma}) Q \trinorm_{\op} + \trinorm P(\Diag(\tilde{\Sigma}) - \Diag(\Sigma)) Q \trinorm_{\op}
\end{align*}
The last term of the right-hand side is controlled by \eqref{diag_sigma_tilde}. As for the first one, let \( \Lambda \) be the support of \( P \) in accordance with Definition~\ref{def_projector_sparsity}, and set \( \Pi_{\Lambda} = \sum_{i \in \Lambda} \ev_i \ev_i^{\T} \), so that \( P = P \Pi_{\Lambda} \) and \( \Rank(\Pi_{\Lambda}) = K_1 \). Moreover, \( \Pi_{\Lambda} \) is diagonal, therefore, \( \Pi_{\Lambda} \Ih = \Ih \Pi_{\Lambda} \). This, we have
\begin{align*}
	\trinorm P(\Ih - I) \Diag({\Sigma}) Q \trinorm_{\op} &= \trinorm P(\Ih - I) \Pi_{\Lambda} \Diag({\Sigma}) Q \trinorm_{\op} \\ &\leq \delta (\trinorm \Sigma \trinorm_{\op} + \trinorm \Pi_{\Lambda} (\Diag(\tilde{\Sigma}) - \Diag({\Sigma})) Q \trinorm_{\op}) .
\end{align*}
By \eqref{diag_sigma_tilde} and \eqref{sparse_projector} we have that with probability at least \( 1 - \frac{1}{8} e^{-u}  \),
\[
	\trinorm \Pi_{\Lambda} (\Diag(\tilde{\Sigma}) - \Diag({\Sigma})) Q \trinorm_{\op} \leq C_1 \trinorm \Sigma \trinorm_{\op},
\]
and in addition,  \( \trinorm \Sigma \trinorm_{\op} \leq (1 - \gamma^2)^{-1} \trinorm S \trinorm_{\op} \). 
Furthermore,
\begin{align*}
	\trinorm P(\Off(\hat{\Sigma}) - \Off(\Sigma)) Q \trinorm_{\op} = & \trinorm P(\Ih \Off(\tilde{\Sigma}) \Ih - \Off(\Sigma))  Q \trinorm_{\op} \\
	\leq &
	\trinorm P (\Off(\tilde{\Sigma}) - \Off(\Sigma)) Q \trinorm_{\op} + \trinorm P (\Ih - I) \Off(\tilde{\Sigma}) \Ih Q \trinorm_{\op} \\
	& + \trinorm P \Off(\tilde{\Sigma}) (\Ih - I) Q \trinorm_{\op}
\end{align*}
Let \( \Lambda' \) be the sparsity pattern for the projector \( Q \) and \( \Pi_{\Lambda'} \) is the corresponding diagonal projector. Then we apply \eqref{off_sigma_tilde} to \( \Pi_{\Lambda}(\Off(\tilde{\Sigma}) - \Off(\Sigma)) \Pi_{\Lambda'}  ) \) so that provided with \eqref{sparse_projector}, we have with probability at least \( 1 - \tfrac{1}{8} e^{-u} \),
\[
	\trinorm \Pi_{\Lambda} \Off(\tilde{\Sigma}) \Pi_{\Lambda'} \trinorm_{\op} \leq C_2 \trinorm S \trinorm_{\op}.
\]
Using that \( \Pi_{\Lambda'}, \Pi_{\Lambda} \) commute with the diagonal matrices \( \Ih \), \( \Ih - I \), and given that \( \delta \leq 1 \), we get
\[
	\trinorm P \Off(\tilde{\Sigma}) (\Ih - I) Q \trinorm_{\op} + \trinorm P (\Ih - I) \Off(\tilde{\Sigma}) \Ih Q \trinorm_{\op} \leq C_3 \delta \trinorm S \trinorm_{\op} .
\]
Applying \eqref{diag_sigma_tilde} to \( P(\Diag(\tilde{\Sigma}) - \Diag(\Sigma)) Q  \) with probability \( 1 - \tfrac{1}{8} e^{-u} \) and \eqref{off_sigma_tilde} to \( P(\Off(\tilde{\Sigma}) - \Off(\Sigma)) Q \), and putting the diagonal and off-diagonal terms together, we get that, with probability at least \( 1 - \tfrac{4}{8} e^{-u} \), it holds that
\begin{align*}
	\trinorm P(\hat{\Sigma} &- \Sigma) Q \trinorm_{\op} \\
	& \leq C \trinorm S \trinorm_{\op} \left( \delta + \sqrt{\frac{(M_1 \vee M_2) (\log N + u)}{T p_{\min}^2}} \bigvee \frac{\sqrt{M_1 M_2}(\log N + u) \log T}{T p_{\min}^2} \right)
\end{align*}
It remains to notice that the bound \eqref{chernov_inverted} for \( \delta \) holds with probability at least \( 1 - \tfrac{1}{2} e^{-u} \), and the corresponding \( \delta \) is dominated by the remaining error term. This concludes the proof of Theorem~\ref{missing_cov_est:prop}.

Theorem~\ref{cros_cov_missing:prop} can be proved treating \( P(\hat{A} - A)Q \) similarly to the off-diagonal case above.
\end{proof}

\def\deltav{{\boldsymbol\delta}}
We now turn to the proof of Proposition~\ref{prop_sigma_tilde}. Let \( \deltav_{t} = (\delta_{t1}, \dots, \delta_{tN})^{\T} \) denotes the vector with Bernoulli variables from above corresponding to the time point \( t \). In what follows we consider the following matrices,
\[
	A_{t, t'}^{k, j} = \diag\{\deltav_{t}\} \Theta^{k} W_{t - k} W_{t' - j}^{\T} [\Theta^{j}]^{\T} \diag\{\deltav_{t'}\},
\]
so that since \( Z_{t} = \sum_{k \geq 0} \diag\{\deltav_t\} \Theta^{k} W_{t - k} \), we have 
\[
	Z_{t}Z_{t}^{\T} = \sum_{k, j \geq 0} \diag\{\deltav_{t}\} \Theta^{k} W_{t - k} W_{t - j}^{\T} [\Theta^{j}]^{\T} \diag\{\deltav_{t}\} = \sum_{k, j \geq 0} A_{t, t}^{k, j}.
\]
Therefore, the decomposition takes place
\begin{equation}\label{S_kj_def}
	\Sigma^{*} = \sum_{k, j \geq 0} S_{k, j}, \qquad S_{k, j} = \frac{1}{T} \sum_{t = 1}^{T} A_{t, t}^{k, j},
\end{equation}
and we shall analyze the sum \(S_{k,j}\) for every pair of \( k, j \geq 0 \) separately. We first introduce two technical lemmata. In what follows we assume w.l.o.g. that \( \trinorm S \trinorm_{\op} = 1 \), since if we scale it, all the covariances and estimators scale correspondingly.

\begin{lemma}\label{opnorm_psi1:lem}
Under the assumptions of Theorem~\ref{missing_cov_est:prop} it holds,
\begin{align*}
	\| \trinorm P \diag\{ \pv \}^{-1} \Diag(A_{t, t'}^{k, j}) Q \trinorm_{\op} \|_{\psi_1} &\leq C p_{\min}^{-1} \sqrt{M_1 M_2} \gamma^{k + j},
	\\
	\| \trinorm P \diag\{ \pv \}^{-1} \Off(A_{t, t'}^{k, j}) \diag\{ \pv \}^{-1} Q \trinorm_{\op} \|_{\psi_1} &\leq C p_{\min}^{-2} \sqrt{M_1 M_2} \gamma^{k + j},
\end{align*}
with some \( C = C(L) > 0 \).
\end{lemma}
\begin{proof}
Denote for simplicity \( \xv = \Theta^{k} W_{t-k} \), \( \yv = \Theta^{j} W_{t'-j} \), as well as \( \xv^{\delta} = \diag\{\deltav_{t} \} \xv \), \( \yv^{\delta} = \diag\{\deltav_{t} \} \yv \), such that \( A_{t, t'}^{k, j} = \xv^{\delta} [\yv^{\delta}]^{\T} \). Since \( W_t \) are \rev{sub-Gaussian} and \( \trinorm \Theta^{k} S \Theta^{k} \trinorm_{\op} \leq \gamma^{2k} \), we have for any \( \uv \in \R^{N} \)
\begin{equation}\label{x_subexp}
\log \E \exp( \uv^{\T} \xv ) \leq C' \gamma^{2 k} \| \uv \|^2,
\end{equation}
and since \( \delta_t \) takes values in \( [0, 1]^{N} \), same takes place for \( \xv^{\deltav} \). By Theorem~{2.1} in \cite{hsu2012tail} it holds for any matrix \( A \) and vector \( \uv \in \R^{N}  \), 
\begin{equation}\label{Px_psi_2:aux}
	\| \| A \xv^{\delta} \| \|_{\psi_2} \leq C'' \gamma^{k} \trinorm A \trinorm_{\Frob},
	\qquad
	\| \uv^{\T} \xv^{\delta} \|_{\psi_2} \leq C'' \gamma^{k}  \| \uv \|,
\end{equation}
and, similarly,
\[
	\| \| A \yv^{\delta} \| \|_{\psi_2} \leq C'' \gamma^{j} \trinorm A \trinorm_{\Frob},
	\qquad
	\| \uv^{\T} \yv^{\delta} \|_{\psi_2} \leq C'' \gamma^{j} \| \uv \|.
\] 
We first deal with the diagonal term. Let \( P = \sum_{i = 1}^{M_1} \uv_j \uv_j^{\T} \) be its eigen-decomposition with \( \| \uv_j \| = 1 \), then
\[
\begin{aligned}
	\| \trinorm P \diag(\xv^{\delta}) \trinorm_{\op} \|_{\psi_2}^2 =& \| \trinorm \diag(\xv^{\delta}) P \diag(\xv^{\delta}) \trinorm_{\op} \|_{\psi_1} \leq \sum_{j = 1}^{M_1} \| \trinorm \diag(\xv_{\delta}) \uv_{j} \uv_{j}^{\T} \diag(\xv^{\delta}) \trinorm_{\op} \|_{\psi_1} \\
	= & \sum_{j = 1}^{M_1} \| \| \diag(\uv_j) \xv^{\delta}\| \|_{\psi_2}^{2},
\end{aligned}
\]
where each term in the latter is bounded by \( \gamma^{2k} \) due the fact that \( \trinorm \diag(\uv_{j}) \trinorm_{\Frob} = 1 \). Summing up and taking square root, we arrive at \( \bigl\| \trinorm P \diag(\xv^{\delta}) \trinorm_{\op} \bigr\|_{\psi_2} \leq  \sqrt{C''M_1} \gamma^{k} \). Taking into account similar bound for \( Q \diag(\yv^{\delta}) \), we have by H\"older inequality
\[
\begin{aligned}
	\| \trinorm P \diag\{\delta\}^{-1} \diag(\xv^{\delta}) \diag(\yv^{\delta}) Q \trinorm_{\op} \|_{\psi_1} \leq & p_{\min}^{-1} \| \trinorm P \diag(\xv^{\delta}) \trinorm_{\op} \bigr\|_{\psi_2} \| \trinorm Q \diag(\yv^{\delta}) \trinorm_{\op} \|_{\psi_2} \\
	\leq &
	C'' \sqrt{M_1 M_2} \gamma^{k + j},
\end{aligned}
\]
which yields the bound for the diagonal. As for the off-diagonal, consider first the whole matrix,
\[
	\| \trinorm P \xv^{\delta} [\yv^{\delta}]^{\T} Q \trinorm_{\op} \|_{\psi_1}  
	\leq 
	\| \| P \xv^{\delta} \| \|_{\psi_2} 
	\| \| Q \yv^{\delta} \| \|_{\psi_2} \leq (C'')^2 \sqrt{M_1 M_2} \gamma^{j + k},
\]
and since \( \Off(A_{t, t'}^{j,k}) = A_{t, t'}^{j, k} - \Diag(A_{t, t'}^{j, k}) \), the bound follows from the triangular inequality.
\end{proof}

The following technical lemma will help us to upper-bound \( \sigma^{2} \) in Theorem~\ref{nikita_thm}.

\def\deltao{\overline{\delta}}
\begin{lemma}\label{delta_squares}
Let \( \delta_1, \dots, \delta_N \) consists of independent Bernoulli components with probabilities of success \( p_1, \dots, p_N \) and set \( p_{\min} = \min_{i \leq N} p_{i} \). Let \( \av, \bv \in \R^{N} \) be two arbitrary vectors. It holds,
\begin{align*}
	\E \left( \sum_{i} \frac{\delta_i}{p_i} a_i b_i \right)^{2} \leq & p_{\min}^{-1} \| \av \|^{2} \| \bv \|^{2},
	\\
	\E \left( \sum_{i \neq j} \frac{\delta_i \delta_j}{p_i p_j} a_i b_j \right)^{2} \leq & 32 p_{\min}^{-2} \| \av \|^{2} \| \bv \|^{2} + 4 \left(\sum_{i} a_i \right)^2\left(\sum_{i} b_i \right)^2 .
\end{align*}
Additionally, if \( \delta_1', \dots, \delta_N' \) are  independent copies of \( \delta_1, \dots, \delta_N \), it holds
\[
	\E \left( \sum_{i, j} \frac{\delta_i \delta_j'}{p_i p_j} a_i b_j \right)^{2} \leq 4 p_{\min}^{-2} \| \av \|^{2} \| \bv \|^{2} + 4 \left(\sum_{i} a_i \right)^2\left(\sum_{i} b_i \right)^2.
\]
\end{lemma}

\begin{proof}
It holds,
\begin{align*}
	\E \left( \sum_{i} \frac{\delta_i}{p_i} a_i b_i \right)^{2} = & \sum_{i, j} \E \frac{\delta_i \delta_j}{p_ip_j} a_ib_ia_jb_j = \sum_{i, j} \{ 1 + \Ind(i = j) (p_{i}^{-1} - 1) \} a_ib_i a_j b_j  \\
	\leq & \left( \sum_{i} a_i b_i \right)^{2} + (p_{\min}^{-1} - 1) \sum_{i} a_i^2 b_i^2 \\
	\leq & \| \av \|^{2} \| \bv \|^{2} + (p_{\min}^{-1} - 1) \| \av \|^{2} \| \bv \|^{2}.
\end{align*}

To show the second inequality we use decoupling (Theorem~6.1.1 in \cite{vershynin2016high}) and the trivial inequality \( (x + y)^{2} \leq 2x^2 + 2y^2 \),
\begin{equation}\label{dcoup_deltas}
\begin{aligned}
	\E \left( \sum_{i \neq j} \frac{\delta_i \delta_j}{p_i p_j} a_i b_j \right)^{2} \leq & 2 \left( \sum_{i \neq j} a_i b_j \right)^{2} + 2 \E \left( \sum_{i \neq j} \frac{(\delta_i - p_i) (\delta_j - p_j)}{p_i p_j} a_i b_j \right)^{2} \\
	\leq & 2 \left( \sum_{i \neq j} a_i b_j \right)^{2} + 32 \E \left( \sum_{i \neq j} \frac{(\delta_i - p_i) (\delta_j' - p_j)}{p_i p_j} a_i b_j \right)^{2} .
\end{aligned}
\end{equation}
Denote for simplicity \( \overline{\delta}_{i} = \delta_{i} - p_{i} \) and \( \overline{\delta}_{i}' =\delta_{i}' - p_{i}\). Since the latter are centered we have,
\begin{align}\label{delta_i_delta_j_prim}
	\E \left( \sum_{i \neq j} \frac{\deltao_i \deltao_j'}{p_i p_j} a_i b_j \right)^{2} = \sum_{\substack{i \neq j \\ k \neq l}} \frac{\E\deltao_{i} \deltao_{k}}{p_i p_k} \frac{\E \deltao_{j}' \deltao_{l}'}{p_{j} p_{j}} a_i a_k b_j b_l
\end{align}
note that the expectation \( \E \deltao_{i} \deltao_{k} \) is only non-vanishing when \( i = k \), in which case it holds \( \E \deltao_{i}^{2} = p_i - p_{i}^{2} \). Taking into account similar property of \( \E \deltao_{j}' \deltao_{l}' \) we have that the sum above is equal to
\[
	\sum_{i \neq j} \frac{(p_i - p_i^2)(p_j - p_{j}^{2})}{p_i^2 p_j^2} a_i^2 b_j^2 \leq (p_{\min}^{-1} - 1)^{2} \sum_{i, j} a_{i}^2 b_j^2 \leq (p_{\min}^{-1} -1)^{2} \| \av \|^{2} \| \bv \|^{2}.
\]
It is left to note that
\[
	\left(\sum_{i \neq j} a_i b_j \right)^2 \leq 2 \left(\sum_{i,  j} a_i b_j  \right)^2 + 2 \left( \sum_{i} a_i b_j  \right)^{2} \leq 2 \left(\sum_{i} a_i \right)^2\left(\sum_{i} b_i \right)^2 + 2 \| \av \|^{2} \| \bv \|^{2},
\]
which recalling \eqref{dcoup_deltas} and noting that \( 32 (p_{\min}^{-1} - 1)^{2} + 4 \leq 32 p_{\min}^{-2} \) for \( p_{\min} \in [0, 1]  \), completes the proof.

Similarly to \eqref{delta_i_delta_j_prim}, we can show the third inequality.
\end{proof}

Now we apply the Bernstein matrix inequality to the sum \( S_{kj} \) defined in \eqref{S_kj_def}, dealing separately with diagonal and off-diagonal parts. After that, we present the proof of Theorem~\ref{missing_cov_est:prop}.

\begin{lemma}\label{diag_S_kj:lem}
Under the assumptions of Theorem~\ref{missing_cov_est:prop}, it holds for any \( u \geq 1 \)  with probability at least \(  1- e^{-u}  \)
\[
\begin{aligned}
	\trinorm P \diag\{\pv\}^{-1}&(\Diag(S_{k, j}) - \E \Diag(S_{k, j})) Q \trinorm_{\op}  \\
	& \leq
	C \gamma^{k + j} \left( \sqrt{\frac{M_1 \vee M_2 (\log N + u)}{T p_{\min}}} \bigvee \frac{\sqrt{M_1 M_2}(\log N + u)}{T p_{\min}}   \right)
\end{aligned}
\]
where \( C = C(K) \) only depends on \( K \).
\end{lemma}

\begin{proof}
Note that,
\[
	P \diag\{\pv\}^{-1} \Diag(S_{kj}) Q = T^{-1} \sum_{t = 1}^{T} A_{t},
	\qquad
	A_t = P \diag\{\pv\}^{-1} \Diag(A^{k,j}_{t, t}) Q .
\]
By Lemma~\ref{opnorm_psi1:lem} we have \( \| \trinorm A_t \trinorm_{\op} \|_{\psi_1} \leq C p_{\min}^{-1} \sqrt{M_1 M_2} \gamma^{k + j} \). Moreover, using decomposition \( Q = \sum_{j = 1}^{M_2} \uv_{j} \uv_{j} \), we have
\[
\begin{aligned}
	\trinorm \E A_t A_t^{\T} \trinorm_{\op} \leq & \trinorm \E \diag\{\pv\}^{-1} \Diag(A^{k,j}_{t, t}) Q \Diag(A^{k,j}_{t, t}) \diag\{\pv\}^{-1} \trinorm_{\op} \\
	\leq & \sum_{j = 1}^{M_2} \trinorm \E \diag\{\pv\}^{-1} \Diag(A^{k,j}_{t, t}) \uv_{j} \uv_{j}^{\T} \Diag(A^{k,j}_{t, t}) \diag\{\pv\}^{-1} \trinorm_{\op} \\
	\leq &
	\sum_{j = 1}^{M_2} \sup_{\| \gammav \| = 1} \E (\gammav^{\T} \diag\{\pv\}^{-1} \Diag(A^{k,j}_{t, t}) \uv_{j})^{2}
\end{aligned}
\]
By definition, \( \Diag(A^{k, j}_{t, t}) = \diag\{\delta_{ti} x_i y_i \}_{i = 1}^{N} \) for \( \xv = \Theta^{k} W_{t - k} \), \( \yv = \Theta^{j} W_{t - j} \). Let \( \E_{\delta} \)~denotes the expectation w.r.t. the Bernoulli variables and conditioned on everything else. Setting \( \av = (x_1 \gamma_1, \dots, x_N \gamma_N)^{\T}  \) and \( \bv = (y_1 u_1, \dots, y_N u_N)^{\T} \), we have by the first inequality of Lemma~\ref{delta_squares},
\begin{align*}
	\E (\gammav^{\T} \diag\{\pv\}^{-1} \Diag(A^{k,j}_{t, t}) \uv_{j})^{2} & = \E \E_{\delta} \left( \sum_{i} \gamma_i x_i \frac{\delta_{ti}}{p_i} y_i u_i \right)^{2} \\
	& \leq p_{\min}^{-1} \E \| \av \|^{2} \| \bv \|^{2} \\
	& \leq p_{\min}^{-1} \E^{1/2} \| \av \|^{4} \E^{1/4} \| \bv \|^{4} .
\end{align*}
Observe that,
\[
	\| \av \|^{2} = \sum_{i} \gamma_i^{2} x_{i}^{2} = \xv^{\T} \diag\{\gammav\}^{2} \xv,
\]
so since \( \Tr(\diag\{\gammav\}^{2}) = 1 \) and due to \eqref{x_subexp} and by Theorem~2.1 \cite{hsu2012tail}, it holds \( \E^{1/2} \| \av \|^{4} \leq  \| \| \av \|^{2} \|_{\psi_1} \leq C' \gamma^{2k} \). Similarly, it holds \( \E^{1/2} \| \av \|^{4} \leq C' \gamma^{2j} \), which together implies
\[
	\trinorm \E A_t A_t^{\T} \trinorm_{\op} \vee \trinorm \E A_t^{\T} A_t^{\T} \trinorm_{\op} \leq C'' M_2 \vee M_1 \gamma^{2k + 2j}.
\]

Now notice that \( A_t \) is not necessary an independent sequence, as \( A_{t} \) depends directly on \( (W_{t - k}, W_{t - j}, \deltav_{t}) \), which might intersect with \( t' = t + |j - k| \). However, if we take a set \( I \subset [1, T] \) such that any two \( t, t' \in I \) satisfy \( |t' - t| \neq |j - k|  \) then the sequence \( (A_{t})_{t \in I} \) is independent. We separate the whole interval \( [1, T] \) into two such independent sets,
\begin{equation}\label{split}
\begin{aligned}
	I_{1} =& \{ t \in [1, T] : \; \lceil t / |j - k| \rceil \text{ is odd }   \},
	\\
	I_{2} =& \{ t \in [1, T] : \; \lceil t/ |j - k| \rceil \text{ is even } \} \\=& [1, T] \setminus I_1.
\end{aligned}
\end{equation}
Indeed, if for \( t, t' \in I_{1} \) then \(  \lceil t / |j - k| \rceil \) and \(  \lceil t' / |j - k| \rceil \) are either equal or differ in at least two, so that in the first case we have \( |t - t'| < |j - k| \) and in the second \( |t - t'| > |j - k| \). Since both intervals have at most \( T \) elements, it holds by Theorem~\ref{nikita_thm} with probability at least \( 1 - e^{-u} \) for both \( j \),
\begin{align*}
	\trinorm \sum_{t \in I_{j}} & A_{t} - \E A_{t} \trinorm_{\op} 
	\\
	& \leq  C \gamma^{j + k}\left(  \sqrt{p_{\min}^{-1} (M_1 \vee M_2) T (\log N + u)} \vee p_{\min}^{-1} \sqrt{M_1 M_2}(\log N + u) \log T \right),
\end{align*}
so summing up the two and dividing by \( T \), we get the result.
\end{proof}

\begin{lemma}\label{off_diag_S_kj:lem}
Under the assumptions of Theorem~\ref{missing_cov_est:prop}, it holds for any \( u \geq 1 \)  with probability at least \(  1- e^{-u}  \)
\[
\begin{aligned}
\trinorm P \diag\{\pv\}^{-1}&(\Off(S_{k, j}) - \E \Off(S_{k, j}))\diag\{\pv\}^{-1} Q \trinorm_{\op}  \\
& \leq
C \gamma^{k + j} \left( \sqrt{\frac{M_1 \vee M_2 (\log N + u)}{T p_{\min}^2}} \bigvee \frac{\sqrt{M_1 M_2}(\log N + u) \log T}{T p_{\min}^2}   \right)
\end{aligned}
\]
where \( C = C(K) \) only depends on \( K \).
\end{lemma}

\begin{proof}
It holds,
\begin{align*}
P \diag\{\pv\}^{-1} \Off(S_{kj}) \diag\{\pv\}^{-1} Q &= T^{-1} \sum_{t = 1}^{T} B_{t},
\\
B_t & = P \diag\{\pv\}^{-1} \Off(A^{k,j}_{t, t}) \diag\{\pv\}^{-1} Q .
\end{align*}
By Lemma~\ref{opnorm_psi1:lem} we have \( \| \trinorm B_t \trinorm_{\op} \|_{\psi_1} \leq C p_{\min}^{-2} \sqrt{M_1 M_2} \gamma^{k + j} \). Using decomposition \( Q = \sum_{j = 1}^{M_2} \uv_{j} \uv_{j} \) with \( \| \uv_j \| = 1 \) we get that
\[
\begin{aligned}
\trinorm \E B_t B_t^{\T} \trinorm_{\op} \leq & \trinorm \E \diag\{\pv\}^{-1} \Off(A^{k,j}_{t, t}) \diag\{\pv\}^{-1} Q \diag\{\pv\}^{-1} \Off(A^{k,j}_{t, t}) \diag\{\pv\}^{-1} \trinorm_{\op} \\
\leq & \sum_{j = 1}^{M_2} \trinorm \E \diag\{\pv\}^{-1} \Off(A^{k,j}_{t, t}) \diag\{\pv\}^{-1} \uv_{j} \uv_{j}^{\T} \diag\{\pv\}^{-1} \Off(A^{k,j}_{t, t}) \diag\{\pv\}^{-1} \trinorm_{\op} \\
\leq &
\sum_{j = 1}^{M_2} \sup_{\| \gammav \| = 1} \E (\gammav^{\T} \diag\{\pv\}^{-1} \Off(A^{k,j}_{t, t}) \diag\{\pv\}^{-1} \uv_{j})^{2}
\end{aligned}
\]
Again, using the notation \( \xv = \Theta^{k} W_{t - k} \), \( \yv = \Theta^{j} W_{t - j} \) and \( \av = \diag\{ \gammav \} \xv \), \( \bv = \diag\{ \uv \} \yv \), we have \( \Off(A_{t,t}^{j, k}) = \Off(\xv \yv^{\T}) \). Therefore, by Lemma~\ref{delta_squares}
\begin{align*}
	\E (\gammav^{\T} \diag\{\pv\}^{-1} \Off(A^{k,j}_{t, t}) \diag\{\pv\}^{-1} \uv_{j})^{2} =& \E \E_{\delta} \left( \sum_{i \neq j} \gamma_i \frac{\delta_{it}}{p_i} x_i y_j \frac{\delta_{jt}}{\delta_{j}} u_{j} \right)^{2} \\
	= & \E \E_{\delta} \left( \sum_{i \neq j} \frac{\delta_{it}}{p_i} \frac{\delta_{jt}}{\delta_{j}} a_i b_j \right)^{2} \\
	\leq & 32 p_{\min}^{-2} \E \| \av \|^{2} \| \bv \|^{2} + 4 \E \left(\sum_{i} a_i \right)^2\left(\sum_{i} b_i \right)^2 .
\end{align*}
From the proof of Lemma~\ref{off_diag_S_kj:lem} we know that \( \E \| \av \|^{2} \| \bv \|^{2} \leq C' \gamma^{2k + 2j} \). Moreover, we have \( \sum_{i} a_i = \gammav^{\T} \xv \) and \( \sum_i b_i = \uv^{\T} \yv \). Thus, by \eqref{Px_psi_2:aux} it holds \( \E^{1/4} \| \gammav^{\T} \xv \|^{4} \leq \| \gammav^{\T} \xv \|_{\psi_2} \leq C' \gamma^{j} \) and, similarly, \( \E^{1/4} \| \uv^{\T} \yv \|^{4} \leq C' \gamma^{k}  \). Putting those bounds together and applying Cauchy-Schwarz inequality, we have
\[
	\trinorm \E B_t B_t^{\T} \trinorm_{\op} \leq C'' p_{\min}^{-2} M_2 \gamma^{2k + 2j}.
\]
By analogy,
\[
	\trinorm \E B_t B_t^{\T} \trinorm_{\op} \vee \trinorm \E B_t^{\T} B_t \trinorm_{\op} \leq C'' p_{\min}^{-2} M_1 \vee M_2 \gamma^{2k + 2j} .
\]
Applying the same sample splitting \eqref{split} we obtain the bound
\[
	\trinorm \sum_{t} A_t - \E A_t \trinorm_{\op} \leq C \gamma^{j + k}\left(  \sqrt{p_{\min}^{-2}(M_1 \vee M_2) T (\log N + u)} \vee p_{\min}^{-2} \sqrt{M_1 M_2}(\log N + u) \right), 
\]
which divided by \( T \) provides the result.
\end{proof}

\begin{proof}[Proof of \eqref{diag_sigma_tilde}]
Setting, 
\(
D_{k, j} = \diag\{\pv\}^{-1} \Diag(S_{k,j}),
\)
by Lemma~\ref{diag_S_kj:lem} for any \( u \geq 1 \),
\[
\trinorm P(D_{k, j} - \E D_{k, j})Q \trinorm_{\op} > C \gamma^{k + j} \left( \sqrt{\frac{M_1 \vee M_2 (\log N + u)}{T p_{\min}^2}} \bigvee \frac{\sqrt{M_1 M_2}(\log N + u)}{T p_{\min}^2}   \right)
\]
holds with probability at least \( 1 - e^{-u} \). Take a union of those bounds for every \(k, j\) with \( u = u_{k, j} = k + j + 1 + u' \) for arbitrary \(u' \geq 0 \). The total probability of complementary event is at most
\[
\sum_{k, j \geq 0} e^{-k - j - 1 - u'} = e^{-1 - u'} \left(\sum_{k \geq 0} e^{-k} \right)^{2} = e^{-u'} / (e - 1) < e^{-u'} .
\]
By definition, \( \Diag(\tilde{\Sigma}) = \diag\{\pv \}^{-1} \sum_{i, j \geq 0} S_{k,j} \).
Due to \eqref{S_kj_def} and since \( \E \tilde{\Sigma} = \Sigma \), it holds on such event
\[
\begin{aligned}
\trinorm P(\Diag(\tilde{\Sigma}) &- \Diag(\Sigma))Q \trinorm_{\op}\\ \leq & \sum_{k , j \geq 0} \trinorm P(D_{k, j} - \E D_{k, j})Q \trinorm_{\op} \\
\leq & C \sum_{k, j \geq 0} \gamma^{k + j} \left( \sqrt{\frac{M_1 \vee M_2 (\log N + u_{k, j})}{T p_{\min}^2}} \bigvee \frac{\sqrt{M_1 M_2}(\log N + u_{k, j})}{T p_{\min}^2}   \right) \\
\leq & C' \left[ \sum_{k, j \geq 0} \gamma^{k + j}\right] \left( \sqrt{\frac{(M_1 \vee M_2) \log N}{T p_{\min}^2}} \bigvee \frac{\sqrt{M_1 M_2}\log N}{T p_{\min}^2}   \right) \\
& + C \left[ \sum_{k, j} (k + j) \gamma^{k+j} \right] \left( \sqrt{\frac{(M_1 \vee M_2) u'}{T p_{\min}^2}} \bigvee \frac{\sqrt{M_1 M_2}u'}{T p_{\min}^2}   \right),
\end{aligned}
\]
which completes the proof due to the equalities
\begin{align*}
	\sum_{k, j \geq 0} \gamma^{k + j} =& \left(  \sum_{k\geq 0} \gamma^{k} \right)^{2} = \frac{1}{(1 - \gamma)^{2}} \\
	\sum_{k, j \geq 0} (k + j) \gamma^{k + j} =& 2 \sum_{k , j \geq 0} k \gamma^{k + j} = \frac{2}{(1 - \gamma)} \sum_{k \geq 0} k \gamma^{k} = \frac{2}{(1 - \gamma)^{3}}.
\end{align*}
\end{proof}

\begin{proof}[Proof of \eqref{off_sigma_tilde}]
This works similarly to the above, but applying Lemma~\ref{off_diag_S_kj:lem} to \( D_{k, j} = \diag\{\pv\}^{-1} \Off(S_{k,j}) \diag\{\pv \}^{-1} \) and using the fact that \( \Off(\tilde{\Sigma}) = \sum_{j, k \geq 0} D_{j,k} \) by definition.
\end{proof}

\begin{proof}[Proof of \eqref{A_tilde}]
Recall the definition,
\[
	A_{t, t'}^{k, j} = \diag\{\deltav_{t}\} \Theta^{k} W_{t - k} W_{t' - j}^{\T} [\Theta^{j}]^{\T} \diag\{\deltav_{t'}\}.
\]
Then, it holds
\[
	Z_{t} Z_{t + 1}^{\T} = \sum_{k ,j \geq 0} \diag\{ \deltav_{t} \} \Theta^{k} W_{t - k} W_{t + 1 - j}^{\T} [\Theta^{j}]^{\T} \diag\{\deltav_{t + 1}\} = \sum_{k, j \geq 0} A_{t, t+1}^{k, j},
\]
and the decomposition takes place,
\[
	A^* = \sum_{k, j \geq 0} S_{k, j},
	\qquad
	S_{k, j} =  \frac{1}{T - 1} \sum_{t = 1}^{T-1} A_{t, t+1}^{k, j} .
\]
We first apply the Bernstein matrix inequality for each \( S_{k, j} \) separately. Observe that
\[
	P \diag\{\pv \}^{-1} S_{k, j}\diag\{\pv \}^{-1} Q = \frac{1}{T-1} \sum_{t = 1}^{T-1} B_t,
	\qquad
	B_t = P \diag\{\pv \}^{-1} A_{t, t+1}^{k, j} \diag\{\pv \}^{-1} Q .
\]
By Lemma~\ref{opnorm_psi1:lem} each term satisfies
\[
	\max_{t} \| \trinorm B_t \trinorm_{\op} \|_{\psi_1} \leq C \sqrt{M_1 M_2} \gamma^{k + j} .
\]
Furthermore, let \( Q = \sum_{j = 1}^{M_2} \uv_{j} \uv_{j}^{\T} \) with unit vectors \( \uv_{j} \). Also, denoting \( \xv = \Theta^{k} W_{t - k} \) and \( \yv = \Theta^{k} W_{t + 1 - k} \) it holds \( A_{t, t+1}^{k, j} = \diag\{\deltav_t\} \xv \yv^{\T} \diag\{\deltav_{t + 1}\} \). Then, using Lemma~\ref{delta_squares} we have for any unit \( \gammav \in \R^{N} \),
\begin{align*}
	\E (\gammav^{\T} & \diag\{ \pv\}^{-1} A_{t, t+ 1}^{k, j} \diag\{ \pv\}^{-1} \uv_j)^{2} \\
	= & \E \E_{\delta} \left( \sum_{i, j} \gamma_i x_i \frac{\delta_{ti}}{p_i} \frac{\delta_{t + 1, j}}{p_j} y_j u_j \right)^{2} \\
	\leq & p_{\min}^{-2} \E \| \diag\{ \gammav \} \xv \|^{2} \| \diag\{\uv \} \yv \|^{2} + \E (\gammav^{\T} \xv)(\uv^{\T} \yv)^{2},
\end{align*}
which due to the subgaussianity of \( \xv \) and \( \yv \) yields,
\begin{align*}
	\E \| \diag\{ \gammav \} \xv \|^{2} \| \diag\{\uv \} \yv \|^{2} \leq & \E^{1/2} \| \diag\{ \gammav \} \xv \|^{4} \E^{1/2} \| \diag\{\uv \} \yv \|^{4} \\
	\leq & C' \gamma^{2k + 2j}  \\
	\E (\gammav^{\T} \xv)(\uv^{\T} \yv)^{2} \leq & \E^{1/2} (\gammav^{\T} \xv)^{4} \E^{1/2} (\uv^{\T} \yv)^{4} \\
	\leq & C' \gamma^{2k + 2j} .
\end{align*}
Therefore, we get that
\[
	\trinorm \E B_t B_t^{\T} \trinorm_{\op} = \sup_{\| \gammav \| = 1} \sum_{j = 1}^{M_2}  \E \left(\gammav^{\T}\diag\{ \pv\}^{-1} A_{t, t+ 1}^{k, j} \diag\{ \pv\}^{-1} \uv_j \right)^{2} \leq C'' p_{\min}^{-2} M_2 \gamma^{2k + 2j}.
\]
Using similar derivations we can arrive at
\[
	\sigma^{2} = \trinorm \E B_t B_t^{\T} \trinorm_{\op} \vee \trinorm \E B_t^{\T} B_t \trinorm_{\op} \leq C'' p_{\min}^{-2} (M_1 \vee M_2) \gamma^{2k + 2j} .
\]

Now we separate the indices \( t = 1, \dots, T \) into four subsets, such that each corresponds to a set of independent matrices \( B_t \). Since each \( B_t \) is generated by \( W_{t-k}, W_{t + 1 - j}, \deltav_{t} \), and \( \deltav_{t + 1} \), we need to ensure that none of the pair of indices \( t, t' \) from the same subset satisfies \( |t - t'| = |k - j + 1| \) nor \( |t - t'| = 1 \). It can be satisfied by the following partition. First, we split the indices into two subsets with odd and even indices, respectively, so that none of the subsets contains two indices with \( |t - t'| = 1 \). Then, both of the subsets need to be separated into two according to the scheme \eqref{split}, so that the assertion \( |t - t'| = |k - j + 1| \) is avoided within each subset. Therefore, applying the Bernstein inequality, Theorem~\ref{nikita_thm}, to each sum separately and summing them up, we get that for any \( u \geq 1 \) with probability at least \( 1 - e^{-u} \),
\begin{align*}
	\trinorm P \diag\{\deltav \}^{-1} & (S_{k, j} - \E S_{k, j}) \diag\{\deltav \}^{-1} Q \trinorm_{\op} \\
	& \leq C \left( \sqrt{p_{\min}^{-2} (M_1 \vee M_2)T (\log N + u)} \bigvee \sqrt{M_1 M_2} (\log N + u) \log T \right) .
\end{align*}
Similarly to the proof of Theorem~\ref{missing_cov_est:prop}, we take the union of those bounds for every \( i, j \) with \( u = j + k + u' \) and then the result follows.
\end{proof}

\section{LASSO and missing observations}\label{tropp_exact_recovery:sec}

Suppose, we observe a signal \( \yv \in \R^{n} \) of the form
\begin{equation*}
\yv = \Phi \bv^{*} + \epsv,
\end{equation*}
where \( \Phi = [\phiv_{1}, \dots, \phiv_{p}] \in \R^{n \times p} \) is a dictionary of words \( \phiv_{j} \in \R^{n} \) and \( \bv^{*} \) is some sparse parameter with support \( \Lambda \subset \{ 1, \dots, p \} \).
We want to recover the exact sparse representation by solving a quadratic program
\begin{equation}\label{l1pen}
\frac{1}{2} \| \yv - \Phi \bv \|^{2} + \gamma \| \bv \|_{1} \rightarrow \min_{\bv \in \R^{p}} .
\end{equation}

Denote by \( \R^{\Lambda} \) the set of vectors with elements indexed by \( \Lambda \), for \( \bv \in \R^{n} \) let \( \xv_{\Lambda} \in \R^{\Lambda} \) be the result of taking only elements indexed by \( \Lambda \). With some abuse of notation we will associate every vector \( \xv_{\Lambda} \in \R^{\Lambda} \) with a vector \( \xv \) from \( \R^{n} \) that has same coefficients on \( \Lambda \) and zeros elsewhere. Let \( \Phi_{\Lambda} = [\phiv_{j}]_{j \in \Lambda} \) be a subdictionary composed of words indexed by \( \Lambda \), and \( P_{\Lambda} \) is the projector onto the corresponding subspace.

The following sufficient conditions for the global minimizer of \eqref{l1pen} to be supported on \( \Lambda \) are due to \cite{tropp2006just}, who uses the notion of \emph{exact recovery coefficient},
\begin{equation*}\label{ERC:def}
\mbox{ERC}_{\Phi}(\Lambda) = 1 - \max_{j \notin \Lambda} \| \Phi_{\Lambda}^{+} \phiv_{j} \|_{1},
\end{equation*}
The results are summarized in the next theorem.

\begin{theorem}[\cite{tropp2006just}]\label{tropp_lemma}
	Let \( \bvt \) be a solution to \eqref{l1pen}. Suppose  that \( \| \Phi^{\T} \epsv \|_{\infty} \leq \gamma \mathrm{ERC}(\Lambda) \). Then,
	\begin{itemize}
		\item the support of \( \bvt \) is contained in \( \Lambda \);
		\item the distance between \( \bvt \) and optimal (non-penalized) parameter satisfies,
		\[
		\begin{aligned}
		\| \bvt - \bv^{*} \|_{\infty} & \leq \| \Phi_{\Lambda}^{+} \epsv \|_{\infty} + \gamma \| (\Phi_{\Lambda} \Phi_{\Lambda}^{\T})^{-1} \|_{1, \infty},
		\\
		\| \Phi_{\Lambda}(\bvt - \bv^{*}) - P_{\Lambda} \epsv \|_{2} & \leq  \gamma \| (\Phi_{\Lambda}^{+})^{\T} \|_{2, \infty} ;
		\end{aligned}
		\]
	\end{itemize}
\end{theorem}

In what follows, we want to extend this result for the possibility of using missing observations model. Observe that the program \eqref{l1pen} is equivalent to
\begin{equation*}
\frac{1}{2} \bv^{\T} [\Phi^{\T} \Phi] \bv - \bv^{\T} [\Phi^{\T} \yv] + \gamma \| \bv \|_{1} \rightarrow \min_{\bv \in \R^{p}},
\end{equation*}
so that the minimization procedure only depends on \( D = \Phi^{\T} \Phi \) and \( \cv = \Phi^{\T} \yv \). Suppose that instead we have only the access to some estimators \( \Dh \geq 0 \) and \( \cvh \) that are close enough to the original matrix and vector, respectively, which may come e.g., from missing observations model. Then, we can solve instead the following problem,
\begin{equation}\label{l1pen_missing}
\frac{1}{2} \bv^{\T} \Dh \bv - \bv^{\T} \cvh + \gamma \| \bv \|_{1} \rightarrow \min_{\bv \in \R^{p}}.
\end{equation}
In what follows, we provide a slight extension of Tropp's result towards missing observations, the proof mainly follows the same steps. 

Below, for a matrix \( D \) and 
two sets of indices \( A, B \), we denote the submatrix on those indices as \( D_{A, B} \), and for a vector \( \cv \), the corresponding subvector is \( \cv_{A} \).

\begin{lemma}\label{exact_lasso:lemma}
	Suppose that
	\[
	\| \Dh_{\Lambda^{c}, \Lambda} \Dh_{\Lambda, \Lambda}^{-1} \cvh_{\Lambda} - \cvh_{\Lambda^{c}} \|_{\infty} \leq \gamma (1 - \| \Dh_{\Lambda^{c}, \Lambda} \Dh_{\Lambda, \Lambda}^{-1}\|_{1, \infty}). 
	\]
	Then, the solution \( \bvt \) to \eqref{l1pen_missing} is supported on \( \Lambda \).
\end{lemma}
\begin{proof}
	Let \( \bvt \) be the solution to \eqref{l1pen_missing} with the restriction \( \supp(\bv) \subset \Lambda  \). Since \( \Dh \geq 0 \) this is a convex problem and therefore the solution is unique and satisfies
	\[
	\Dh_{\Lambda, \Lambda} \bvt - \cvh_{\Lambda} + \gamma \gv = 0,
	\qquad
	\gv \in \partial \| \bvt \|_{1},
	\]
	where \( \partial f(\bv) \) denotes the subdifferential of a convex function \( f \) at a point \( \bv \), in the case of \( \ell_1 \) norm we have \( \| \gv \|_{\infty} \leq 1 \). Thus,
	\begin{equation}\label{bvt:subdiff}
	\bvt = \Dh_{\Lambda, \Lambda}^{-1} \cvh_{\Lambda} - \gamma \Dh_{\Lambda, \Lambda}^{-1} \gv.
	\end{equation}
	
	Next, we want to check that \( \bvt \) is a global minimizer. To do so, let us compare the objective function at a point  \( \bvo = \bvt + \delta \ev_{j} \) for arbitrary index \( j \notin \Lambda \). Since \( \| \bvo \|_{1} = \| \bvt \|_{1} + |\delta| \), we have
	\[
	\begin{aligned}
	L(\bvt) - L(\bvo) &= \frac{1}{2} \bvt^{\T} \Dh \bvt - \frac{1}{2} \bvo^{\T} \Dh \bvo - \cvh^{\T}(\bvt - \bvo) - \gamma | \delta | \\
	& = \frac{\delta^{2}}{2} \ev_{j}^{\T} \Dh \ev_{j} + |\delta| \gamma -  \delta \ev_{j}^{\T} \Dh \bvt + \delta \widehat{c}_{j} \\
	& > |\delta| \gamma -  \delta \ev_{j}^{\T} \Dh \bvt + \delta \widehat{c}_{j},
	\end{aligned}
	\]
	where the latter comes from the fact that \( \Dh \) is positively definite. Applying the equality \eqref{bvt:subdiff} yields,
	\begin{equation*}
	\ev_{j}^{\T} \Dh \bvt =  \Dh_{j, \Lambda} \Dh_{\Lambda, \Lambda}^{-1} \cvh_{\Lambda} -  \gamma \Dh_{j, \Lambda} \Dh_{\Lambda, \Lambda}^{-1} \gv,
	\end{equation*}
	therefore, taking into account \( \| \gv \|_{\infty} \leq 1 \) we have,
	\[
	L(\bvt) - L(\bvo) > |\delta| \left[\gamma (1 - \| \Dh_{\Lambda^{c}, \Lambda} \Dh_{\Lambda, \Lambda}^{-1}\|_{1, \infty}) -  \bigl| \Dh_{j, \Lambda} \Dh_{\Lambda, \Lambda}^{-1} \cvh_{\Lambda} - \widehat{c}_{j} \bigr|  \right],
	\]
	where the right-hand side is nonnegative by the condition of the lemma. Since \( j \notin \Lambda \) is arbitrary, \( \bvt \) is a global solution as well.
	
\end{proof}

\begin{remark}
	It is not hard to see that in the exact case \( \Dh = \Phi^{\T} \Phi \) and \( \cvh = \Phi^{\T} \yv \) the condition of the lemma above turns into the condition \(  \| \Phi_{\Lambda^{c}}^{\T} P_{\Lambda} \epsv \|_{\infty} \leq \gamma \mathrm{ERC}(\Lambda) \) of Theorem\tilspace\ref{tropp_lemma}.
\end{remark}

Since we are particularly interested in applications to time series, the features matrix \( \Phi \) should in fact be random, thus stating a ERC-like condition onto it might result in additional unnecessary technical difficulties. Instead, let us assume that there is some other matrix \( \Db \), potentially the expectation of \( \Phi^{\T} \Phi \), such that it is close enough to \( \Dh \) (with some probability, but we are stating all the results deterministically in this section), and the value that controls the exact recovery looks like
\[
\mathrm{ERC}(\Lambda; \Db) = 1 - \| \Db_{\Lambda^c, \Lambda} \Db_{\Lambda, \Lambda}^{-1} \|_{1, \infty} .
\]
Additionally, we set \( \cvb = \Db \bv^{*} = \Db_{\cdot, \Lambda} \bv^{*}_{\Lambda} \) --- the vector that \( \cvh \) is intended to approximate. Note that in this case we have \( \Db_{\Lambda^{c}, \Lambda} \Db_{\Lambda, \Lambda}^{-1} \cvb_{\Lambda} - \cvb_{\Lambda^{c}} = \Db_{\Lambda^{c}, \Lambda} \bv^{*}_\Lambda - \cvb_{\Lambda^{c}} = 0 \), thus the conditions of Lemma~\ref{exact_lasso:lemma} hold for \( \Db, \cvb \) once \( \mathrm{ERC}(\Lambda; \Db) \) and \( \gamma \) are nonnegative. In what follows, we control the values appearing in the lemma for \( \Dh \) and \( \cvh \) through the differences between \( \cvb \), \( \Db \) and \( \cvh \), \( \Dh \), respectively, thus allowing the exact recovery of the sparsity pattern.
Lemma~\ref{exact:lem}

\begin{corollary}\label{lasso_exact:cor}
	Let \( \Db \) and \( \cvb \) be such that \( \cvb = \Db \bv^{*} \). Assume that 
	\begin{align*}
	\| \cvh - \cvb \|_{\infty} \leq \delta_{c},
	\qquad
	\| \Db_{\Lambda, \Lambda}^{-1}(\cvh_{\Lambda} - \cvb_{\Lambda}) \|_{\infty} &\leq \delta_{c}',
	\qquad
	\| \Db^{-1}_{\Lambda, \Lambda}(\Dh_{\Lambda, \cdot} - \Db_{\Lambda, \cdot}) \|_{\infty, \infty} \leq \delta_{D} , \\
	\| (\Dh_{\cdot, \Lambda} - \Db_{\cdot, \Lambda}) \bv^{*}_{\Lambda} \|_{\infty} & \leq \delta_{D}' ,
	\qquad
	\| \Db_{\Lambda, \Lambda}^{-1}(\Db_{\Lambda, \Lambda} -  \Dh_{\Lambda, \Lambda})\bv^{*}_{\Lambda} \|_{\infty} \leq \delta_{D}'' .
	\end{align*}
	Suppose, \( \mathrm{ERC}(\Lambda) \geq 3/4 \) and
	\[
	3 \delta_c + 3 \delta_{D}' \leq \gamma,
	\qquad  {s} \delta_{D} \leq \frac{1}{16}	,
	\]
	where \( |\Lambda| = s \).
	Then, the solution to \eqref{l1pen_missing} is supported on a subset of \( \Lambda \) and satisfies
	\begin{equation}\label{bv_Lambda_subset}
	\tilde{\bv}_{\Lambda} = \Dh_{\Lambda, \Lambda}^{-1} \cvh_{\Lambda} - \gamma \Dh_{\Lambda, \Lambda}^{-1} \gv,
	\end{equation}
	with some \( \gv \in \R^{s} \) satisfying \( \| \gv_{\Lambda} \|_{\infty} \leq 1 \) and the max-norm error satisfies
	\[
	\| \tilde{\bv} - \bv^{*} \|_{\infty} \leq 2 (\delta_{D}'' + \delta_{c}' + \gamma \| \Db_{\Lambda, \Lambda}^{-1} \|_{1, \infty}) ,
	\]
	while the \(\ell_{2}\)-norm error satisfies
	\[
	\| \tilde{\bv} - \bv^{*} \| \leq 2 \sqrt{s} (\delta_{D}'' + \delta_{c}' + \gamma \sigma_{\min}^{-1}) .
	\]
	
	If additionally \( 2 (\delta_{D}'' + \delta_{c}' + \gamma \| \Db_{\Lambda, \Lambda}^{-1} \|_{1, \infty}) \leq \min_{j \in \Lambda} |\bv_j^*|, \) then we have the exact recovery, so that the following equality takes place
	\[
	\tilde{\bv}_{\Lambda} = \Dh_{\Lambda, \Lambda}^{-1} \cvh_{\lambda} - \gamma \Dh_{\Lambda, \Lambda}^{-1} \sv_{\Lambda},
	\]
	where \( \sv = \sign(\bv^*) \).
\end{corollary}
\begin{proof}
	First, observe that \(  D_{\Lambda^{c}, \Lambda} D_{\Lambda, \Lambda}^{-1} \cv_{\Lambda} - \cv_{\Lambda^{c}} = \Phi_{\Lambda^{c}}^{\T}(\Phi_{\Lambda}^{+} \yv -  \yv) = \Phi_{\Lambda^{c}}^{\T} (P_{\Lambda} - I) \epsv  \).
	By Lemma~\ref{easy_bounds} we have,
	\begin{align*}
	\| \Dh_{\Lambda^{c}, \Lambda} \Dh_{\Lambda, \Lambda}^{-1}\|_{1, \infty} \leq \| \Db_{\Lambda^{c}, \Lambda} \Db_{\Lambda, \Lambda}^{-1}\|_{1, \infty} + 4 s \delta_{D} \leq 1/2,
	\end{align*}
	while since \( \cvb_{\Lambda^{c}} = \Db_{\Lambda^{c}, \Lambda} \bv^{*}_{\Lambda} = \Db_{\Lambda^{c}, \Lambda} \Db_{\Lambda, \Lambda}^{-1} \cvb_{\Lambda} \),
	\begin{align*}
	\| \Dh_{\Lambda^{c}, \Lambda} \Dh_{\Lambda, \Lambda}^{-1} \cvh_{\Lambda} - \cvh_{\Lambda^{c}} \|_{\infty} & 
	\leq
	\| \Dh_{\Lambda^{c}, \Lambda} \Dh_{\Lambda, \Lambda}^{-1} \cvh_{\Lambda} - \Db_{\Lambda^{c}, \Lambda} \Db_{\Lambda, \Lambda}^{-1} \cvb_{\Lambda} \|_{\infty} + \| \cvh_{\Lambda^{c}} - \cvb_{\Lambda^{c}} \|_{\infty} \\
	& \leq
	\| \Dh_{\Lambda^{c}, \Lambda} \Dh_{\Lambda, \Lambda}^{-1} (\cvh_{\Lambda} - \cvb_{\Lambda}) \|_{\infty} + \| \Dh_{\Lambda^{c}, \Lambda} (\Dh_{\Lambda, \Lambda}^{-1} - \Db_{\Lambda, \Lambda}^{-1}) \cvb_{\Lambda} \|_{\infty} \\
	& \phantom{\leq}\, + \| (\Dh_{\Lambda^{c}, \Lambda} - \Db_{\Lambda^{c}, \Lambda}) \Db_{\Lambda, \Lambda}^{-1} \cvb_{\Lambda} \|_{\infty} + \delta_{c} \\ 
	& \leq
	\| \Dh_{\Lambda^{c}, \Lambda} \Dh_{\Lambda, \Lambda}^{-1} (\cvh_{\Lambda} - \cvb_{\Lambda}) \|_{\infty} + \| \Dh_{\Lambda^{c}, \Lambda} (\Dh_{\Lambda, \Lambda}^{-1} - \Db_{\Lambda, \Lambda}^{-1}) \cvb_{\Lambda} \|_{\infty} + \delta_{D}' + \delta_{c} .
	\end{align*}
	Here, \( \| \Dh_{\Lambda^{c}, \Lambda} \Dh_{\Lambda, \Lambda}^{-1} (\cvh_{\Lambda} - \cvb_{\Lambda}) \|_{\infty} \leq \delta_c / 2 \) due to \( \| \Dh_{\Lambda^{c}, \Lambda} \Dh_{\Lambda, \Lambda}^{-1}\|_{1, \infty}  \leq 1/2 \). Moreover, we have
	\begin{align*}
	\| \Dh_{\Lambda^{c}, \Lambda} (\Dh_{\Lambda, \Lambda}^{-1} - \Db_{\Lambda, \Lambda}^{-1}) \cvb_{\Lambda} \|_{\infty} &= \| \Dh_{\Lambda^{c}, \Lambda} \Dh_{\Lambda, \Lambda}^{-1} (\Db_{\Lambda, \Lambda} - \Dh_{\Lambda, \Lambda})\Db_{\Lambda, \Lambda}^{-1} \cvb_{\Lambda} \|_{\infty} \\
	& \leq 
	\| \Dh_{\Lambda^{c}, \Lambda} \Dh_{\Lambda, \Lambda}^{-1} \|_{1, \infty} \| (\Db_{\Lambda, \Lambda} - \Dh_{\Lambda, \Lambda})\Db_{\Lambda, \Lambda}^{-1} \cvb_{\Lambda} \|_{\infty} \\
	& \leq \delta_{D}' /2 .
	\end{align*}
	Using the condition on \( \gamma \), we get that
	\[
	\| \Dh_{\Lambda^{c}, \Lambda} \Dh_{\Lambda, \Lambda}^{-1} \cvh_{\Lambda} - \cvh_{\Lambda^{c}} \|_{\infty} \leq \frac{3}{2} (\delta_{D}' + \delta_{c}) \leq \frac{\gamma}{2} \leq \gamma (1 - \| \Dh_{\Lambda^{c}, \Lambda} \Dh_{\Lambda, \Lambda}^{-1}\|_{1, \infty}) ,
	\]
	so that the conditions of Lemma~\ref{exact_lasso:lemma} are satisfied and \eqref{bv_Lambda_subset} takes place. Therefore, we can write
	\begin{align*}
	\tilde{\bv}_{\Lambda} - \bv^{*}_{\Lambda} & = \Dh_{\Lambda, \Lambda}^{-1} \cvh_{\Lambda} - \Db_{\Lambda, \Lambda}^{-1} \cvb_{\Lambda} - \gamma \Dh_{\Lambda, \Lambda}^{-1} \gv , \\
	& = \Dh_{\Lambda, \Lambda}^{-1}(\Db_{\Lambda, \Lambda} -  \Dh_{\Lambda, \Lambda})\Db_{\Lambda, \Lambda}^{-1} \cvb_{\Lambda} + \Dh_{\Lambda, \Lambda}^{-1} (\cvh_{\Lambda} - \cvb_{\Lambda}) - \gamma \Dh_{\Lambda, \Lambda}^{-1} \gv \\
	& = \Dh_{\Lambda, \Lambda}^{-1}(\Db_{\Lambda, \Lambda} -  \Dh_{\Lambda, \Lambda})\bv^{*}_{\Lambda} + \Dh_{\Lambda, \Lambda}^{-1} (\cvh_{\Lambda} - \cvb_{\Lambda}) - \gamma \Dh_{\Lambda, \Lambda}^{-1} \gv \\
	& = \Dh_{\Lambda, \Lambda}^{-1} \Db_{\Lambda, \Lambda} \left( \Db_{\Lambda, \Lambda}^{-1}(\Db_{\Lambda, \Lambda} -  \Dh_{\Lambda, \Lambda})\bv^{*}_{\Lambda} + \Db_{\Lambda, \Lambda}^{-1} (\cvh_{\Lambda} - \cvb_{\Lambda}) - \gamma \Db_{\Lambda, \Lambda}^{-1} \gv \right)
	\end{align*}
	By Lemma~\ref{easy_bounds} we have \( \| \Dh_{\Lambda, \Lambda}^{-1} \Db_{\Lambda, \Lambda} \|_{\infty \mapsto \infty} \leq 2 \) so that
	\[
	\| \tilde{\bv}_{\Lambda} - \bv^{*}_{\Lambda}  \|_{\infty} \leq 2 \| \Db_{\Lambda, \Lambda}^{-1}(\Db_{\Lambda, \Lambda} -  \Dh_{\Lambda, \Lambda})\bv^{*}_{\Lambda} \|_{\infty} + 2 \| \Db_{\Lambda, \Lambda}^{-1}(\cvh_{\Lambda} - \cvb_{\Lambda}) \|_{\infty} + 2 \gamma \| \Db_{\Lambda, \Lambda}^{-1} \|_{1, \infty}  \, .
	\]
	Since we also have \( \trinorm \Dh_{\Lambda, \Lambda}^{-1} \Db_{\Lambda, \Lambda} \trinorm_{\op} \leq 2 \) and \( \| \gv \| \leq \sqrt{s} \), it holds
	\[
	\| \tilde{\bv}_{\Lambda} - \bv^{*}_{\Lambda}  \| \leq 2 \sqrt{s} \left( \| \Db_{\Lambda, \Lambda}^{-1}(\Db_{\Lambda, \Lambda} -  \Dh_{\Lambda, \Lambda})\bv^{*}_{\Lambda} \|_{\infty} + \| \Db_{\Lambda, \Lambda}^{-1}(\cvh_{\Lambda} - \cvb_{\Lambda}) \|_{\infty} + \gamma \trinorm \Db_{\Lambda, \Lambda}^{-1} \trinorm_{\op} \right).
	\]
\end{proof}

Before we proceed with the proof of this corollary, we present a technical lemma that collects some trivial inequalities. 

\begin{lemma}\label{easy_bounds}
	Set \( \delta_{c} = \| \cvh - \cvb \|_{\infty} \), \( \delta_{D} = \| (\Dh_{\Lambda^{c}, \Lambda} - \Db_{\Lambda^{c}, \Lambda}) \Db_{\Lambda, \Lambda}^{-1} \|_{\infty, \infty} \). Suppose, \( \| \Db_{\Lambda^c \Lambda} \Db_{\Lambda \Lambda}^{-1} \|_{1, \infty} \leq 1 \) and \( {s}\delta_{D} \leq 1/2 \). It holds,
	\begin{itemize}
		\item for any \( q \geq 1 \)
		\[
		\| D_{\Lambda, \Lambda} \Dh_{\Lambda, \Lambda}^{-1} \|_{q \rightarrow q} \leq 2,
		\qquad
		\| \Dh_{\Lambda, \Lambda}^{-1} D_{\Lambda, \Lambda} \|_{q \rightarrow q} \leq 2 \, ;
		\]
		\item
		\[
		\| \Dh_{\Lambda^{c}, \Lambda} \Dh_{\Lambda, \Lambda}^{-1} - D_{\Lambda^{c}, \Lambda} D_{\Lambda, \Lambda}^{-1} \|_{1, \infty}
		\leq 
		4 s \delta_{D} .
		\]
	\end{itemize}
\end{lemma}
\begin{proof}
	First, we have
	\begin{align*}
	\| D_{\Lambda, \Lambda} \Dh_{\Lambda, \Lambda}^{-1} \|_{q \rightarrow q} & = \| I + (D_{\Lambda, \Lambda} - \Dh_{\Lambda, \Lambda}) \Dh_{\Lambda, \Lambda}^{-1} \|_{q \rightarrow q} \\
	&\leq 1 + \| (D_{\Lambda, \Lambda} - \Dh_{\Lambda, \Lambda}) D_{\Lambda, \Lambda}^{-1} \|_{q \rightarrow q} \| D_{\Lambda, \Lambda} \Dh_{\Lambda, \Lambda}^{-1} \|_{q \rightarrow q} \\
	& \leq 1 + s \delta_{D} \| D_{\Lambda, \Lambda} \Dh_{\Lambda, \Lambda}^{-1} \|_{q \rightarrow q} ,
	\end{align*}
	which solving the inequality and since \( s \delta_{D} \leq 1/2 \), turns into 
	\[
	\| D_{\Lambda, \Lambda} \Dh_{\Lambda, \Lambda}^{-1} \|_{q \rightarrow q} \leq \frac{1}{1 - s \delta_{D}} \leq 2 .
	\]
	Similarly, \( \| \Dh_{\Lambda, \Lambda}^{-1} D_{\Lambda, \Lambda}  \|_{q \rightarrow q} \leq 2 \).
	
	Furthermore,
	\begin{align*}
	\| (\Dh_{\Lambda^{c}, \Lambda} - D_{\Lambda^{c}, \Lambda}) \Dh_{\Lambda, \Lambda}^{-1} \|_{1, \infty} & \leq \| (\Dh_{\Lambda^{c}, \Lambda} - D_{\Lambda^{c}, \Lambda}) D_{\Lambda, \Lambda}^{-1}\|_{1, \infty} \| D_{\Lambda, \Lambda} \Dh_{\Lambda, \Lambda}^{-1} \|_{1 \rightarrow 1} \\
	& \leq 2 s \delta_{D} .
	\end{align*}
	and
	\[
	\begin{aligned}
	\| D_{\Lambda^c, \Lambda} (D_{\Lambda, \Lambda}^{-1} - \Dh_{\Lambda, \Lambda}^{-1} ) \|_{1, \infty} \leq & \| D_{\Lambda, \Lambda^{c}} D_{\Lambda, \Lambda}^{-1} \|_{1, \infty} \| \Dh_{\Lambda, \Lambda}^{-1} (\Dh_{\Lambda, \Lambda} - D_{\Lambda, \Lambda})  \|_{1 \rightarrow 1}
	\\
	\leq & \| D_{\Lambda, \Lambda^{c}} D_{\Lambda, \Lambda}^{-1} \|_{1, \infty} \| \Dh_{\Lambda, \Lambda}^{-1} D_{\Lambda, \Lambda} \|_{1 \rightarrow 1} \| D_{\Lambda, \Lambda}^{-1}(\Dh - D) \|_{1 \rightarrow 1} \\
	\leq & 2 \| D_{\Lambda, \Lambda^{c}} D_{\Lambda, \Lambda}^{-1} \|_{1, \infty} s \delta_{D} ,
	\end{aligned}
	\]
	which together give us the second inequality. 
\end{proof}

\end{document}